%% file: arxiv.tex
\documentclass[a4paper,UKenglish]{llncs}

\newif\ifdebug
\debugfalse 

\input{macro}

\graphicspath{{./graphics/}}

\title{Proving Linearizability Using Partial Orders\\
(Extended Version)}

\titlerunning{Proving Linearizability Using Partial Orders}
\author{Artem Khyzha\inst{1} \and
Mike Dodds\inst{2} \and
Alexey Gotsman\inst{1} \and
Matthew Parkinson\inst{3}}

\institute{IMDEA Software Institute, Madrid, Spain
\and University of York, UK
\and Microsoft Research Cambridge, UK}

\newcommand{\ag}[1]{
	\ifdebug
	{\color{red}{\bf AG:  #1}}
	\fi
}

\newcommand{\artem}[1]{
	\ifdebug
	{\bf\color{red} Artem:} {\it #1}
	\fi
}

\DeclareTextFontCommand{\textbfit}{
  \fontseries\bfdefault 
  \itshape
}
\makeatletter\def\@listI{\leftmargin\leftmargini\parsep 1\p@ \@plus1\p@
\@minus\p@\topsep 1\p@ \@plus2\p@ \@minus0\p@\itemsep0\p@}\let\@listi
\@listI\@listi\makeatother

\begin{document}

\let\oldaddcontentsline\addcontentsline
\def\addcontentsline#1#2#3{}
\maketitle

\input{abstract}
\input{intro}
\input{lin}
\input{running}
\input{informal}
\input{proglang}

\input{logic}
\input{details}\FloatBarrier
\input{set}\FloatBarrier
\input{related}

\bibliographystyle{abbrv}
\bibliography{references}

\input{appendix}
\end{document}

%% file: macro.tex
\usepackage{amsmath}
\usepackage{amssymb} 
\usepackage{graphicx}
\usepackage{extarrows}
\usepackage{listings}

\usepackage{dsfont}
\usepackage{bbm}

\usepackage[frame,all]{xy}
\usepackage{etoolbox}
\usepackage{mathrsfs}
\usepackage{mathtools}
\usepackage{placeins}
\usepackage{stmaryrd}
\usepackage{wrapfig}
\usepackage{xspace}
\usepackage{setspace}
\usepackage{yfonts}
\usepackage{float}

\usepackage{cancel}
\usepackage{color}
\usepackage{enumerate}
\usepackage{enumitem}

\newcommand{\newparagraph}[1]{\vspace{5pt}\noindent\textbf{#1}}

\DeclarePairedDelimiter{\db}{\llbracket}{\rrbracket}

\newcommand{\eidType}{{\sf EventID}}
\newcommand{\eventType}{{\sf Event}}
\newcommand{\ghostType}{{\sf Ghost}}
\newcommand{\historyType}{{\sf History}}
\newcommand{\opType}{{\sf Op}}
\newcommand{\stateType}{{\sf State}}
\newcommand{\tidType}{{\sf ThreadID}}
\newcommand{\valueType}{{\sf Val}}

\newcommand{\tsType}{{\sf TS}}

\newcommand{\abs}{{\sf abs}}

\newcommand{\history}{H} 

\newcommand{\event}{e}
\newcommand{\ghost}{G}
\newcommand{\state}{s}

\newcommand{\evalf}[2]{\db{{#1}}_{#2}}
\newcommand{\lint}{\ell}

\newcommand{\isSeq}{{\sf seq}}
\newcommand{\seq}{\Sigma}
\newcommand{\args}{a}
\newcommand{\deqOp}{{\rm Deq}}
\newcommand{\eid}{i}
\newcommand{\eidp}{j}
\newcommand{\enqOp}{{\rm Enq}}
\newcommand{\op}{{\sf op}}
\newcommand{\res}{r}
\newcommand{\tid}{t}

\newcommand{\tidof}[1]{{#1}.{\sf tid}}
\newcommand{\opof}[1]{{#1}.{\sf op}}
\newcommand{\argsof}[1]{{#1}.{\sf arg}}
\newcommand{\resof}[1]{{#1}.{\sf rval}}
\newcommand{\getEvent}{{\tt enqOf}}

\newcommand{\events}{E}
\newcommand{\order}{R}

\newcommand{\rely}{\mathcal{R}}
\newcommand{\guar}{\mathcal{G}}

\newcommand{\complete}[1]{\left \lfloor {#1} \right \rfloor}
\newcommand{\incomplete}[1]{{#1}\setminus \complete{#1}}
\newcommand{\dom}[1]{{\sf dom}(#1)}
\newcommand{\powerset}[1]{\mathcal{P}({#1})}
\newcommand{\sub}[2]{[{#1}\,{:}\,{#2}]}

\newcommand{\inv}{{\sf INV}}
\newcommand{\invlin}{\inv_{\sf LIN}}
\newcommand{\invwf}{\inv_{\sf WF}}

\newcommand{\invord}{\inv_{\sf ORD}}
\newcommand{\invalg}{\inv_{\sf ALG}}
\newcommand{\linv}{{\sf LI}}

\newcommand{\Done}{\None}

\newcommand{\None}{\bot}
\newcommand{\NULL}{{\rm NULL}}
\newcommand{\Todo}{{\sf todo}}
\newcommand{\FALSE}{{\sf false}}
\newcommand{\TRUE}{{\sf true}}

\newcommand{\refines}[2]{{#1} \sqsubseteq {#2}}
\newcommand{\linearizes}[2]{{#1} \sqsubseteq {#2}}
\newcommand{\completes}[2]{{#1} \unlhd {#2}}

\newcommand{\histories}{\mathcal{H}}

\newcommand{\qhist}{\histories_{\sf queue}}

\newcommand{\pool}[1]{{\tt pools}({#1})}
\newcommand{\pid}{p} 

\newcommand{\edge}[3]{{#1} \xrightarrow{#2} {#3}}
\newcommand{\nedge}[3]{\neg ({#1} \xrightarrow{#2} {#3})}

\newcommand{\tsless}{\mathrel{{} <_{\sf TS} {}}}
\newcommand{\untaken}[1]{{\sf untaken}({#1})}
\newcommand{\taken}[1]{{\sf taken}({#1})}
\newcommand{\inqueue}[1]{{\sf inQ}({#1})}

\newcommand{\started}{{\rm started}}
\newcommand{\finished}{{\rm ended}}

\newcommand{\larg}{{\tt arg}}
\newcommand{\lres}{{\tt res}}

\newcommand{\lstassert}[1]{\color{blue}
    \ensuremath{\left\lbrace \begin{array}{l}#1\end{array}\right\rbrace}
}
\definecolor{gray}{rgb}{0.8,0.8,0.8}
\definecolor{numcol}{rgb}{0.5,0.5,0.5}


%% file: abstract.tex

\begin{abstract}

Linearizability is the commonly accepted notion of correctness for concurrent
data structures. It requires that any execution of the data structure is
justified by a linearization --- a linear order on operations satisfying
the data structure's sequential specification. Proving linearizability is often
challenging because an operation's position in the linearization order may
depend on future operations. This makes it very difficult to incrementally
construct the linearization in a proof. 

We propose a new proof method that can handle data structures with such
future-dependent linearizations. Our key idea is to incrementally construct not
a single linear order of operations, but a partial order that describes multiple
linearizations satisfying the sequential specification. This allows decisions
about the ordering of operations to be delayed, mirroring the behaviour of data
structure implementations. We formalise our method as a program logic based on
rely-guarantee reasoning, and demonstrate its effectiveness by verifying several
challenging data structures: the Herlihy-Wing queue, the TS queue and
the Optimistic set.

\end{abstract}


%% file: intro.tex

\section{Introduction}
\label{sec:intro}

Linearizability is a commonly accepted notion of correctness of concurrent data
structures. It matters for programmers using such data structures because it
implies {\em contextual refinement}: any behaviour of a program using a
concurrent data structure can be reproduced if the program uses its sequential
implementation where all operations are executed
atomically~\cite{filipovic-tcs}. This allows the programmer to soundly reason
about the behaviour of the program assuming a simple sequential specification of
the data structure.
\begin{wrapfigure}[8]{r}{.49\textwidth}
\vspace{-2.1em}
\hspace{-0.5em}
\includegraphics[scale=0.4,trim=1 13 1 1,clip]{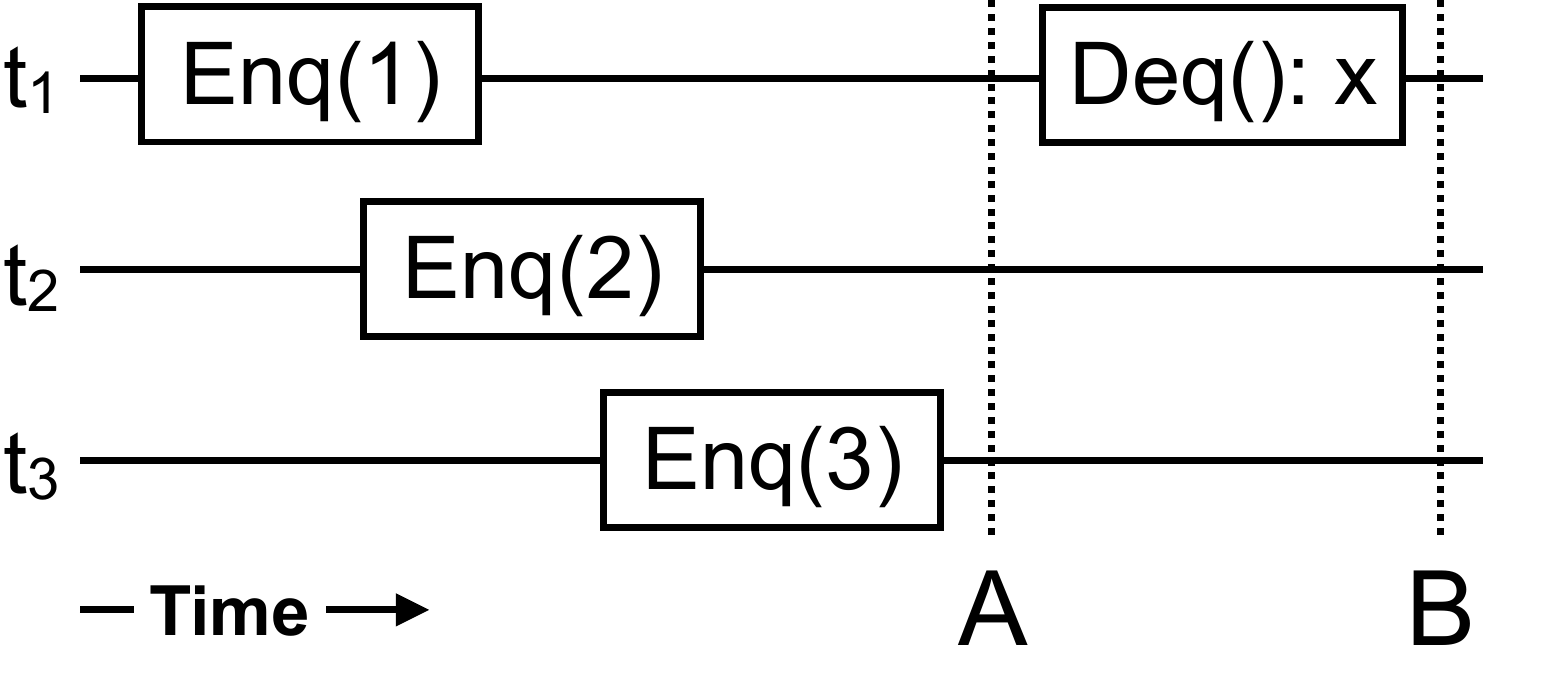}
\vspace{-2em}
\caption{Example execution.
\label{fig:tsqtrace}}
\end{wrapfigure}
Linearizability requires that for any execution of operations on the data
structure there exists a linear order of these operations, called a {\em
  linearization}, such that: {\em (i)} the linearization respects the order of
non-overlapping operations (the {\em real-time order}); and {\em (ii)} the
behaviour of operations in the linearization matches the sequential
specification of the data structure. To illustrate this, consider an execution
in Figure~\ref{fig:tsqtrace}, where three threads are accessing a queue.
Linearizability determines which values $x$ the dequeue operation is allowed to
return by considering the possible linearizations of this execution. Given {\em
  (i)}, we know that in any linearization the enqueues must be ordered before
the dequeue, and Enq(1) must be ordered before Enq(3). Given {\em (ii)}, a
linearization must satisfy the sequential specification of a queue, so the
dequeue must return the oldest enqueued value. Hence, the execution in
Figure~\ref{fig:tsqtrace} has three possible linearizations: [Enq(1); Enq(2);
Enq(3); Deq():1], [Enq(1); Enq(3); Enq(2); Deq():1] and [Enq(2); Enq(1); Enq(3);
Deq():2]. This means that the dequeue is allowed to return 1 or 2, but not 3.

For a large class of algorithms, linearizability can be proved by incrementally
constructing a linearization as the program executes. Effectively, one shows
that the program execution and its linearization stay in correspondence under
each program step (this is formally known as a {\em forward simulation}). The
point in the execution of an operation at which it is appended to the
linearization is called its \emph{linearization point}. This must occur
somewhere between the start and end of the operation, to ensure that the
linearization preserves the real-time order. For example, when applying the
linearization point method to the execution in Figure~\ref{fig:tsqtrace}, by
point (A) we must have decided if Enq(1) occurs before or after Enq(2) in the
linearization. Thus, by this point, we know which of the three possible
linearizations matches the execution. This method of establishing
linearizability is very popular, to the extent that most papers proposing new
concurrent data structures include a placement of linearization points. However,
there are algorithms that cannot be proved linerizable using the linearization
point method.

In this paper we consider several examples of such algorithms, including the
{\em time-stamped (TS) queue}~\cite{haas-thesis,dodds-popl15}---a recent
high-performance data structure with an extremely subtle correctness
argument. Its key idea is for enqueues to attach timestamps to values, and for
these to determine the order in which values are dequeued. As illustrated by the
above analysis of Figure~\ref{fig:tsqtrace}, linearizability allows concurrent
operations, such as Enq(1) and Enq(2), to take effect in any order. The TS queue
exploits this by allowing values from concurrent enqueues to receive
incomparable timestamps; only pairs of timestamps for non-overlapping enqueue
operations must be ordered. Hence, a dequeue can potentially have a choice of
the ``earliest'' enqueue to take values from. This allows concurrent dequeues to
go after different values, thus reducing contention and improving performance.

The linearization point method simply does not apply to the TS queue. In the
execution in Figure~\ref{fig:tsqtrace}, values 1 and 2 could receive
incomparable timestamps. Thus, at point (A) we do not know which of them will be
dequeued first and, hence, in which order their enqueues should go in the
linearization: this is only determined by the behaviour of dequeues later in the
execution. Similar challenges exist for other queue algorithms such as the
baskets queue~\cite{DBLP:conf/opodis/HoffmanSS07}, LCR queue~\cite{lcrq} and
Herlihy-Wing queue~\cite{linearizability}. In all of these algorithms, when an
enqueue operation returns, the precise linearization of earlier enqueue
operations is not necessarily known. Similar challenges arise in the
time-stamped stack~\cite{dodds-popl15} algorithm. We conjecture that our proof
technique can be applied to prove the time-stamped stack linearizable, and we
are currently working on a proof.




In this paper, we propose a new proof method that can handle algorithms where
incremental construction of linearizations is not possible. We formalise it as a
program logic, based on Rely-Guarantee~\cite{rg}, and apply it to give simple
proofs to the TS queue~\cite{dodds-popl15}, the Herlihy-Wing
queue~\cite{linearizability} and the Optimistic Set~\cite{hindsight}. The key
idea of our method is to incrementally construct not a single linearization of
an algorithm execution, but an \emph{abstract history}---a partially ordered
history of operations such that it contains the real-time order of the original
execution and {\em all} its linearizations satisfy the sequential specification.
By embracing partiality, we enable decisions about order to be delayed,
mirroring the behaviour of the algorithms. At the same time, we maintain the
simple inductive style of the standard linearization-point method: the proof of
linearizability of an algorithm establishes a simulation between its execution
and a growing abstract history. By analogy with linearization points, we call
the points in the execution where the abstract history is extended
\emph{commitment points}.
\begin{wrapfigure}[22]{r}{.37\textwidth}
\vspace{-2em}
\hspace{-0.8em}
\includegraphics[scale=0.35]{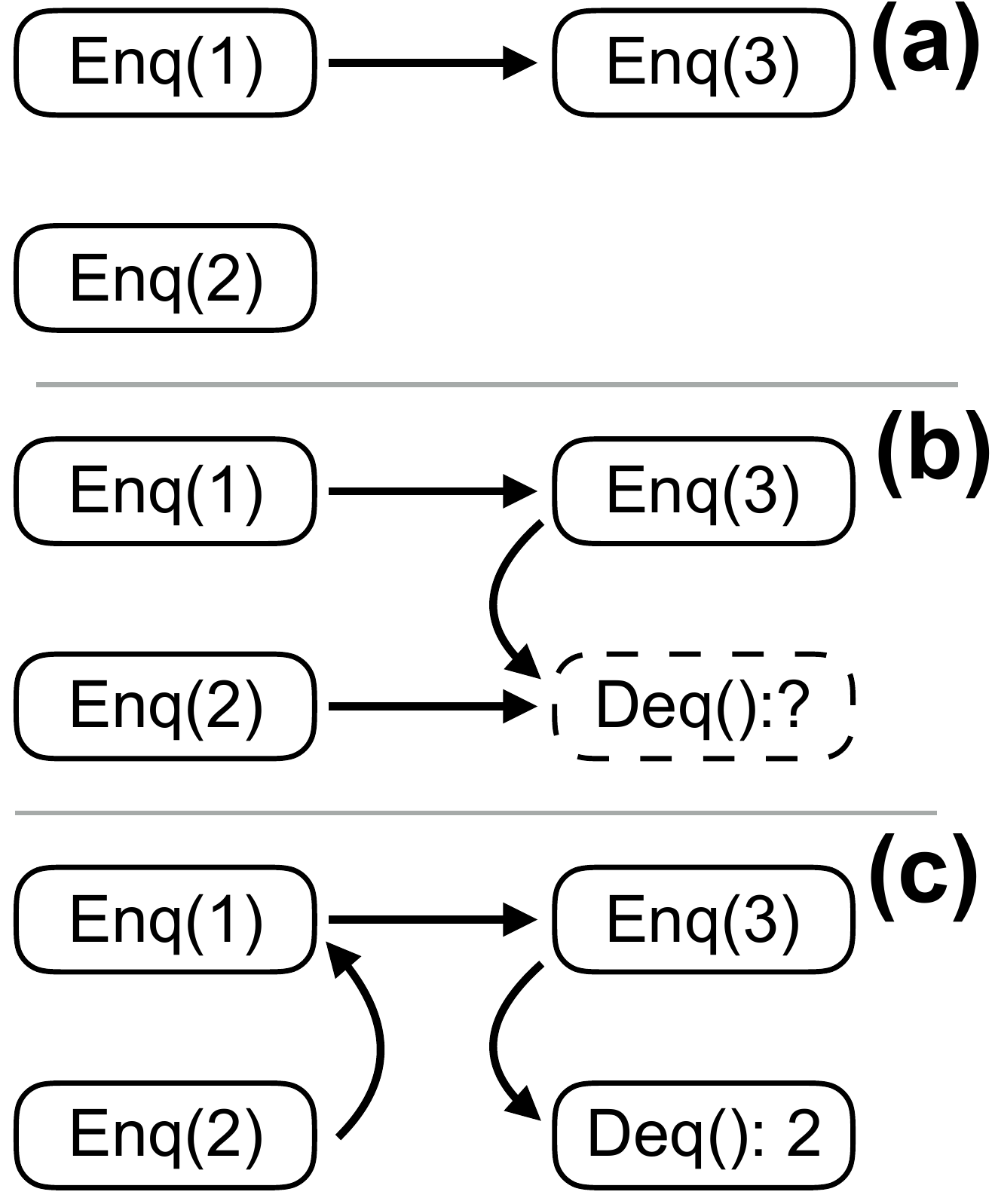}

\vspace{-0.5em}
\caption{Abstract histories constructed for prefixes of the execution in
  Figure~\ref{fig:tsqtrace}: (a) is at point (A); (b) is at the start of the
  dequeue operation; and (c) is at point (B). We omit the transitive
  consequences of the edges shown.
  \label{fig:partialhistory}}
\end{wrapfigure}
The extension can be done in several ways: (1) committing to perform an
operation; (2) committing to an order between previously unordered operatons;
(3) completing an operation.

Consider again the TS queue execution in Figure~\ref{fig:tsqtrace}. By point (A)
we construct the abstract history in Figure~\ref{fig:partialhistory}(a). The
edge in the figure is mandated by the real-time order in the original execution;
Enq(1) and Enq(2) are left unordered, and so are Enq(2) and Enq(3). At the start
of the execution of the dequeue, we update the history to the one in
Figure~\ref{fig:partialhistory}(b). A dashed ellipse represents an operation
that is not yet completed, but we have committed to performing it (case 1
above). When the dequeue successfully removes a value, e.g., 2, we update the
history to the one in Figure~\ref{fig:partialhistory}(c). To this end, we
complete the dequeue by recording its result (case 3). We also commit to an
order between the Enq(1) and Enq(2) operations (case 2). This is needed to
ensure that all linearizations of the resulting history satisfy the sequential
queue specification, which requires a dequeue to remove the oldest value in the
queue.

We demonstrate the simplicity of our method by giving proofs to challenging
algorithms that match the intuition for why they work.
Our method is also similar in spirit to the standard linearization point method.
Thus, even though in this paper we formulate the method as a program logic, we
believe that algorithm designers can also benefit from it in informal reasoning,
using abstract histories and commitment points instead of single linearizations
and linearization points.


%% file: lin.tex

\newcommand{\retType}{{\sf RetVals}}

\newcommand{\ids}[1]{{\sf id}({#1})}

\section{Linearizability, Abstract Histories and Commitment Points}
\label{sec:linearizability}

\newparagraph{Preliminaries.}
We consider a data structure that can be accessed concurrently via operations
$\op \in \opType$ in several threads, identified by $\tid \in \tidType$. Each
operation takes one argument and returns one value, both from a set
$\valueType$; we use a special value $\bot \in \valueType$ to model operations
that take no argument or return no value. Linearizability relates the observable
behaviour of an implementation of such a concurrent data structure to its
sequential specification~\cite{linearizability}. We formalise both of these by
sets of {\em histories}, which
are partially ordered sets of {\em events}, recording operations invoked on the
data structure. Formally, an event is of the form $\event = \sub{\eid}{(\tid,
\op, \args, \res)}$. It includes a unique identifier $\eid \in \eidType$ and
records an operation $\op \in \opType$ called by a thread $\tid \in \tidType$
with an argument $\args \in \valueType$, which returns a value $\res \in
\valueType \uplus \{ \Todo \}$. We use the special return value $\Todo$ for
events describing operations that have not yet terminated, and call such events
{\em uncompleted}. We denote the set of all events by $\eventType$. Given a set
$\events \subseteq \eventType$, we write $\events(\eid) = (\tid, \op, \args,
\res)$ if $\sub{\eid}{(\tid, \op, \args, \res)} \in \events$ and let
$\complete{\events}$ consist of all completed events from $\events$. We let
$\ids{\events}$ denote the set of all identifiers of events from $\events$.
Given an event identifier $\eid$, we also use $\tidof{\events(\eid)}$,
$\opof{\events(\eid)}$, $\argsof{\events(\eid)}$ and $\resof{\events(\eid)}$ to
refer to the corresponding components of the tuple $\events(\eid)$.
\begin{definition}\label{def:history}
  A {\em history}\footnote{
    For technical convenience, our notion of a history is different from the one
    in the classical linearizability definition~\cite{linearizability}, which
    uses separate events to denote the start and the end of an operation. We
    require that $\order$ be an interval order, we ensure that our notion is
    consistent with an interpretation of events as segments of time during which
    the corresponding operations are executed, with $\order$ ordering $\eid_1$
    before $\eid_2$ if $\eid_1$ finishes before $\eid_2$
    starts~\cite{interval-order}.
}
  is a pair $\history = (\events, \order)$, where $\events \subseteq \eventType$
  is a finite set of events with distinct identifiers and $\order \subseteq
  \ids{\events} \times \ids{\events}$ is a strict partial order (i.e.,
  transitive and irreflexive), called the {\em real-time order}. We require that
  for each $\tid \in \tidType$:
\begin{itemize}
\item events in $\tid$ are totally ordered by $\order$:
\\ $\forall \eid, \eidp \in \ids{\events} \ldotp \eid
  \neq \eidp \land \tidof{\events(\eid)} = \tidof{\events(\eidp)} = \tid
  \implies
  (\edge{\eid}{\order}{\eidp} \lor \edge{\eidp}{\order}{\eid})$;
\item only maximal events in $\order$ can be uncompleted:\\ 
$\forall \eid \,{\in}\,
  \ids{\events} \ldotp \forall \tid \,{\in}\, \tidType \ldotp
  \resof{\events(\eid)} = \Todo \implies
  \neg\exists \eidp \in \ids{\events} \ldotp \edge{\eid}{\order}{\eidp}$;
\item $\order$ is an interval order:\\
$\forall \eid_1, \eid_2, \eid_3, \eid_4 \ldotp \edge{\eid_1}{\order}{\eid_2} \land \edge{\eid_3}{\order}{\eid_4} \implies \edge{\eid_1}{\order}{\eid_4} \lor \edge{\eid_2}{\order}{\eid_3}$.
\end{itemize}
We let $\historyType$ be the set of all histories.
A history $(\events, \order)$ is {\em sequential}, written $\isSeq(\events,
\order)$, if $\ids{\events} = \complete{\events}$ and $\order$ is total on $\events$.
\end{definition}

Informally, $\edge{\eid}{\order}{\eidp}$ means that the operation recorded by
$E(\eid)$ completed before the one recorded by $E(\eidp)$ started. The
real-time order in histories produced by concurrent data structure
implementations may be partial, since in this case the execution of operations
may overlap in time; in contrast, specifications are defined using sequential
histories, where the real-time order is total.



\newparagraph{Linearizability.} Assume we are given a set of histories that can
be produced by a given data structure implementation (we introduce a programming
language for implementations and formally define the set of histories an
implementation produces in \S\ref{sec:lang}). Linearizability requires all
of these histories to be matched by a similar history of the data structure
specification (its {\em linearization}) that, in particular, preserves the
real-time order between events in the following sense: the real-time order of a
history $\history = (\events, \order)$ is {\em preserved} in a history
$\history' = (\events', \order')$, written $\refines{\history}{\history'}$, if
$\events = \events'$ and $\order \subseteq \order'$.

The full definition of linearizability is slightly more complicated due to the
need to handle uncompleted events: since operations they denote have not
terminated, we do not know whether they have made a change to the data
structure or not. To account for this, the definition makes all events in the
implementation history complete by discarding some uncompleted events and
completing the remaining ones with an arbitrary return value. Formally, an
event $\event = \sub{\eid}{(\tid, \op, \args, \res)}$ {\em can be completed}
to an event $\event' = \sub{\eid'}{(\tid', \op', \args', \res')}$, written
$\completes{e}{e'}$, if $\eid = \eid'$, $\tid = \tid'$, $\op = \op'$, $\args =
\args'$ and either $\res = \res' \neq \Todo$ or $\res' = \Todo$. A history
$\history = (\events, \order)$ {\em can be completed} to a history $\history' =
(\events', \order')$, written $\completes{\history}{\history'}$, if
$\ids{\events'} \subseteq \ids{\events}$, $\complete{\events} \subseteq
\complete{\events'}$, $\order \cap (\ids{\events'}
\times \ids{\events'}) = \order'$ and $\forall \eid \in \ids{\events'} \ldotp
\completes{\sub{\eid}{\events(\eid)}}{\sub{\eid}{\events'(\eid)}}$.
\begin{definition}\label{dfn:linearizability}
  A set of histories $\histories_1$ (defining the data structure
  implementation) is {\em linearized} by a set of sequential histories
  $\histories_2$ (defining its specification), written $\histories_1
  \sqsubseteq \histories_2$, if $\forall \history_1 \in \histories_1.\,
  \exists \history_2 \in \histories_2.\, \exists \history'_1.\,
  \completes{\history_1}{\history'_1} \wedge
  \refines{\history'_1}{\history_2}$.
\end{definition}


Let $\qhist$ be the set of sequential histories defining the behaviour of a
queue with $\opType = \{\enqOp, \deqOp\}$. Due to space constraints, we provide
its formal definition in the extended version of this paper \cite{ext}, but for
example, [Enq(2); Enq(1); Enq(3); Deq():2] $\in \qhist$ and [Enq(1); Enq(2);
Enq(3); Deq():2] $\not\in \qhist$.



\newparagraph{Proof method.} In general, a history of a data structure ($H_1$
in Definition~\ref{dfn:linearizability}) may have multiple linearizations
($H_2$) satisfying a given specification $\histories$. In our proof method, we
use this observation and construct a partially ordered history, an {\em abstract
history}, all linearizations of which belong to $\histories$.
\begin{definition}\label{dfn:abstracthistory}
  A history $\history$ is an {\em abstract history} of a specification given by
  the set of sequential histories $\histories$ if $\{\history' \mid \lfloor
  \history \rfloor \sqsubseteq \history' \land \isSeq(\history')\} \subseteq
  \histories$, where $\complete{(\events, \order)} = (\complete{\events}, \order
  \cap (\ids{\complete{\events}} \times \ids{\complete{\events}}))$. We denote
  this by $\abs(\history, \histories)$.
\end{definition}


We define the construction of an abstract history $\history = (E, R)$ by
instrumenting the data structure operations with auxiliary code that updates the
history at certain {\em commitment points} during operation execution. There are
three kinds of commitment points:
\begin{enumerate}[leftmargin=12pt]
\item When an operation $\op$ with an argument $a$ starts executing in a thread
$t$, we extend $E$ by a fresh event $[i: (t, \op, a, \Todo)]$, which we order in
$R$ after all events in $\complete{E}$.
\item At any time, we can add more edges to $R$.
\item By the time an operation finishes, we have to assign its return value to
its event in $E$.
\end{enumerate}

Note that, unlike Definition~\ref{dfn:linearizability},
Definition~\ref{dfn:abstracthistory} uses a particular way of completing an
abstract history $\history$, which just discards all uncompleted events using
$\lfloor - \rfloor$. This does not limit generality because, when constructing
an abstract history, we can complete an event (item 3) right after the
corresponding operation makes a change to the data structure, without waiting
for the operation to finish.

In \S\ref{sec:logic} we formalise our proof method as a program logic and show
that it indeed establishes linearizability. Before this, we demonstrate
informally how the obligations of our proof method are discharged on an example.





%% file: running.tex

\newcommand{\myEid}{{\tt myEid}()}
\newcommand{\freshEid}{{\tt freshEid}()}
\newcommand{\myTid}{{\tt myTid}()}
\newcommand{\ghostrf}{\ghost_{\sf rf}}
\newcommand{\ghostts}{\ghost_{\sf ts}}

\newcommand{\ttcandpid}{{\tt cand\_pid}}
\newcommand{\ttcandtid}{{\tt cand\_tid}}
\newcommand{\ttcandts}{{\tt cand\_ts}}
\newcommand{\ttstartts}{{\tt start\_ts}}

\newcommand{\ttpid}{{\tt pid}}
\newcommand{\ttts}{{\tt ts}}
\newcommand{\tttid}{{\tt tid}}

\newcommand{\ttcand}{{\tt CAND}}
\newcommand{\ttpools}{{\tt pools}}

\section{Running Example: the Time-Stamped Queue}\label{sec:running}
\begin{figure}[t]
\begin{minipage}[t]{.48\textwidth}
\begin{lstlisting}[numbers=none]
PoolID insert(ThreadID t, Val v) {
  p := new PoolID();
  #$\pool{\tt t}$ := $\pool{\tt t} \cdot ({\tt p}, {\tt v}, \top)$#;
  return p;
}

Val remove(ThreadID t, PoolID p) {
  if (#$\exists \seq, \seq', {\tt v}, \tau \ldotp {}$#
      #$\pool{\tt t}$ = $\seq \cdot ({\tt p}, {\tt v}, \tau) \cdot \seq'$#) {
    #$\pool{\tt t}$ := $\seq \cdot \seq'$#;
    return v;
  } else return NULL;
}
\end{lstlisting}
\end{minipage}
\hfill
\begin{minipage}[t]{.48\textwidth}
\begin{lstlisting}[numbers=none]
#$({\tt PoolID} \times {\tt TS})$# getOldest(ThreadID t) {
  if (#$\exists {\tt p}, \tau \ldotp \pool{\tt t}$ = $({\tt p}, \_, \tau) \cdot \_$#)
    return #$({\tt p}, \tau)$#;
  else
    return #$({\tt NULL}, {\tt NULL})$#;
}

setTimestamp(ThreadID t,
                  PoolID p, TS #$\tau$#) {
  if (#$\exists \seq, \seq', {\tt v} \ldotp {}$#
      #$\pool{\tt t}$ = $\seq \cdot ({\tt p}, {\tt v}, \_) \cdot \seq'$#)
    #$\pool{\tt t}$# := #$\seq \cdot ({\tt p}, {\tt v}, \tau) \cdot \seq'$#;
}
\end{lstlisting}
\end{minipage}
\caption{Operations on abstract SP pools $\ttpools : \tidType \to {\sf
  Pool}$. All operations are atomic.
}
\label{fig:pool}
\end{figure}

We use the TS queue~\cite{haas-thesis} as our running example. Values in the
queue are stored in per-thread single-producer (SP) multi-consumer pools, and we
begin by describing this auxiliary data structure.


\newparagraph{SP pools.} SP pools have well-known linearizable
implementations~\cite{haas-thesis}, so
we simplify our presentation by using abstract pools with the atomic operations
given in Figure~\ref{fig:pool}.
This does not limit generality: since
linerarizability implies contextual refinement (\S\ref{sec:intro}),

\begin{wrapfigure}[16]{r}{.515\textwidth}
\vspace{0pt}
\begin{lstlisting}
enqueue(Val v) {
  atomic {
    PoolID node := insert(#$\myTid$#, v);$\label{line:tsq_insert}$
\end{lstlisting}
\begin{lstlisting}[firstnumber=4,backgroundcolor=\color{gray}]      
    #$\ghostts[\myEid] := \top;\label{line:tsq_ghostts1}$#
\end{lstlisting}
\begin{lstlisting}[firstnumber=5]         
  }
  TS timestamp := newTimestamp(); $\label{line:tsq_newts}$
  atomic {
    setTimestamp(#$\myTid$#, node, timestamp); $\label{line:tsq_setts}$
\end{lstlisting}
\begin{lstlisting}[firstnumber=9,backgroundcolor=\color{gray}]      
    #$\ghostts[\myEid] := {\tt timestamp};$# $\label{line:tsq_ghostts2}$
    #$\resof{\events(\myEid)} := \Done;$#
\end{lstlisting}
\begin{lstlisting}[firstnumber=11]      
  }
  return $\Done$;
}
\end{lstlisting}
\caption{
  The TS queue: enqueue. Shaded portions are auxiliary code used in the proof.
}
\label{fig:tsq_enqueue} 
\end{wrapfigure}
\noindent
properties
proved using the abstract
pools will stay valid for their linearizable
implementations. In the figure and in the following we denote irrelevant
expressions by~$\_$.


The SP pool of a thread contains a sequence of triples $(p, v, \tau)$, each
consisting of a unique identifier $\pid \in {\sf PoolID}$, a value $v \in
\valueType$ enqueued into the TS queue by the thread and the associated
timestamp $\tau \in \tsType$. The set of timestamps $\tsType$ is
partially ordered by $\tsless$, with a distinguished timestamp $\top$ that is
greater than all others. We let ${\tt pool}$ be the set of states of an
abstract SP pool. Initially all pools are empty.
The operations on SP pools are as follows:
\begin{itemize} 
\item {\tt insert(t,v)} appends a value {\tt v} to the back of the pool of
  thread {\tt t} and associates it with the special timestamp $\top$; it returns
  an identifier for the added element.
\item {\tt setTimestamp(t,p,$\tau$)} sets to $\tau$ the timestamp of the element
  identified by {\tt p} in the pool of thread {\tt t}.
\item {\tt getOldest(t)} returns the identifier and timestamp of the value from
  the front of the pool of thread {\tt t}, or $(\NULL, \NULL)$ if the pool is
  empty.
\item {\tt remove(t,p)} tries to remove a value identified by {\tt p} from the
  pool of thread {\tt t}. Note this can fail if some other thread removes the
  value first.
\end{itemize}
Separating {\tt insert} from {\tt setTimestamp} and {\tt getOldest} from {\tt
  remove} in the SP pool interface reduces the atomicity granularity, and
permits more efficient implementations.


\begin{figure}[t]
\hfill
\begin{minipage}[t]{0.99\textwidth}
\begin{lstlisting}[firstnumber=14]
Val dequeue() {
  Val ret := NULL;
\end{lstlisting}
\begin{lstlisting}[firstnumber=16,backgroundcolor=\color{gray}]
  #$\eidType$ $\ttcand$;#
\end{lstlisting}
\begin{lstlisting}[firstnumber=17]
  do {
    TS #\ttstartts# := newTimestamp(); $\label{line:tsq_startts}$
    PoolID pid, #\ttcandpid# := NULL;
    TS ts, #\ttcandts# := $\top$;
    ThreadID #\ttcandtid#;
    for each k in 1..NThreads do { $\label{line:tsq_loop_begin}$
      atomic {
        (pid, ts) := getOldest(k); $\label{line:tsq_getoldest}$
\end{lstlisting}
\begin{lstlisting}[firstnumber=25,backgroundcolor=\color{gray}]  
        #$\order$ := ($\order \cup \{ (e, \myEid) \mid
        e \in \ids{\complete{\events}} \cap
          \inqueue{\ttpools, \events, \ghostts}$ \label{line:tsq_scan}# 
                                      #${} \land
        \lnot(\ttstartts \tsless \ghostts(e))\})^+;$ \label{line:tsq_scan_end}#
\end{lstlisting}
\begin{lstlisting}[firstnumber=26]           
      }
      if (pid ${\neq}$ NULL && ${\tt ts} \tsless \ttcandts$ && #$\lnot$#(#$\ttstartts \tsless {\tt ts}$#)) { $\label{line:tsq_compare}$
        #(\ttcandpid, \ttcandts, \ttcandtid)# := (pid, ts, k);
\end{lstlisting}
\begin{lstlisting}[firstnumber=31,backgroundcolor=\color{gray}]        
        #$\ttcand$ := \getEvent($\events$, $\ghostts$, \ttcandtid, \ttcandts);# $\label{line:tsq_getevent}$
\end{lstlisting}
\begin{lstlisting}[firstnumber=32]        
      }
    }  $\label{line:tsq_loop_end}$
    if (#\ttcandpid# #$\neq$# NULL) $\label{line:tsq_checknull}$
      atomic { $\label{line:tsq_remove}$
        ret := remove(#\ttcandtid#, #\ttcandpid#);
\end{lstlisting}
\begin{lstlisting}[firstnumber=37,backgroundcolor=\color{gray}]          
        if (#${\tt ret} \neq \NULL$#) {
          #$\resof{\events(\myEid)}$ := ret;#
          #$\order$ := ($\order \cup
%            \{ (\ttcand, \myEid) \}$#
%                  #${}\cup
                  \{ (\ttcand, e) \mid
                    e \in \inqueue{\ttpools, \events, \ghostts} \}$#
                  #${}\cup
                  \{ (\myEid, d) \mid
                            \opof{\events(d)} = \deqOp \land
                              d \in \ids{\incomplete{\events}} \}
                  )^+;$#
        }
\end{lstlisting}
\begin{lstlisting}[firstnumber=43]        
      } $\label{line:tsq_remove_ends}$
  } while (ret = NULL);
  return ret;
}
\end{lstlisting}
\end{minipage}
\caption{The TS queue: dequeue. Shaded portions are auxiliary code used in the
proof.
}
\label{fig:tsq_dequeue} 
\end{figure}

\newparagraph{Core TS queue algorithm.}  Figures~\ref{fig:tsq_enqueue}
and~\ref{fig:tsq_dequeue} give the code for our version of the TS queue. Shaded
portions are auxiliary code needed in the linearizability proof to update the
abstract history at commitment points; it can be ignored for now. In the overall
TS queue, enqueuing means adding a value with a certain timestamp to the pool of
the current thread, while dequeuing means searching for the value with the
minimal timestamp across per-thread pools and removing it.



In more detail, the ${\tt enqueue}(v)$ operation first inserts the value $v$
into the pool of the current thread, defined by {\tt myTid}
(line~\ref{line:tsq_insert}). At this point the value $v$ has the default,
maximal timestamp $\top$. The code then generates a new timestamp using {\tt
  newTimestamp} and sets the timestamp of the new value to it
(lines~\ref{line:tsq_newts}-\ref{line:tsq_setts}). We describe an implementation
of {\tt newTimestamp} later in this section. The key property that it ensures is
that out of two non-overlapping calls to this function, the latter returns a
higher timestamp than the former; only concurrent calls may generate
incomparable timestamps. Hence, timestamps in each pool appear in the ascending
order.





The ${\tt dequeue}$ operation first generates a timestamp $\ttstartts$ at
line~\ref{line:tsq_startts}, which it further uses to determine a consistent
snapshot of the data structure. After generating $\ttstartts$, the operation
iterates through per-thread pools, searching for a value with a minimal
timestamp (lines \ref{line:tsq_loop_begin}--\ref{line:tsq_loop_end}). The search
starts from a random pool, to make different threads more likely to pick
different elements for removal and thus reduce contention. The pool identifier
of the current candidate for removal is stored in $\ttcandpid$, its timestamp in
$\ttcandts$ and the thread that inserted it in $\ttcandtid$. On each iteration
of the loop, the code fetches the earliest value enqueued by thread ${\tt k}$
(line \ref{line:tsq_getoldest}) and checks whether its timestamp is smaller than
the current candidate's $\ttcandts$ (line~\ref{line:tsq_compare}). If the
timestamps are incomparable, the algorithm keeps the first one (either would be
legitimate). Additionally, the algorithm never chooses a value as a candidate if
its timestamp is greater than $\ttstartts$, because such values are not
guaranteed to be read in a consistent manner.

If a candidate has been chosen once the iteration has completed, the code tries
to remove it (line~\ref{line:tsq_remove}). This may fail if some other thread
got there first, in which case the operation restarts. Likewise, the algorithm
restarts if no candidate was identified (the full algorithm
in~\cite{haas-thesis} includes an emptiness check, which we omit for
simplicity).




\newparagraph{Timestamp generation.}
The TS queue requires that sequential calls to {\tt newTimestamp} generate
ordered timestamps. This ensures that the two sequentially enqueued values
cannot be dequeued out of order. However, concurrent calls to {\tt newTimestamp}
may generate incomparable timestamps. This is desirable because it increases
flexibility in choosing which value to dequeue, reducing contention.


\begin{wrapfigure}[12]{r}{0.42\textwidth}
\vspace{-22.5pt}
\hfill
\begin{minipage}{0.395\textwidth}
\begin{lstlisting}[firstnumber=37]
int counter = 1;

TS newTimestamp() {
  int ts = counter;$\label{line:read_count}$
  TS result;
  if (CAS(counter, ts, ts+1))$\label{line:ts_CAS}$ 
    result = (ts, ts); $\label{line:ts_eq}$
  else
    result = (ts, counter-1);$\label{line:ts_new}$
  return result;
}
\end{lstlisting}
\end{minipage}
\vspace{-5pt}
\caption{
Timestamp generation algorithm.
}
\label{fig:ts-CAS} 
\end{wrapfigure}

There are a number of implementations of {\tt newTimestamp} satisfying the above
requirements~\cite{dodds-popl15}. For concreteness, we consider the
implementation given in Figure~\ref{fig:ts-CAS}. Here a timestamp is either
$\top$ or a pair of integers $(s,e)$, representing a time interval. In every
timestamp $(s, e)$, $s \leq e$. Two timestamps 
are considered ordered $(s_1, e_1) \tsless (s_2, e_2)$ if $e_1 < s_2$, i.e., if
the time intervals do not overlap. Intervals are generated with the help of a
shared ${\tt counter}$. The algorithm reads the counter as the start of the
interval and attempts to atomically increment it with a CAS
(lines~\ref{line:read_count}-\ref{line:ts_CAS}), which is a well-known atomic
compare-and-swap operation. It atomically reads the counter and, if it still
contains the previously read value {\tt ts}, updates it with the new
timestamp ${\tt ts}+1$ and returns $\TRUE$; otherwise, it does nothing
and returns $\FALSE$. If CAS succeeds, then the algorithm takes the interval
start and end values as equal (line~\ref{line:ts_eq}). If not, some other
thread(s) increased the counter. The algorithm reads the counter again and
subtracts 1 to give the end of the interval (line~\ref{line:ts_new}). Thus,
either the current call to {\tt newTimestamp} increases the counter, or some
other thread does so. In either case, subsequent calls will generate timestamps
greater than the current one.

This timestamping algorithm allows concurrent enqueue operations in
Figure~\ref{fig:tsqtrace} to get incomparable timestamps. Then the dequeue may
remove either $1$ or $2$ depending on where it starts traversing the
pools\footnote{Recall that the randomness is required to reduce contention}
(line~\ref{line:tsq_loop_begin}). As we explained in
\S\ref{sec:intro}, this makes the standard method of linearization point
inapplicable for verifying the TS queue.


%% file: informal.tex

\section{The TS Queue: Informal Development}\label{sec:informal}

In this section we explain how the abstract history is updated at the commitment
points of the TS Queue and justify informally why these updates preserve the key
property of this history---that all its linearizations satisfy the sequential
queue specification. We present the details of the proof of the TS queue in
\S\ref{sec:details}.

\newparagraph{Ghost state and auxiliary definitions.}  To aid in constructing
the abstract history $(\events, \order)$, we instrument the code of the
algorithm to maintain a piece of ghost state---a partial function $\ghostts :
\eidType \rightharpoonup \tsType$. Given the identifier $\eid$ of an event
$\events(\eid)$ denoting an ${\tt enqueue}$ that has inserted its value into a
pool, $\ghostts(\eid)$ gives the timestamp currently associated with the value.
The statements in lines \ref{line:tsq_ghostts1} and \ref{line:tsq_ghostts2} in
Figure~\ref{fig:tsq_enqueue} update $\ghostts$ accordingly. These statements use
a special command $\myEid$ that returns the identifier of the event associated
with the current operation.

As explained in \S\ref{sec:running}, the timestamps of values in each pool
appear in strictly ascending order. As a consequence, all timestamps assigned by
$\ghostts$ to events of a given thread $\tid$ are distinct, which is formalised
by the following property:
\[
\forall \eid, \eidp \ldotp \eid \neq \eidp \land \tidof{\events(\eid)} =
\tidof{\events(\eidp)} \land \eid, \eidp \in \dom{\ghostts}
\implies \ghostts(\eid) \neq \ghostts(\eidp) 
\]
Hence, for a given thread $\tid$ and a timestamp $\tau$, there is at most one
enqueue event in $\events$ that inserted a value with the timestamp $\tau$ in
the pool of a thread $\tid$. In the following, we denote the identifier of this
event by $\getEvent(\events, \ghostts, \tid, \tau)$ and let the set of the
identifiers of such events for all values currently in the pools be
$\inqueue{{\tt pools},
\events, \ghostts}$:
\[
\inqueue{{\tt pools}, \events, \ghostts}
\triangleq
\{ \getEvent(\events, \ghostts, \tid, \tau) \mid \exists \pid \ldotp
{\tt pools}(\tid) = \_ \cdot (\pid, \_, \tau) \cdot \_
\}
\]

\newparagraph{Commitment points and history updates.} We further instrument the
code with statements that update the abstract history at commitment points,
which we now explain. As a running example, we use the execution in
Figure~\ref{fig:tsqtrace2}, extending that in Figure~\ref{fig:tsqtrace}. As we
noted in \S\ref{sec:linearizability}, when an operations starts, we
automatically add a new uncompleted event to $E$ to represent this operation and
order it after all completed events
\begin{wrapfigure}[9]{r}{.48\textwidth}
\vspace{-20pt}
\hspace{-0.5em}
\includegraphics[scale=0.4,trim=1 13 1 1,clip]{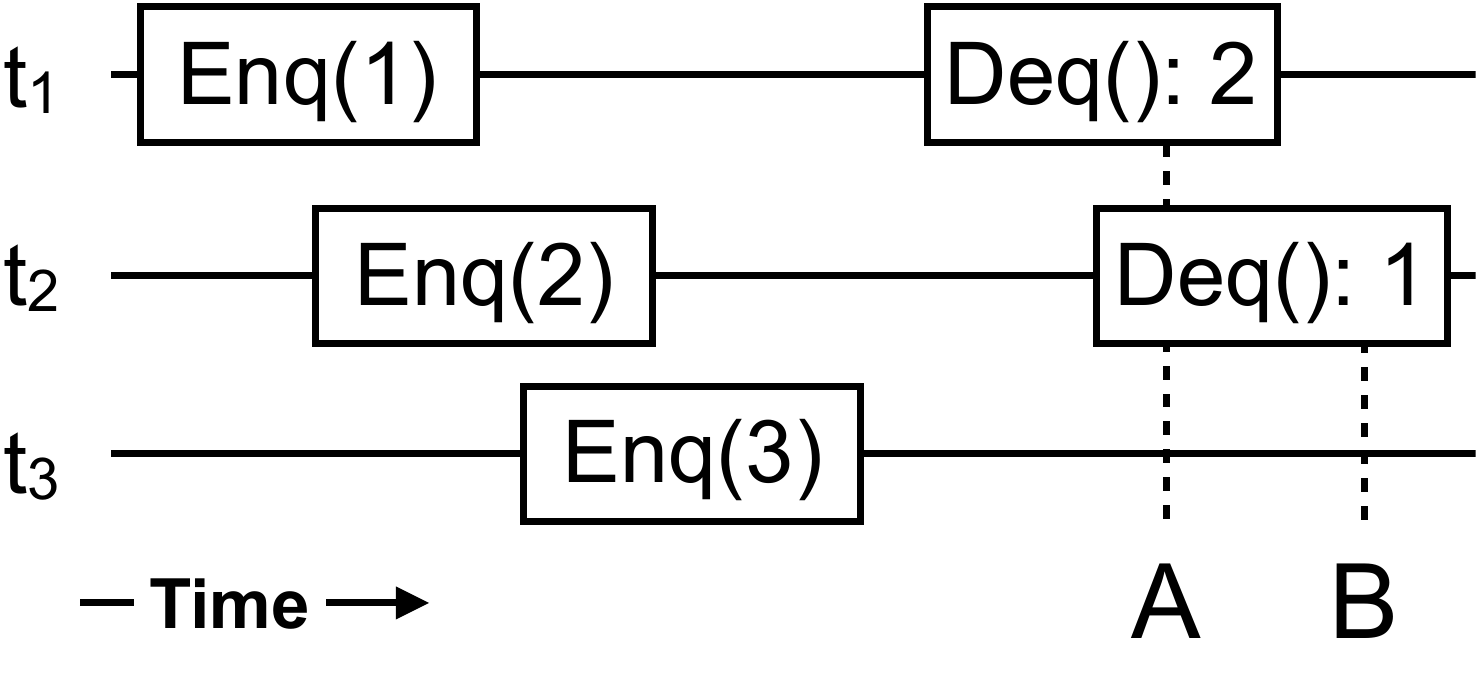}

\vspace{-10pt}
\caption{Example execution extending Figure~\ref{fig:tsqtrace}. Dotted lines indicate commitment points at lines~\ref{line:tsq_remove}--\ref{line:tsq_remove_ends} of the dequeues.
\label{fig:tsqtrace2}}
\end{wrapfigure}
\noindent
 in $\order$. For example, before the start of
\textrm{Enq(3)} in the execution of Figure~\ref{fig:tsqtrace2}, the abstract
history contains two events \textrm{Enq(1)} and \textrm{Enq(2)} and no edges in
the real-time order. At the start of \textrm{Enq(3)} the history gets
transformed to that in Figure~\ref{fig:tsqcommit}(a). The commitment point at
line~\ref{line:tsq_setts} in Figure~\ref{fig:tsq_enqueue} completes the enqueue
by giving it a return value $\Done$, which results in the abstract history in
Figure~\ref{fig:tsqcommit}(b).

\begin{figure}[t]
\includegraphics[width=\textwidth,trim=1 1 1 1,clip]{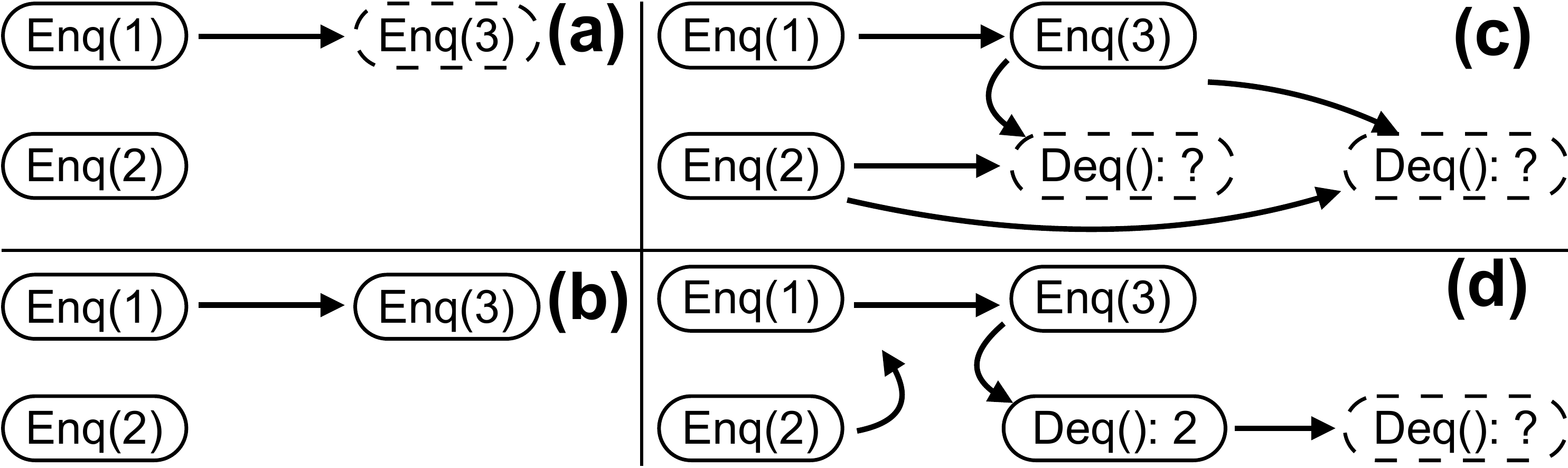}
\caption{Changes to the abstract history of the execution in
Figure~\ref{fig:tsqtrace2}.
\label{fig:tsqcommit}}
\end{figure}



Upon a dequeue's start, we similarly add an event representing it. Thus, by
point (A) in Figure~\ref{fig:tsqtrace2}, the abstract history is as shown in
Figure~\ref{fig:tsqcommit}(c). At every iteration ${\tt k}$ of the loop, the
dequeue performs a commitment point at
lines~\ref{line:tsq_scan}--\ref{line:tsq_scan_end}, where we order enqueue
events of values currently present in the pool of a thread ${\tt k}$ before the
current dequeue event. Specifically, we add an edge $(e, \myEid)$ for each
identifier $e$ of an enqueue event whose value is in the ${\tt k}$'s pool and
whose timestamp is not greater than the dequeue's own timestamp $\ttstartts$.
Such ordering ensures that in all linearizations of the abstract history, the
values that the current dequeue observes in the pool according to the algorithm
are also enqueued in the sequential queue prior to the dequeue. In particular,
this also ensures that in all linearizations, the dequeue returns a value that
has already been inserted.

The key commitment point in dequeue occurs in
lines~\ref{line:tsq_remove}--\ref{line:tsq_remove_ends}, where the abstract
history is updated if the dequeue successfully removes a value from a pool. The
ghost code at line~\ref{line:tsq_getevent} stores the event identifier for the
enqueue that inserted this value in $\ttcand$. At the commitment point we first
complete the current dequeue event by assigning the value removed from a pool as
its return value. This ensures that the dequeue returns the same value in the
concrete execution and the abstract history. Finally, we order events in the
abstract history to ensure that all linearizations of the abstract history
satisfy the sequential queue specification. To this end, we add the following
edges to $\order$ and then transitively close it:
\begin{enumerate}
\item $(\ttcand, e)$ for each identifier $e$ of an enqueue event whose value
  is still in the pools. This ensures that the dequeue removes the oldest value
  in the queue.
\item $(\myEid, d)$ for each identifier $d$ of an uncompleted dequeue event.
  This ensures that dequeues occur in the same order as they remove values from
  the queue.
\end{enumerate}
At the commitment point (A) in Figure~\ref{fig:tsqtrace2} the abstract history
gets transformed from the one in Figure~\ref{fig:tsqcommit}(c) to the one in
Figure~\ref{fig:tsqcommit}(d).


%% file: proglang.tex

\newcommand{\seqc}{\mathbin{;}}
\newcommand{\com}{C}
\newcommand{\cskip}{{\sf skip}}

\newcommand{\pcom}{\alpha}
\newcommand{\pcomType}{{\sf PCom}}

\newcommand{\intp}[3]{\llbracket {#2} \rrbracket_{#1} ({#3})}

\newcommand{\comType}{{\sf Com}}
\newcommand{\sttrans}[5]{\langle {#1}, {#2} \rangle
\mathrel{{\longrightarrow}_{#3}} \langle {#4}, {#5} \rangle}

\newcommand{\astate}{\sigma}

\newcommand{\myfrac}[2]{\genfrac{}{}{0.5pt}{}{\displaystyle #1}{\displaystyle #2}}

\newcommand{\cdt}{D}
\newcommand{\adt}{\mathbb{D}}

\newcommand{\tp}{c}
\newcommand{\tpType}{{\sf Cont}}
\newcommand{\idle}{{\sf idle}}

\newcommand{\hsem}[1]{\mathcal{H}(#1)}

\newcommand{\htrans}[3]{
\langle {#1} \rangle
\mathrel{{\twoheadrightarrow}_{#2}}
\langle {#3} \rangle
}

\newcommand{\mhtrans}[3]{
\langle {#1} \rangle
\mathrel{{\twoheadrightarrow}^*_{#2}}
\langle {#3} \rangle
}

\newcommand{\locType}{{\sf Loc}}

\section{Programming Language}
\label{sec:lang}

To formalise our proof method, we first introduce a programming language for
data structure implementations. This defines such implementations by functions
$\cdt : \opType \to \comType$ mapping operations to {\em commands} from a set
$\comType$. The commands, ranged over by $C$, are written in a simple
while-language, which includes {\em atomic} commands $\alpha$ from a set
$\pcomType$ (assignment, CAS, etc.) and standard control-flow constructs. To
conserve space, we describe the precise syntax in the extended version of this paper \cite{ext}.

Let $\locType \subseteq \valueType$ be the set of all memory locations. We let
$\stateType = \locType \to \valueType$ be the set of all states of the data
structure implementation, ranged over by $\state$. Recall from
\S\ref{sec:linearizability} that operations of a data structure can be called
concurrently in multiple threads from $\tidType$. For every thread $\tid$, we
use distinguished locations $\larg[\tid], \lres[\tid] \in \locType$ to store an
argument, respectively, the return value of an operation called in this thread.

We assume the semantics of each atomic command $\pcom \in \pcomType$ given by a
non-deterministic state transformers $\db{\pcom}_\tid : \stateType \to
\powerset{\stateType}$, $\tid \in \tidType$. For a state $\state$,
$\intp{\tid}{\pcom}{\state}$ is the set of states resulting from thread
$\tid$ executing $\pcom$ atomically in $\state$. We then lift this semantics to
a sequential small-step operational semantics of arbitrary commands from
$\comType$: $\sttrans{\com}{\state}{\tid}{\com'}{\state'}$. Again, we omit the
standard rules of the semantics; see \cite{ext}.

We now define the set of histories produced by a data structure implementation
$\cdt$, which is required by the definition of linearizability
(Definition~\ref{dfn:linearizability}, \S\ref{sec:linearizability}). Informally,
these are the histories produced by threads repeatedly invoking data structure
operations in any order and with any possible arguments (this can be thought of
as running the data structure implementation under its {\em most general
  client}~\cite{linown}). We define this formally using a concurrent small-step
semantics of the data structure $D$ that also constructs corresponding
histories: ${\twoheadrightarrow}_{\cdt} \subseteq (\tpType \times \stateType
\times \historyType)^2$, where $\tpType = \tidType \to (\comType \uplus
\{\idle\})$. Here a function $\tp \in \tpType$ characterises the progress of an
operation execution in each thread $\tid$: $\tp(\tid)$ gives the continuation of
the code of the operation executing in thread $\tid$, or $\idle$ if no operation
is executing. The relation ${\twoheadrightarrow}_{\cdt}$ defines how a step of
an operation in some thread transforms the data structure state and the history:
\[
\begin{array}{@{}c@{}}
\myfrac{
  \eid \notin \ids{\events}
  \quad \args \in \valueType
    \quad \events' = \events \sub{\eid}{(\tid, \op, \args, \Todo)}
    \quad \order' = \order \cup \{ (\eidp, \eid) \mid \eidp \in \complete{\events} \}
}{
  \htrans{\tp \sub{\tid}{\idle}, \state, (\events, \order)}
         {\cdt}
         {\tp \sub{\tid}{\cdt(\op)}, \state \sub{\larg[\tid]}{\args}, (\events', \order')}
}
\\[10pt]
\myfrac{
  \sttrans{\com}{\state}{\tid}{\com'}{\state'}
}{
  \htrans{\tp \sub{\tid}{\com}, \state, (\events, \order)}
         {\cdt} 
         {\tp \sub{\tid}{\com'}, \state', (\events, \order)}
}
\\[10pt]
\myfrac{
  \eid = {\sf last}(\tid, (\events, \order))
  \quad \events(\eid) = (\tid, \op, \args, \Todo)
  \quad \events' = \events \sub{\eid}{(\tid, \op, \args, \state(\lres[\tid]))}
}{
  \htrans{\tp \sub{\tid}{\cskip}, \state, (\events, \order)}
         {\cdt}
         {\tp \sub{\tid}{\idle}, \state, (\events', \order)}
}
\end{array}
\]
First, an idle thread $\tid$ may call any operation $\op \in \opType$ with any
argument $\args$. This sets the continuation of thread $\tid$ to $\cdt(\op)$,
stores $\args$ into $\larg[\tid]$, adds a new event $i$ to the history, ordered
after all completed events. Second, a thread $\tid$ executing an operation may
do a transition allowed by the sequential semantics of the operation's
implementation. Finally, when a thread $\tid$ finishes executing an operation,
as denoted by a continuation $\cskip$, the corresponding event is completed with
the return value in $\lres[\tid]$. The identifier ${\sf last}(t, (E, R))$ of
this event is determined as the last one in $E$ by thread $t$ according to $R$:
as per Definition~\ref{def:history}, events by each thread are totally ordered
in a history, ensuring that ${\sf last}(\tid, \history)$ is well-defined.


Now given an initial state
$s_0 \in \stateType$, we define the set of histories of a data structure
$\cdt$ as 
$\hsem{\cdt, \state_0} = \{
  \history \mid
  \mhtrans{(\lambda \tid \ldotp \idle), \state_0, (\emptyset, \emptyset)}
         {\cdt}
         {\_, \_, \history} \}$.
We say that a data structure $(\cdt, \state_0)$ is {\em linearizable} with
respect to a set of sequential histories $\histories$ if $\hsem{\cdt,
\state_0} \sqsubseteq \histories$ (Definition~\ref{dfn:linearizability}).


%% file: logic.tex

\newcommand{\conf}{\kappa}
\newcommand{\Config}{{\sf Config}}

\newcommand{\lP}{P}
\newcommand{\lQ}{Q}

\newcommand{\hupd}[2]{{#1} \leadsto {#2}}
\newcommand{\mhupd}[2]{{#1} \mathrel{{\leadsto}^*} {#2}}

\newcommand{\spec}[5]{
{#1} \vdash_{#2}
\left \{#3 \right \}\, {#4}\,\left\{#5\right\}
}

\newcommand{\semj}[5]{
{#1}\vDash_{#2}
\left \{#3 \right \}
\,{#4}\,
\left \{ #5 \right \}
}

\newcommand{\addevent}[3]{
\left \langle {#2} \right \rangle
\dashrightarrow_{#1}
\left \langle {#3} \right \rangle
}

\newcommand{\assnType}{{\sf Assn}}
\newcommand{\lvarsType}{{\sf LVars}}

\section{Logic}\label{sec:logic}

We now formalise our proof method as a Hoare logic based on
rely-guarantee~\cite{rg}. We make this choice to keep presentation simple; our
method is general and can be combined with more advanced methods for reasoning
about concurrency~\cite{rgsep,CAP,turon-icfp14}.

Assertions $P, Q \in \assnType$ in our logic denote sets of {\em
  configurations} $\conf \in \Config = \stateType \times \historyType \times
\ghostType$, relating the data structure state, the abstract history and the
ghost state from a set $\ghostType$. The latter can be chosen separately for
each proof; e.g., in the proof of the TS queue in \S\ref{sec:informal} we used
$\ghostType = \eidType \to \tsType$.  We do not prescribe a particular syntax
for assertions, but assume that it includes at least the first-order logic, with
a set $\lvarsType$ of special {\em logical variables} used in specifications
and not in programs. We assume a function $\evalf{-}{-} : \assnType \times
(\lvarsType \to \valueType) \to \powerset{\Config}$ such that
$\evalf{P}{\lint}$ gives the denotation of an assertion $P$ with respect to an
interpretation $\lint : \lvarsType \to \valueType$ of logical variables.

Rely-guarantee is a {\em compositional} verification method: it allows reasoning
about the code executing in each thread separately under some assumption on its
environment, specified by a {\em rely}. In exchange, the thread has to ensure
that its behaviour conforms to a {\em guarantee}. Accordingly, judgements of our
logic take the form $\spec{\rely, \guar}{\tid}{\lP}{\com}{\lQ}$, where $\com$
is a command executing in thread $t$, $\lP$ and $\lQ$ are Hoare pre- and
post-conditions from $\assnType$, and $\rely, \guar \subseteq \Config^2$ are
relations defining the rely and the guarantee. Informally, the judgement states
that $\com$ satisfies the Hoare specification $\{\lP\} \_ \{\lQ\}$ and changes
program configurations according to $\guar$, assuming that concurrent threads
change program configurations according to $\rely$.



\begin{figure}[t]
\[
\myfrac{
\forall \lint \ldotp
\semj{\guar}{\tid}{\evalf{\lP}{\lint}}{\pcom}{\evalf{\lQ}{\lint}} 
\land {\sf stable}(\evalf{\lP}{\lint}, \rely)
\land {\sf stable}(\evalf{\lQ}{\lint}, \rely)
}{
  \spec{\rely, \guar}{\tid}{\lP}{\pcom}{\lQ}
}
\]
where for $p, q \in \powerset{\Config}$:
\[
{\sf stable}(p, \rely) \triangleq
\forall \conf, \conf' \ldotp
\conf \in p \land
(\conf, \conf') \in \rely
\implies
\conf' \in p
\]
\begin{multline*}
\semj{\guar}{\tid}{p}{\pcom}{q} \triangleq
\forall \state, \state', \history, \ghost \ldotp 
(\state, \history, \ghost) \in p
\land
\state' \in \intp{\tid}{\pcom}{\state} \implies
{} \\
\exists \history', \ghost' \ldotp
(\state', \history', \ghost') \in q \land
\mhupd{\history}{\history'} \land
((\state, \history, \ghost), (\state', \history', \ghost')) \in \guar
\end{multline*}
and for $(E, R), (E', R') \in \historyType$:
\begin{multline*}
\hupd{(\events, \order)}{(\events', \order')}
\triangleq
(\events = \events' \land \order \subseteq \order')
\lor{}\\
(\exists \eid, \tid, \op, \args, \res \ldotp
(\forall \eidp \ldotp \eidp \,{\neq}\, \eid \implies \events(\eidp) \,{=}\, \events'(\eidp)) \land{} \\
\events(\eid) \,{=}\, (\tid, \op, \args, \Todo) \land
\events'(\eid) \,{=}\, (\tid, \op, \args, \res))
\end{multline*}

\vspace{-10pt}
\caption{Proof rule for primitive commands.}
\label{fig:prim}
\end{figure}

Our logic includes the standard Hoare proof rules for reasoning about sequential
control-flow constructs, which we defer to \cite{ext} due to space
constraints. We now explain the rule for atomic commands in
Figure~\ref{fig:prim}, which plays a crucial role in formalising our proof
method. The proof rule derives judgements of the form $\spec{\rely,
\guar}{\tid}{\lP}{\pcom}{\lQ}$. The rule takes into account possible
interference from concurrent threads by requiring the denotations of $\lP$ and
$\lQ$ to be {\em stable} under the rely $\rely$, meaning that they are preserved
under transitions the latter allows. The rest of the requirements are expressed
by the judgement $\semj{\guar}{\tid}{p}{\pcom}{q}$. This requires that for any
configuration $(\state, \history, \ghost)$ from the precondition denotation $p$
and any data structure state $\state'$ resulting from thread $\tid$ executing
$\pcom$ in $\state$, we can find a history $\history'$ and a ghost state
$\ghost'$ such that the new configuration $(\state', \history', \ghost')$
belongs to the postcondition denotation $q$. This allows updating the history
and the ghost state (almost) arbitrarily, since these are only part of the proof
and not of the actual data structure implementation; the shaded code in
Figures~\ref{fig:tsq_enqueue} and~\ref{fig:tsq_dequeue} indicates how we perform
these updates in the proof of the TS queue. Updates to the history, performed
when $\pcom$ is a commitment point, are constrained by a relation ${\leadsto}
\subseteq \historyType^2$, which only allows adding new edges to the real-time
order or completing events with a return value. This corresponds to commitment
points of kinds 2 and 3 from \S\ref{sec:linearizability}. Finally, as is usual
in rely-guarantee, the judgement $\semj{\guar}{\tid}{p}{\pcom}{q}$ requires that
the change to the program configuration be allowed by the guarantee $\guar$.


Note that ${\leadsto}$ does not allow adding new events into histories
(commitment point of kind 1): this happens automatically when an operation is
invoked. In the following, we use a relation ${\dashrightarrow}_\tid \subseteq
\Config^2$ to constrain the change to the program configuration upon an
operation invocation in thread $t$:
\[
\begin{array}{@{}l@{}}
\addevent{\tid}
  {\state, (\events, \order), \ghost}
  {\state', (\events', \order'), \ghost'}
\iff
(\forall l \in {\sf Loc} \ldotp l \neq \larg[\tid] \implies
\state(l) = \state'(l)) \\ \hfill {} \land
\exists \eid \notin \ids{\events} \ldotp
\events' = \events  \uplus \{[i : t, \_, \_, \Todo]\}
\\ \hfill {} \land
\order' = (\order \cup \{ (\eidp, \eid) \mid \eidp \in \complete{\events}\})
\land \ghost = \ghost'
\end{array}
\]
Thus, when an operation is invoked in thread $\tid$, $\larg[\tid]$ is
overwritten by the operation argument and an uncompleted event associated with
thread $\tid$ and a new identifier $\eid$ is added to the history; this event is
ordered after all completed events, as required by our proof method
(\S\ref{sec:linearizability}).

The rule for primitive commands and the standard Hoare logic proof rules allow
deriving judgements about the implementations $\cdt(\op)$ of every operation
$\op$ in a data structure $\cdt$. The following theorem formalises the
requirements on these judgements sufficient to conclude the linearizability of
$\cdt$ with respect to a given set of sequential histories $\histories$. The
theorem uses the following auxiliary assertions, describing the event
corresponding to the current operation $\op$ in a thread $\tid$ at the start
and end of its execution (${\sf last}$ is defined in \S\ref{sec:lang}):
\[
\begin{array}{r@{\ }c@{\ }l}
\evalf{\started_{\mathcal{I}}(\tid, \op)}{\lint}
&=& \{ (\state, (\events, \order), \ghost) \mid 
\events({\sf last}(\tid, (\events, \order))) = (\tid, \op, \state(\larg[\tid]), \Todo)
\\ && \hfill
{} \land
\exists \conf \in \evalf{\mathcal{I}}{\lint} \ldotp
\addevent{\tid}
  {\conf}
  {\state, (\events, \order), \ghost}
\};
\\
\evalf{\finished(\tid, \op)}{\lint}
&=& \{(\state, (\events, \order), \ghost) \mid
  \events({\sf last}(\tid, (\events, \order))) = (\tid, \op, \_, \state(\lres[\tid]))\}.
\end{array}
\]
The assertion $\started_{\mathcal{I}}(\tid, \op)$ is parametrised by a global
invariant $\mathcal{I}$ used in the proof. With the help of it,
$\started_{\mathcal{I}}(\tid, \op)$ requires that configurations in its
denotation be results of adding a new event into histories satisfying
$\mathcal{I}$.

\begin{theorem}\label{thm:soundness}
  Given a data structure $\cdt$, its initial state $\state_0 \in \stateType$ and
  a set of sequential histories $\histories$, we have $(\cdt, \state_0)$ {\em
  linearizable} with respect to $\histories$ if there exists an assertion
  $\mathcal{I}$ and relations $\rely_{\tid}, \guar_{\tid} \subseteq \Config^2$
  for each $\tid \in \tidType$ such that:
\begin{enumerate}
\item \label{item:init} 
  $\exists \ghost_0.\, \forall \lint \ldotp (\state_0, (\emptyset, \emptyset),
  \ghost_0) \in \evalf{\mathcal{I}}{\lint}$;
\item \label{item:stable}
$\forall \tid, \lint. \, {\sf stable}(\evalf{\mathcal{I}}{\lint}, \rely_\tid)$;
\item \label{item:inv} 
  $\forall \history, \lint \ldotp (\_, \history, \_) \in \evalf{\mathcal{I}}{\lint} \implies
  \abs(\history, \histories)$;
\item \label{item:spec}
$\forall \tid, \op.\, (\spec{\rely_\tid, \guar_\tid}
      {\tid}
      {\begin{array}{@{}c@{}}
      \mathcal{I} \land \started_{\mathcal{I}}(\tid, \op)
      \end{array}}
      {\cdt(\op)}
      {\begin{array}{@{}c@{}}
      \mathcal{I} \land \finished(\tid, \op)
      \end{array}})$;
\item \label{item:rg}
$\forall \tid, \tid'.\, \tid \neq \tid' \implies
  \guar_{\tid} \cup {\dashrightarrow}_{\tid} \subseteq \rely_{\tid'}$.
\end{enumerate}
\end{theorem} 

Here $\mathcal{I}$ is the invariant used in the proof, which
item~\ref{item:init} requires to hold of the initial data structure state $s_0$,
the empty history and some some initial ghost state $G_0$.
Item~\ref{item:stable} then ensures that the invariant holds at all times.
Item~\ref{item:inv} requires any history satisfying the invariant to be an
abstract history of the given specification $\histories$
(Definition~\ref{dfn:abstracthistory},
\S\ref{sec:linearizability}). Item~\ref{item:spec} constraints the judgement
about an operation $\op$ executed in a thread $t$: the operation is executed
from a configuration satisfying the invariant and with a corresponding event
added to the history; by the end of the operation's execution, we need to
complete the event with the return value matching the one produced by the
code. Finally, item~\ref{item:rg} formalises a usual requirement in
rely-guarantee reasoning: actions allowed by the guarantee of a thread $\tid$
have to be included into the rely of any other thread $\tid'$. We also include
the relation ${\dashrightarrow}_\tid$, describing the automatic creation of a
new event upon an operation invocation in thread $t$.


%% file: details.tex

\newcommand{\newts}{{\rm newTS}}
\newcommand{\isEmpty}{{\rm isEmpty}}

\newcommand{\seen}{{\sf seen}}
\newcommand{\fnocand}{{\rm noCand}}
\newcommand{\fiscand}{{\rm isCand}}
\newcommand{\fmints}{{\rm minTS}}

\newcommand{\ttvisited}{{\tt A}}

\newcommand{\invdata}{{\sf same\_data}}

\begin{figure}[t]
\begin{itemize}
\item[$(\invlin)$] all linearizations of completed events of the abstract
history satisfy the queue specification:
\vspace{-5pt}
\[
\forall H' \ldotp \linearizes{\complete{\history}}{H'} \land \isSeq(H')
\implies H' \in \qhist \land \invdata(\state, \history, \ghostts, H')
\]

\vspace{-7pt}
\end{itemize}
\begin{itemize}
\item[$(\invord)$] properties of the partial order of the abstract history:
\begin{itemize}[leftmargin=0pt]
\item[(i)] completed dequeues precede uncompleted ones:
\vspace{-5pt}
\[
\forall \eid \in \ids{\complete{\events}} \ldotp
  \forall \eidp \in \ids{\incomplete{\events}} \ldotp
    \opof{\events(\eid)} = \opof{\events(\eidp)} = \deqOp \implies \edge{\eid}{\order}{\eidp}
\]

\vspace{-7pt}
\item[(ii)] enqueues of already dequeued values precede enqueues of values
in the pools:
\vspace{-5pt}
\[
\forall \eid \in \ids{\complete{\events}} \setminus \inqueue{\state(\ttpools), \events, \ghostts} \ldotp
  \forall \eidp \in \inqueue{\state(\ttpools), \events, \ghostts} \ldotp
    \edge{\eid}{\order}{\eidp}
\]
\end{itemize}

\item[$(\invalg)$] properties of the algorithm used to build the loop
invariant:
\begin{itemize}[leftmargin=0pt]
\item[(i)]
enqueues of values in the pools are ordered only if so are
their timestamps:
\vspace{-5pt}
\[
\forall \eid, \eidp \in \inqueue{\state({\tt pools}), \events, \ghostts} \ldotp
\edge{\eid}{\order}{\eidp}
\implies \ghostts(\eid) \tsless \ghostts(\eidp)
\]

\vspace{-7pt}
\item[(ii)]
values in each pool appear in the order of enqueues that inserted them:
\vspace{-5pt}
\begin{multline*}
\forall \tid, \tau_1, \tau_2 \ldotp
\pool{\tid} = \_ \cdot (\_, \_, \tau_1) \cdot \_ \cdot (\_, \_, \tau_2) \cdot \_ \implies{}\\
\edge{\getEvent(\events, \ghostts, \tid, \tau_1)}{\order}{\getEvent(\events, \ghostts, \tid,
\tau_2)}
\end{multline*}

\vspace{-7pt}
\item[(iii)]
the timestamps of values are smaller than the global counter:
\[
\forall \eid, a, b \ldotp \ghostts(\eid) = (a, b) \implies b < \state({\tt counter})
\]
\end{itemize}
\item[($\invwf$)] properties of ghost state:
\begin{itemize}[leftmargin=0pt]
\item[(i)]
$\ghostts$ associates timestamps with enqueue events:
\vspace{-5pt}
\[
\forall \eid \ldotp \eid \in \dom{\ghostts} \implies
\opof{\events(\eid)} = \enqOp
\]

\vspace{-7pt}
\item[(ii)]
each value in a pool has a matching event for the enqueue that inserted
it:
\vspace{-5pt}
\[
\forall \tid, v, \tau \ldotp \pool{\tid} = \_ \cdot (\_, v, \tau) \cdot \_
\implies
\exists \eid \ldotp
\events(\eid) = (\tid, \enqOp, v, \_) \land \ghostts(\eid) = \tau
\]

\vspace{-7pt}
\item[(iii)]
all timestamps assigned by $\ghostts$ to events of a given thread are distinct:
\vspace{-5pt}
\[
\forall \eid, \eidp \ldotp \eid \neq \eidp \land \tidof{\events(\eid)} =
\tidof{\events(\eidp)} \land \eid, \eidp \in \dom{\ghostts}
\implies \ghostts(\eid) \neq \ghostts(\eidp) 
\]

\vspace{-7pt}
\item[(iv)]
$\ghostts$ associates uncompleted enqueues events with the timestamp
$\top$:
\vspace{-5pt}
\[
\begin{array}{l}
\forall \eid \ldotp \opof{\events(\eid)} = \enqOp \implies(\eid \not\in \ids{\complete{\events}}
\iff \eid \notin \dom{\ghostts} \lor \ghostts(\eid) = \top)
\end{array}
\]
\end{itemize}
\end{itemize}
\vspace{-15pt}
\caption{The invariant $\inv = \invlin \land \invord \land \invalg \land \invwf$}
\label{tsq:inv}
\end{figure}

\section{The TS Queue: proof details}\label{sec:details}

In this section, we present some of the details of the proof of the TS
Queue. Due to space constraints, we provide the rest of them in
the extended version of the paper \cite{ext}.

\newparagraph{Invariant.}
We satisfy the obligation~\ref{item:spec} from Theorem~\ref{thm:soundness} by
proving the invariant $\inv$ defined in Figure~\ref{tsq:inv}. The invariant is
an assertion consisting of four parts: $\invlin$, $\invord$, $\invalg$ and
$\invwf$. Each of them denotes a set of configurations satisfying the listed
constraints for a given interpretation of logical variables $\lint$. The first
part of the invariant, $\invlin$, ensures that every history satisfying the
invariant is an abstract history of the queue, which discharges the
obligation~\ref{item:inv} from Theorem~\ref{thm:soundness}. In addition to that,
$\invlin$ requires that a relation $\invdata$ hold of a configuration $(\state,
\history, \ghostts)$ and every linearization $H'$. In this way, we ensure that
the pools and the final state of the sequential queue after $H'$ contain values
inserted by the same enqueue events (we formalise $\invdata$ in \cite{ext}). The
second part, $\invord$, asserts ordering properties of events in the partial
order that hold by construction. The third part, $\invalg$, is a
collection of properties relating the order on timestamps to the partial order
in abstract history. Finally, $\invwf$ is a collection of well-formedness
properties of the ghost state.


\newparagraph{Loop invariant.}
We now present the key verification condition that arises in the {\tt dequeue}
operation: demonstrating that the ordering enforced at the commitment points at
lines~\ref{line:tsq_scan}--\ref{line:tsq_scan_end} and
\ref{line:tsq_remove}--\ref{line:tsq_remove_ends} does not invalidate acyclicity
of the abstract history. To this end, for the {\tt foreach} loop
(lines~\ref{line:tsq_loop_begin}--\ref{line:tsq_loop_end}) we build a loop
invariant based on distinguishing certain values in the pools as {\it seen} by
the {\tt dequeue} operation. With the help of the loop invariant we establish
that acyclicity is preserved at the commitment points.

\begin{figure}[t]
\begin{itemize}[leftmargin=58pt]
\item[$(\fnocand)$:]
  $\seen((\state, \history, \ghostts), \myEid) = \emptyset \land \state(\ttcandpid) = \NULL$
\item[$(\fmints(e))$:]
  $\forall e' \in \seen((\state, \history, \ghostts), \myEid) \ldotp
    \lnot(\ghostts(e') \tsless \ghostts(e))$
\item[$(\fiscand)$:] \hspace{-4pt}
  $\begin{multlined}[t]
  \exists \ttcand \ldotp
  \ttcand = \getEvent(\events, \ghostts,
                         \state(\ttcandtid), \state(\ttcandts))
  \\ {} \land
    \fmints(\ttcand)
  \land
  (\ttcand \in \inqueue{\state(\ttpools), \events, \ghostts} \implies{}\\
    \ttcand \in \seen((\state, \history, \ghostts), \myEid))
    \land \state(\ttcandpid) \neq \NULL
  \end{multlined}$
\end{itemize}

\caption{Auxiliary assertions for the loop invariant
  \label{fig:tsq_linv}
}
\end{figure}

Recall from \S\ref{sec:running}, that the {\tt foreach} loop starts iterating
from a random pool. In the proof, we assume that the loop uses a thread-local
variable $\ttvisited$ for storing a set of identifiers of threads that have been
iterated over in the loop. We also assume that at the end of each iteration the
set $\ttvisited$ is extended with the current loop index ${\tt k}$.

Note also that for each thread ${\tt k}$, the commitment point of a dequeue $d$
at lines~\ref{line:tsq_scan}--\ref{line:tsq_scan_end} ensures that enqueue
events of values the operation sees in ${\tt k}$'s pool precede $d$ in the
abstract history. Based on that, during the {\tt foreach} loop we can we
distinguish enqueue events with values in the pools that a dequeue $d$ has
seen after looking into pools of threads from $\ttvisited$. We define the set of
all such enqueue events as follows:
\begin{multline}\label{def:seen}
\seen((\state, (\events, \order), \ghostts), d) \triangleq
  \{ e \mid
  e \in \ids{\complete{\events}} \cap
            \inqueue{\state({\tt pools}), \events, \ghostts} \\ {}
  \land \edge{e}{\order}{d}
  \land \lnot(\state(\ttstartts) \tsless \ghostts(e))
  \land \tidof{\events(e)} \in \ttvisited
  \}
\end{multline}

A loop invariant $\linv$ is simply a disjunction of two auxiliary assertions,
$\fiscand$ and $\fnocand$, which are defined in Figure~\ref{fig:tsq_linv} (given
an interpretation of logical variables $\lint$, each of assertions denotes a set
of configurations satisfying the listed constraints). The assertion $\fnocand$
denotes a set of configurations $\conf = (\state, (\events, \order), \ghostts)$,
in which the {\tt dequeue} operation has not chosen a candidate for removal
after having iterated over the pools of threads from $\ttvisited$. In this case,
$\state(\ttcandpid) = \NULL$, and the current dequeue has not seen any enqueue
event in the pools of threads from $\ttvisited$.

The assertion $\fiscand$ denotes a set of configurations $\conf = (\state,
(\events, \order), \ghostts)$, in which an enqueue event $\ttcand =
\getEvent(\events, \ghostts, \ttcandtid, \ttcandts)$ has been chosen as a
candidate for removal out of the enqueues seen in the pools of threads from
$\ttvisited$. As $\ttcand$ may be removed by a concurrent dequeue, $\fiscand$
requires that $\ttcand$ remain in the set $\seen(\conf, \myEid)$ as long as
$\ttcand$'s value remains in the pools. Additionally, by requiring
$\fmints(\ttcand)$, $\fiscand$ asserts that the timestamp of $\ttcand$ is
minimal among other enqueues seen by $\myEid$.

In the following lemma, we prove that the assertion $\fiscand$ implies
minimality of $\ttcand$ in the abstract history among enqueue events with values
in the pools of threads from $\ttvisited$. The proof is based on the observation
that enqueues of values seen in the pools by a dequeue are never preceded by
unseen enqueues.

\begin{lemma}\label{lem:tsqmin}
For every $\lint : \lvarsType \to \valueType$ and configuration $(\state,
(\events, \order), \ghostts) \in \evalf{\fiscand}{\lint}$, if $\ttcand =
\getEvent(\events, \ghostts, \ttcandtid, \ttcandts)$ and $\ttcand \in
\inqueue{\state({\tt pools}), \events, \ghostts}$ both hold, then the following is true:
\begin{equation*}
\forall e \in \inqueue{\state(\tt pools), \events, \ghostts} \ldotp
\tidof{\events(e)} \in \ttvisited
\implies
\nedge{e}{\order}{\ttcand}
\end{equation*}
\end{lemma}

\newparagraph{Acyclicity.} At the commitment points extending the order of the
abstract history, we need to show that the extended order is acyclic as required
by Definition~\ref{def:history} of the abstract history. To this end, we argue
that the commitment points at lines~\ref{line:tsq_scan}--\ref{line:tsq_scan_end}
and lines~\ref{line:tsq_remove}--\ref{line:tsq_remove_ends} preserve acyclicity
of the abstract history.

The commitment point at lines~\ref{line:tsq_scan}--\ref{line:tsq_scan_end}
orders certain completed enqueue events before the current uncompleted dequeue
event $\myEid$. By Definition~\ref{def:history} of the abstract history, the
partial order on its events is transitive, and uncompleted events do not precede
other events. Since $\myEid$ does not precede any other event, ordering any
completed enqueue event before $\myEid$ cannot create a cycle in the abstract
history.

We now consider the commitment point at
lines~\ref{line:tsq_remove}--\ref{line:tsq_remove_ends} in the current dequeue
$\myEid$. Prior to the commitment point, the loop invariant $\linv$ has been
established in all threads, and the check $\ttcandpid \neq \NULL$ at
line~\ref{line:tsq_checknull} has ruled out the case when $\fnocand$ holds.
Thus, the candidate for removal $\ttcand$ has the properties described by
$\fiscand$. If $\ttcand$'s value has already been dequeued concurrently, the
removal fails, and the abstract history remains intact (and acyclic). When the
removal succeeds, we consider separately the two kind of edges added into the
abstract history $(\events, \order)$:
\begin{enumerate}
\item \textbf{The case of $(\ttcand, e)$ for each $e \in \inqueue{{\tt
pools}, \events, \ghostts}$}. By Lemma~\ref{lem:tsqmin}, an edge $(e, \ttcand)$
is not in the partial order $\order$ of the abstract history. There is also no
sequence of edges $\edge{\edge{e}{\order}{...}}{\order}{\ttcand}$, since
$\order$ is transitive by Definition~\ref{def:history}. Hence, cycles do not
arise from ordering $\ttcand$ before $e$.
\item \textbf{The case of $(\myEid, d)$ for each identifier $d$ of an
uncompleted dequeue event}. By Definition~\ref{def:history} of the abstract
history, uncompleted events do not precede other events. Since $d$ is
uncompleted event, it does not precede $\myEid$. Hence, ordering $\myEid$ in
front of all such dequeue events does not create cycles.
\end{enumerate}

\FloatBarrier
\newparagraph{Rely and guarantee relations.}
We now explain how we generate rely and guarantee relations for the proof.
Instead of constructing the relations with the help of abstracted intermediate
assertions of a proof outline for the {\tt enqueue} and {\tt dequeue}
operations, we use the non-deterministic state transformers of primitive
commands together with the ghost code in Figure~\ref{fig:tsq_enqueue} and
Figure~\ref{fig:tsq_dequeue}. To this end, the semantics of state transformers
is extended to account for changes to abstract histories and ghost state. We
found that generating rely and guarantee relations in such non-standard way
results in cleaner stability proofs for the TS Queue, and makes them similar in
style to checking non-interference in the Owicki-Gries method
\cite{owicki-gries}.

Let us refer to atomic blocks with corresponding ghost code at
line~\ref{line:tsq_insert}, line~\ref{line:tsq_setts}, line~\ref{line:tsq_scan}
and line~\ref{line:tsq_remove} as atomic steps {\tt insert}, {\tt setTS}, {\tt
scan(k)} (${\tt k} \in \tidType$) and {\tt remove} respectively, and let us also
refer to the CAS operation at line~\ref{line:ts_CAS} as {\tt genTS}. For each
thread $\tid$ and atomic step $\hat{\alpha}$, we assume a non-deterministic
configuration transformer $\db{\hat{\alpha}}_\tid : \Config \to
\powerset{\Config}$ that updates state according to the semantics of a
corresponding primitive command, and history with ghost state as specified by
ghost code.

Given an assertion $P$, an atomic step $\hat{\pcom}$ and a thread $\tid$,
we associate them with the following relation $\guar_{\tid,\hat{\alpha},P}
\subseteq \Config^2$:
\[
\guar_{\tid,\hat{\alpha},P} \triangleq
\{
(\conf, \conf') \mid \exists \lint \ldotp
\conf \in \evalf{P}{\lint} \land \conf' \in \intp{\tid}{\hat{\pcom}}{\conf}
\}
\]
Additionally, we assume a relation $\guar_{\tid, {\tt local}}$, which describes
arbitrary changes to certain program variables and no changes to the abstract
history and the ghost state. That is, we say that {\tt pools} and {\tt counter}
are {\em shared} program variables in the algorithm, and all others are {\em
thread-local}, in the sense that every thread has its own copy of them. We let
$\guar_{\tid, {\tt local}}$ denote every possible change to thread-local
variables of a thread $\tid$ only.

For each thread $\tid$, relations $\guar_\tid$ and $\rely_\tid$ are defined as follows:
\[
\begin{array}{rcl}
P_\op & \triangleq & \inv\land\started(\tid,\op) \\
\guar_\tid & \triangleq &
(\bigcup_{\tid' \in \tidType} \guar_{\tid,{\tt scan(\tid')}, P_\deqOp}) \cup
\guar_{\tid,{\tt remove},P_\deqOp}
\\ & & \hfill {}
\cup
\guar_{\tid,{\tt insert},P_\enqOp} \cup
\guar_{\tid,{\tt setTS},P_\enqOp} \cup
\guar_{\tid,{\tt genTS},\inv} \cup
\guar_{\tid, {\tt local}},
\\
\rely_\tid & \triangleq &
\cup_{\tid' \in \tidType \setminus \{\tid\}}
(\guar_{\tid'} \cup {\dashrightarrow}_{\tid'})
\end{array}
\]
As required by Theorem~\ref{thm:soundness}, the rely relation of a thread $\tid$
accounts for addition of new events in every other thread $\tid'$ by including
${\dashrightarrow}_\tid'$. Also, $\rely_{\tid}$ takes into consideration every
atomic step by the other threads. Thus, the rely and guarantee relations satisfy
all the requirement~\ref{item:rg} of the proof method from
Theorem~\ref{thm:soundness}. It is easy to see that the
requirement~\ref{item:stable} is also fulfilled: the global invariant $\inv$ is
simply preserved by each atomic step, so it is indeed stable under rely
relations of each thread.



The key observation implying stability of the loop invariant in every thread
$\tid$ is presented in the following lemma, which states that environment
transitions in the rely relation never extend the set of enqueues seen by a
given dequeue.
\begin{lemma}\label{lem:tsq:linv:stablevis}
If a dequeue event ${\tt DEQ}$ generated its timestamp $\ttstartts$, then:
\[
\forall \conf, \conf' \ldotp
(\conf, \conf') \in \rely_{\tid} \implies
\seen(\conf', {\tt DEQ}) \subseteq \seen(\conf, {\tt DEQ})
\]
\end{lemma}



%% file: set.tex

\newcommand{\nodeType}{{\sf NodeID}}
\newcommand{\node}{n}
\newcommand{\ghostnode}{\ghost_{\sf node}}

\newcommand{\sllist}{{\tt list}}

\newcommand{\containsOp}{{\sf contains}}
\newcommand{\insertOp}{{\sf insert}}
\newcommand{\removeOp}{{\sf remove}}

\newcommand{\insertOf}{{\sf insOf}}
\newcommand{\removeOf}{{\sf remOf}}
\newcommand{\lastremove}{{\sf lastRemOf}}

\newcommand{\ttafter}{{\tt obs}}

\section{The Optimistic Set: Informal Development}
\label{sec:set}

\begin{figure}[t]
\hspace{1em}
\begin{minipage}[t]{.41\textwidth}
\begin{lstlisting}
struct Node {
  Node *next;
  Int val;
  Bool marked;
}

Bool contains(v) {
  p, c := locate(v);
  return (c.val = v);
}

Bool insert(v) {
  ${\tt Node}{\times}{\tt Node}$ p, c;
  do {
    p, c := locate(v);
    atomic { $\label{line:insert_atomic_start}$
      if (p.next = c 
          && !p.marked) {
\end{lstlisting}
\begin{lstlisting}[firstnumber=18,backgroundcolor=\color{gray}]
        #${\tt commit}_{\tt insert}$#();
\end{lstlisting}
\begin{lstlisting}[firstnumber=20]
        if (c.val ${\neq}$ v) {
          Node *n := new Node;
          n->next := c;
          n->val := v;
          n->marked := $\FALSE$;
          p.next := n;
          return $\TRUE$;
        } else
          return $\FALSE$;
      }
    } $\label{line:insert_atomic_end}$
  } while ($\TRUE$);
}
\end{lstlisting}
\end{minipage}
\hfill
\begin{minipage}[t]{.46\textwidth}
\begin{lstlisting}[firstnumber=33]
${\tt Node}{\times}{\tt Node}$ locate(v) {
  Node prev := head;
  Node curr := prev.next;
  while (curr.val < v) {
    prev := curr;
    atomic { $\label{line:locate_atomic_start}$
      curr := curr.next;
\end{lstlisting}
\begin{lstlisting}[firstnumber=40,backgroundcolor=\color{gray}]
      if (#$\opof{\events(\myEid)} = \containsOp$# #$\label{line:locatecommit_start}$#
          && (curr.val ${}\geq{}$ v))
        #${\tt commit}_{\tt contains}$#();$\label{line:locatecommit_end}$
\end{lstlisting}
\begin{lstlisting}[firstnumber=45]
    } $\label{line:locate_atomic_end}$
  }
  return prev, curr;
}

Bool remove(v) {
  ${\tt Node}{\times}{\tt Node}$ p, c;
  do {
    p, c := locate(v);
    atomic { $\label{line:remove_atomic_start}$
      if (p.next = c 
          && !p.marked) {
\end{lstlisting}
\begin{lstlisting}[firstnumber=55,backgroundcolor=\color{gray}]
        #${\tt commit}_{\tt remove}$#();
\end{lstlisting}
\begin{lstlisting}[firstnumber=57]
        if (c.val = v) {
          c.marked := $\TRUE$;
          p.next := c.next;
          return $\TRUE$;
        } else
          return $\FALSE$;
      }
    } $\label{line:remove_atomic_end}$
  } while ($\TRUE$);
}
\end{lstlisting}
\end{minipage}
\caption{The Optimistic Set. Shaded portions are auxiliary code used in the
proof}
\label{fig:setcode}
\end{figure}

\newparagraph{The algorithm.}
We now present another example, the Optimistic Set~\cite{hindsight}, which is a
variant of a classic algorithm by Heller et al.~\cite{lazy-list}, rewritten to
use atomic sections instead of locks. However, this is a highly-concurrent
algorithm: every atomic section accesses a small bounded number of memory
locations. In this section we only give an informal explanation of the proof and
commitment points; the details are provided in \cite{ext}.

The set is implemented as a sorted singly-linked list. Each node in the list has
three fields: an integer {\tt val} storing the key of the node, a pointer {\tt
next} to the subsequent node in the list, and a boolean flag {\tt marked} that
is set true when the node gets removed. The list also has sentinel nodes {\tt
head} and {\tt tail} that store $-\infty$ and $+\infty$ as keys accordingly. The
set defines three operations: {\tt insert}, {\tt remove} and {\tt contains}.
Each of them uses an internal operation {\tt locate} to traverse the list. Given
a value {\tt v}, {\tt locate} traverses the list nodes and returns a pair of
nodes {\tt (p, c)}, out of which {\tt c} has a key greater or equal to {\tt v},
and {\tt p} is the node preceding {\tt c}.

The {\tt insert} ({\tt remove}) operation spins in a loop locating a place after
which a new node should be inserted (after which a candidate for removal should
be) and attempting to atomically modify the data structure. The attempt may fail
if either {\tt p.next = c} or {\tt !p.marked} do not hold: the former condition
ensures that concurrent operations have not removed or inserted new nodes
immediately after {\tt p.next}, and the latter checks that {\tt p} has not been
removed from the set. When either check fails, the operation restarts. Both
conditions are necessary for preserving integrity of the data structure.

When the elements are removed from the set, their corresponding nodes have the
{\tt marked} flag set and get unlinked from the list. However, the {\tt next}
field of the removed node is not altered, so marked and unmarked nodes of the
list form a tree such that each node points towards the root, and only nodes
reachable from the head of the list are unmarked. In Figure~\ref{fig:set}, we
have an example state of the data structure. The {\tt insert} and {\tt remove}
operations determine the position of a node {\tt p} in the tree by checking
the flag {\tt p.marked}. In {\tt remove}, this check prevents
removing the same node from the data structure twice. In {\tt insert}, checking
{\tt !p.marked} ensures that the new node {\tt n} is not inserted
into a branch of removed nodes and is reachable from the head of the list.

In contrast to {\tt insert} and {\tt remove}, {\tt contains} never modifies the
shared state and never restarts. This leads to a subtle interaction that may
happen due to interference by concurrent events: it may be correct for {\tt
contains} to return $\TRUE$ even though the node may have been removed by the
time {\tt contains} finds it in the list.


\begin{figure}[t]
\begin{minipage}{0.48\textwidth}
\includegraphics[height=135pt,trim=1 1 1 1,clip]{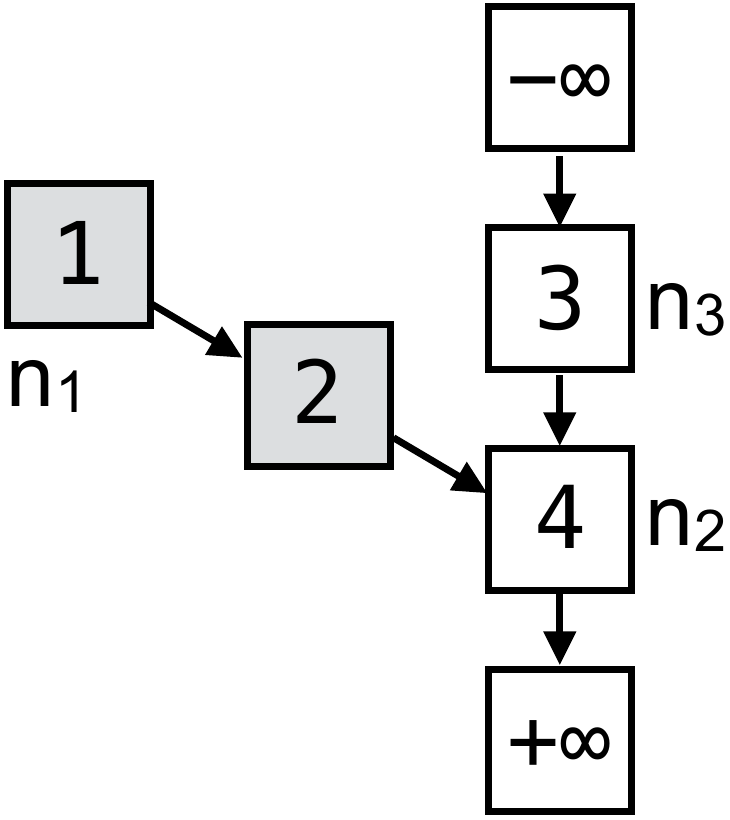}
\caption{Example state of the optimistic set. Shaded nodes have their ``marked''
field set.
  \label{fig:set}}
\end{minipage}
\hfill
\begin{minipage}{0.47\textwidth}
\hspace{-30pt}
\includegraphics[height=95pt,trim=1 1 1 1,clip]{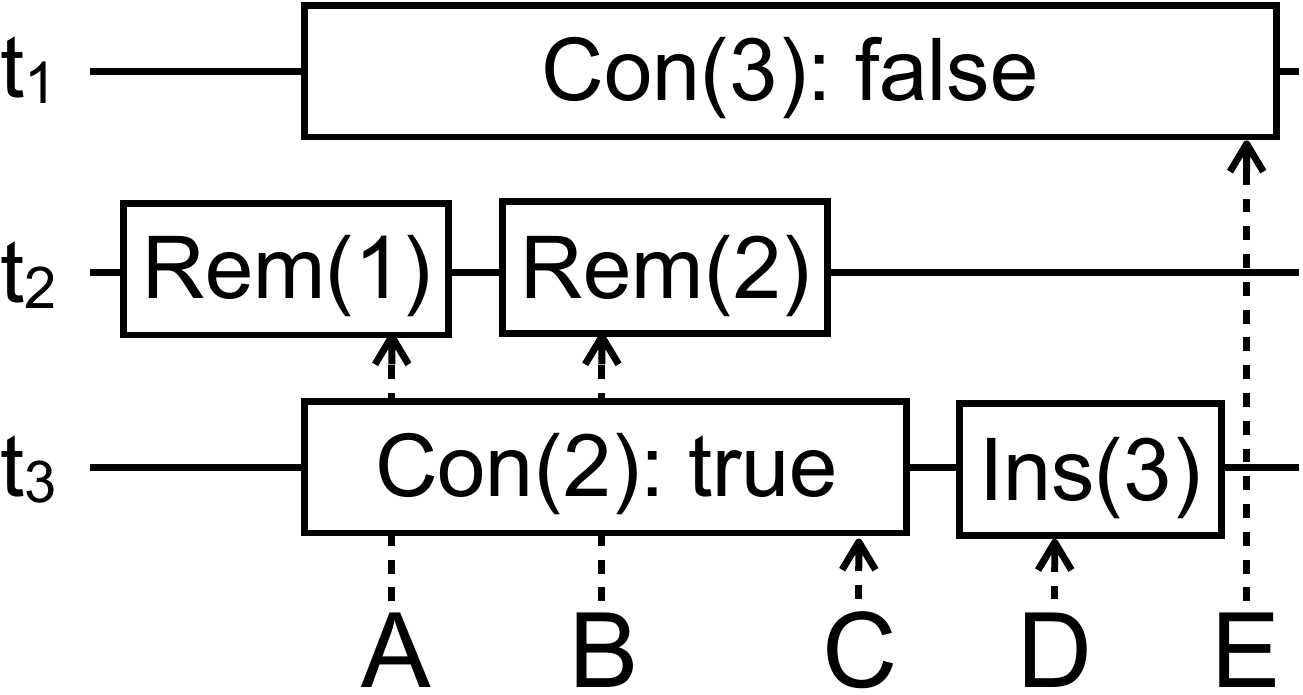}
\caption{Example execution of the set. ``Ins'' and ``Rem'' denote successful {\tt insert} and {\tt
  remove} operations accordingly, and ``Con'' denotes {\tt contains} operations.
  A--E correspond to commitment points of operations.
\label{fig:settrace}}
\end{minipage}

\medskip
\begin{center}
\includegraphics[height=170pt,trim=1 1 1 1,clip]{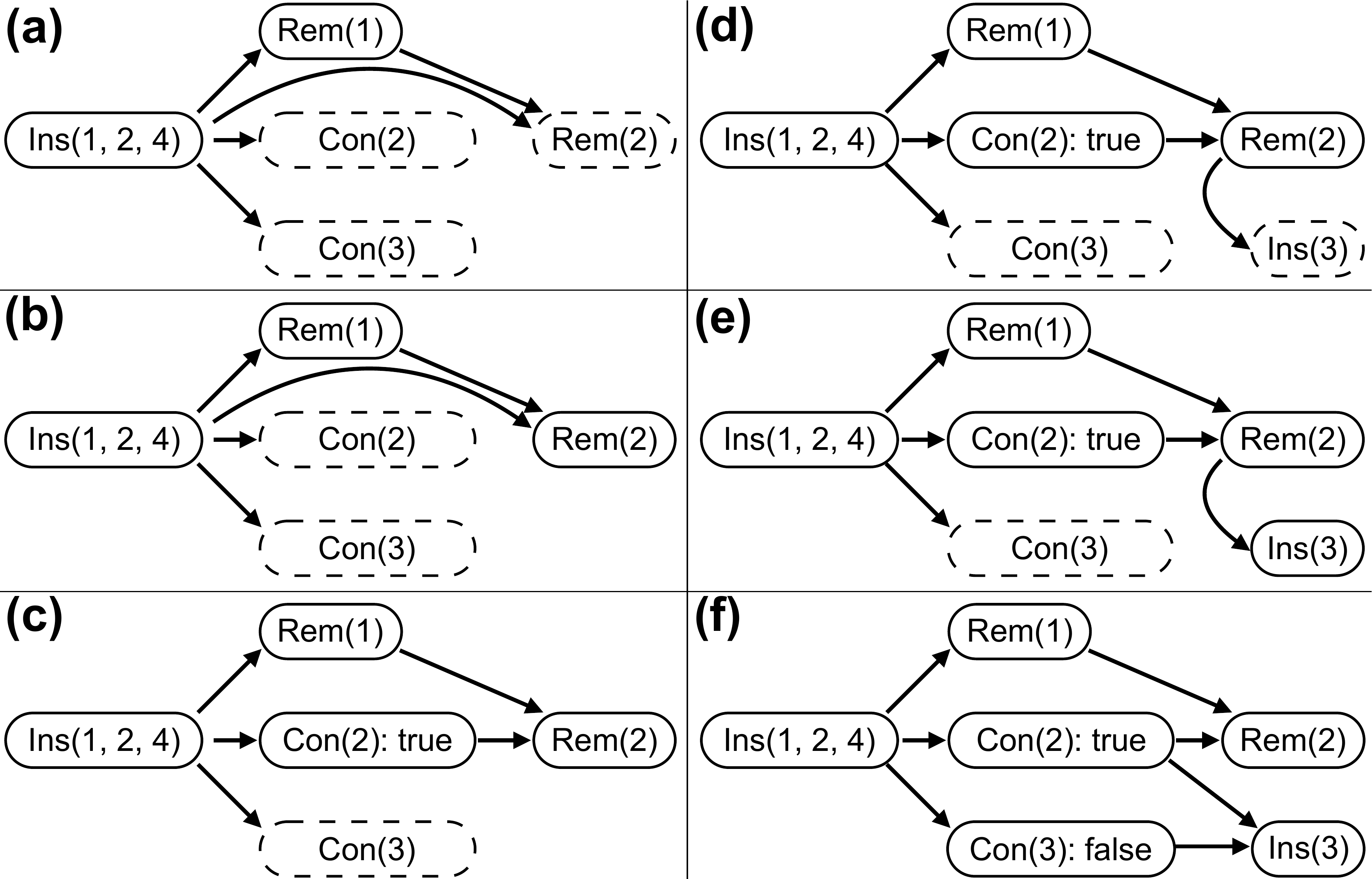}
\end{center}
\caption{Changes to the abstract history of the execution in
Figure~\ref{fig:settrace}. Edges implied by transitivity are omitted.
\label{fig:setcommit}}
\end{figure}

In Figure~\ref{fig:set}, we illustrate the subtleties with the help of a state
of the set, which is a result of executing the trace from
Figure~\ref{fig:settrace}, assuming that values 1, 2 and 4 have been initially
inserted in sequence by performing ``Ins(1)'', ``Ins(2)'' and ``Ins(4)''. 
We consider the following scenario. First, ``Con(2)'' and ``Con(3)'' start
traversing through the list and get preempted when they reach the node
containing $1$, which we denote by $\node_1$. Then the operations are finished
in the order depicted in Figure~\ref{fig:settrace}. Note that ``Con(2)'' returns
$\TRUE$ even though the node containing $2$ is removed from the data structure
by the time the {\tt contains} operation locates it. This surprising behaviour
occurs due to the values $1$ and $2$ being on the same branch of marked nodes in
the list, which makes it possible for ``Con(2)'' to resume traversing from
$\node_1$ and find $2$. On the other hand, ``Con(3)'' cannot find $3$ by
traversing the nodes from $\node_1$: the {\tt contains} operation will reach the
node $\node_2$ and return $\FALSE$, even though $3$ has been concurrently
inserted into the set by this time. Such behaviour is correct, since it can be
justified by a linearization [``Ins(1)'', ``Ins(2)'', ``Ins(4)'', ``Rem(1)'',
``Con(2): true'', ``Rem(2)'', ``Con(3): false'', ``Ins(3)'']. Intuitively, such
linearization order is possible, because pairs of events (``Con(2): true'',
``Rem(2)'') and (``Con(3): false'', ``Ins(3)'') overlap in the execution.

Building a correct linearization order by identifying a linearization point of
{\tt contains} is complex, since it depends on presence of concurrent {\tt
insert} and {\tt remove} operation as well as on current position in the
traversal of the data structure. We demonstrate a different approach to the
proof of the Optimistic Set based on the following insights. Firstly, we observe
that only decisions about a relative order of operations with the same argument
need to be committed into the abstract history, since linearizability w.r.t. the
sequential specification of a set does not require enforcing any additional
order on concurrent operations with different arguments. Secondly, we postpone
decisions about ordering {\tt contains} operations w.r.t. concurrent events till
their return values are determined. Thus, in the abstract history for
Figure~\ref{fig:settrace}, ``Con(2): true'' and ``Rem(2)'' remain unordered
until the former encounters the node removed by the latter, and the order
between operations becomes clear. Intuitively, we construct a linear order on
completed events with the same argument, and let {\tt contains} operations
be inserted in a certain place in that order rather than appended to it.

\newparagraph{Preliminaries.}
We assume that a set $\nodeType$ is a set of pointers to nodes, and that the
state of the linked list is represented by a partial map $\nodeType
\rightharpoonup \nodeType \times {\sf Int} \times {\sf Bool}$. To aid in
constructing the abstract history $(\events, \order)$, the code maintains a
piece of ghost state---a partial function $\ghostnode : \eidType
\rightharpoonup \nodeType$. Given the identifier $\eid$ of
an event $\events(\eid)$ denoting an ${\tt insert}$ that has inserted its value
into the set, $\ghostnode(\eid)$ returns a node identifier (a pointer) of that
value in the data structure. Similarly, for a successful remove event identifier
$\eid$, $\ghostnode(\eid)$ returns a node identifier that the corresponding
operation removed from the data structure.

\newparagraph{Commitment points.}
The commitment points in the ${\tt insert}$ and ${\tt remove}$ operations are
denoted by ghost code in Figure~\ref{fig:setghostcode1}. They are similar in
structure and update the order of events in the abstract history in the same way
described by ${\tt OrderInsRem}$. That is, these commitment points maintain a
linear order on completed events of operations with the same argument: on the
first line of ${\tt OrderInsRem}$, the current {\tt insert}/{\tt remove} event
identified by $\myEid$ gets ordered after each operation $e$ with the same
argument as $\myEid$. On the second line of ${\tt OrderInsRem}$, uncompleted
{\tt insert} and {\tt remove} events with the same argument are ordered after
$\myEid$. Note that uncompleted {\tt contains} events remain unordered w.r.t.
$\myEid$, so that later on at the commitment point of {\tt contains} they could
be ordered before the current {\tt insert} or {\tt remove} operation (depending
on whether they return $\FALSE$ or $\TRUE$ accordingly), if it is necessary.

\begin{figure}[t]
\begin{lstlisting}[numbers=none]
OrderInsRem() {
  #\colorbox{gray}{
    \begin{tabular}{@{}l@{}}
    $\order := (\order
          \cup \{ (e, \myEid) \mid
                e \in \ids{\complete{\events}}
                \land \argsof{\events(e)} = \argsof{\events(\myEid)} \})^+;$
    \\
    $\order := (\order
          \cup \{ (\myEid, e) \mid
              \begin{array}[t]{@{}l@{}}
              e \in \ids{\incomplete{\events}}
              \land \argsof{\events(e)} = \argsof{\myEid}
              \\ \hfill
              {} \land \opof{\events(e)} \neq \containsOp \})^+;
              \end{array}
      $
    \end{tabular}
    }
  #
}
\end{lstlisting}

\begin{minipage}{0.45\textwidth}
\begin{lstlisting}[numbers=none]
#${\tt commit}_{\tt insert}$#() {
  #\colorbox{gray}{
    \begin{tabular}{@{}l@{}}
    $\resof{\events(\myEid)} := ({\tt c.val} \neq {\tt v})$;\\
    if (${\tt c.val} \neq {\tt v}$)\\
    \ \ $\ghostnode[\myEid] := c$;\\
    OrderInsRem();
    \end{tabular}
    }
  #
}
\end{lstlisting}
\end{minipage}
\hspace{10pt}
\begin{minipage}{0.45\textwidth}
\begin{lstlisting}[numbers=none]
#${\tt commit}_{\tt remove}$#() {
  #\colorbox{gray}{
    \begin{tabular}{@{}l@{}}
    $\resof{\events(\myEid)} := ({\tt c.val} = {\tt v})$;\\
    if (${\tt c.val} = {\tt v}$)\\
    \ \ $\ghostnode[\myEid] := c$;\\
    OrderInsRem();
    \end{tabular}
    }
    #
}
\end{lstlisting}
\end{minipage}
\caption{
  The auxiliary code executed at the commitment points of {\tt insert} and {\tt remove}
  \label{fig:setghostcode1}
}
\end{figure}

\begin{figure}

\begin{lstlisting}[numbers=none]
#${\tt commit}_{\tt contains}$#() {
  #\colorbox{gray}{
  \begin{tabular}{@{}l@{}}
    $\resof{\events(\myEid)}$ := (curr.val = v);
    \\
    $\eidType$ \ttafter := if (curr.val = v) 
    \begin{tabular}[t]{l}
    then $\insertOf(\events, {\tt curr})$\\
    else $\lastremove(\events, \order, {\tt v})$;
    \end{tabular}
    \\
    if ($\ttafter \neq \bot$)
    \\
    \ \ $\order := (\order \cup \{ (\ttafter, \myEid) \})^+;$
    \\
    $\order := (\order \cup
      \{ (\myEid, \eid) \mid
            \nedge{\eid}{\order}{\myEid} \land
            \argsof{\events(\eid)} = \argsof{\events(\myEid)}
            \})^+;$
  \end{tabular}
  }
  #
}
\end{lstlisting}
where for an abstract history $(\events, \order)$, a node identifier $n$ and a value $v$:
\[
\begin{array}{@{}c@{}}
\insertOf(\events, n) =
\begin{cases}
\eid, \mbox{if } \ghostnode(\eid) = n, \opof{\events(\eid)} = \insertOp \mbox{ and } \resof{\events(\eid)} = \TRUE
\\
\mbox{undefined otherwise}
\end{cases}
\\[15pt]
\lastremove(\events, \order, v) =
\begin{cases}
\eid, & \mbox{if }
  \events(\eid) = (\_, \removeOp, v, \TRUE)
\\ & {}\land
(\forall \eid' \ldotp \opof{\events(\eid)} = \removeOp
  \land \argsof{\events(\eid')} = v \implies{}
\\ & \hfill
\edge{\eid'}{\order}{\eid})
\\
\bot, & \mbox{if }
\lnot\exists \eid \ldotp \events(\eid) = (\_, \removeOp, v, \TRUE)
\end{cases}
\end{array}
\]

\vspace{-10pt}
\caption{
  The auxiliary code executed at the commitment point of {\tt contains}
  \label{fig:setghostcode2}
}
\end{figure}

At the commitment point, the {\tt remove} operation assigns a return value to
the corresponding event. When the removal is successful, the commitment point
associates the removed node with the event by updating $\ghostnode$. Let us
illustrate how ${\tt commit}_{\tt remove}$ changes abstract histories on the
example. For the execution in Figure~\ref{fig:settrace}, after starting the
operation ``Rem(2)'' we have the abstract history Figure~\ref{fig:setcommit}(a),
and then at point (B) ``Rem(2)'' changes the history to
Figure~\ref{fig:setcommit}(b). The uncompleted event ``Con(2)'' remains
unordered w.r.t. ``Rem(2)'' until it determines its return value ($\TRUE$) later
on in the execution, at which point it gets ordered before ``Rem(2)''.

At the commitment point, the {\tt insert} operation assigns a return value to
the event based on the check ${\tt c.val} \neq {\tt v}$ determining whether
${\tt v}$ is already in the set. In the execution Figure~\ref{fig:settrace},
prior to the start of ``Ins(3)'' we have the abstract history
Figure~\ref{fig:setcommit}(c). When the event starts, a new event is added into
the history (commitment point of kind 1), which changes it to
Figure~\ref{fig:setcommit}(d). At point (D) in the execution, ${\tt commit}_{\tt
insert}$ takes place, and the history is updated to
Figure~\ref{fig:setcommit}(e). Note that ``Ins(3)'' and ``Con(3)'' remain
unordered until the latter determines its return value ($\FALSE$) and orders
itself before ``Ins(3)'' in the abstract history.


The commitment point at
lines~\ref{line:locatecommit_start}--\ref{line:locatecommit_end} of the {\tt
contains} operation occurs at the last iteration of the sorted list traversal in
the {\tt locate} method. The last iteration takes place when ${\tt curr.val}
\geq v$ holds. In Figure~\ref{fig:setghostcode2}, we present the auxiliary code
${\tt commit}_{\tt contains}$ executed at line~\ref{line:locatecommit_end} in
this case. Depending on whether a requested value is found or not, the abstract
history is updated differently, so we further explain the two cases separately.
In both cases, the {\tt contains} operation determines which event in the
history it should immediately follow in all linearizations.

\smallskip
\noindent\textbf{Case (i).} If {\tt curr.val = v}, the requested value {\tt v}
is found, so the current event $\myEid$ receives $\TRUE$ as its return value.
In this case, ${\tt commit}_{\tt contains}$ adds two kinds of edges in the
abstract history.
\begin{itemize}
\item Firstly, $(\insertOf(\events, {\tt curr}), \myEid)$ is added to ensure
that $\myEid$ occurs in all linearizations of the abstract history after the
{\tt insert} event of the node ${\tt curr}$.
\item Secondly, $(\myEid, \eid)$ is added for every other identifier $\eid$ of
an event that does not precede $\myEid$ and has an argument $v$. The requirement
not to precede $\myEid$ is explained by the following. Even though at commitment
points of {\tt insert} and {\tt remove} operations we never order events w.r.t.
{\tt contains} events, there still may be events preceding $\myEid$ in real-time
order. Consequently, it may be impossible to order $\myEid$ immediately after
$\insertOf(\events, {\tt curr})$.
\end{itemize}
At point (C) in the example from Figure~\ref{fig:settrace}, ${\tt commit}_{\tt
contains}$ in ``Con(2)'' changes the history from Figure~\ref{fig:setcommit}(b)
to Figure~\ref{fig:setcommit}(c). To this end, ``Con(2)'' is completed with a
return value $\TRUE$ and gets ordered after ``Ins(2)'' (this edge happened to be
already in the abstract history due to the real-time order), and also in front
of events following ``Ins(2)'', but not preceding ``Con(2)''. This does not
include ``Ins(4)'' due to the real-time ordering, but includes ``Rem(2)'', so
the latter is ordered after the {\tt contains} event, and all linearizations of
the abstract history Figure~\ref{fig:setcommit}(c) meet the sequential
specification in this example. In general case, we also need to show that
successful {\tt remove} events do not occur between $\insertOf(\events, {\tt
curr}), \myEid)$ and $\myEid$ in the resulting abstract history, which we
establish formally in \cite{ext}. Intuitively, when $\myEid$ returns $\TRUE$,
all successful removes after $\insertOf(\events, {\tt curr})$ are concurrent
with $\myEid$: if they preceded $\myEid$ in the real-time order, it would be
impossible for the {\tt contains} operation to reach the removed node by
starting from the head of the list in order return $\TRUE$.

\smallskip
\noindent\textbf{Case (ii).} Prior to executing ${\tt commit}_{\tt contains}$,
at line~\ref{line:locatecommit_start} we check that ${\tt curr.val} \geq {\tt
v}$. Thus, if {\tt curr.val = v} does not hold in ${\tt commit}_{\tt contains}$,
the requested value {\tt v} is not found in the sorted list, and $\FALSE$
becomes the return value of the current event $\myEid$. In this case, ${\tt
commit}_{\tt contains}$ adds two kinds of edges in the abstract history.
\begin{itemize}
\item Firstly, $(\lastremove(\events, \order, {\tt v}), \myEid)$ is added, when
there are successful remove events of value $v$ (note that they are linearly
ordered by construction of the abstract history, so we can choose the last of
them). This ensures that $\myEid$ occurs after a successful remove event in all
linearizations of the abstract history.
\item Secondly, $(\myEid, \eid)$ is added for every other identifier $\eid$ of
an event that does not precede $\myEid$ and has an argument $v$, which is
analogous to the case (i).
\end{itemize}
Intuitively, if $v$ has never been removed from the set, $\myEid$ needs to
happen in the beginning of the abstract history and does not need to be ordered
after any event.

For example, at point (D) in the execution from Figure~\ref{fig:settrace}, ${\tt
commit}_{\tt contains}$ changes the abstract history from
Figure~\ref{fig:setcommit}(e) to Figure~\ref{fig:setcommit}(f). To this end,
``Con(3)'' is ordered in front of all events with argument $3$ (specifically,
``Ins(3)''), since there are no successful removes of $3$ in the abstract
history. Analogously to the case (i), in general to ensure that all
linearizations of the resulting abstract history meet the sequential
specification, we need to show that there cannot be any successful {\tt insert}
events of $v$ between $\lastremove(\events, \order, {\tt v})$ (or the beginning
of the abstract history, if it is undefined) and $\myEid$. We prove this
formally in \cite{ext}. Intuitively, when $\myEid$ returns
$\FALSE$, all successful {\tt insert} events after $\lastremove(\events, \order,
{\tt v})$ (or the beginning of the history) are concurrent with $\myEid$: if
they preceded $\myEid$ in the real-time order, the inserted nodes would be
possible to reach by starting from the head of the list, in which case the {\tt
contains} operation could not possibly return $\FALSE$.







%% file: related.tex

\section{Related Work}


\vspace{-1pt}
There has been a great deal of work on proving algorithms linearizable;
see~\cite{DongolD14} for a broad survey. However, despite a large number of
techniques, often supported by novel mathematical theory, it remains the case
that all but the simplest algorithms are difficult to verify. Our aim is to
verify the most complex kind of linearizable algorithms, those where the
linearization of a set of operations cannot be determined solely by examining
the prefix of the program execution consisting of these operations. Furthermore,
we aim to do this while maintaining a relatively simple proof argument.


Much work on proving linearizability is based on different kinds of
\emph{simulation proofs}. Loosely speaking, in this approach the linearization
of an execution is built incrementally by considering either its prefixes or
suffixes (respectively known as \emph{forward} and \emph{backward}
simulations). This supports inductive proofs of linearizability: the proof
involves showing that the execution and its linearization stay in correspondence
under forward or backward program steps. The linearization point method is an
instance of forward simulation: a syntactic point in the code of an operation is
used to determine when to add it to the linearization.


As we explained in \S\ref{sec:intro}, forward simulation alone is not sufficient
in general to verify linearizability. However, Schellhorn et
al.~\cite{schellhorn} prove that backward simulation alone \emph{is} always
sufficient.
They also present a proof technique 
and use it to verify the Herlihy-Wing queue~\cite{linearizability}. However,
backwards simulation proofs are difficult to understand intuitively:
programs execute forwards in time, and therefore it is much more natural to
reason this way.

The queue originally proposed by Herlihy and Wing in their paper on
linearizability~\cite{linearizability} has proved very difficult to verify.
Their proof sketch is based on reasoning about the possible linearizations
arising from a given queue configuration. 
Our method could be seen as being midway between this approach and linearization
points. We use partiality in the abstract history to represent sets of possible
linearizations, which helps us simplify the proof by omitting irrelevant
ordering (\S\ref{sec:linearizability}).


Another class of approach to proving linearizability is based on special-purpose
program logics. These can be seen as a kind of forward simulation: assertions in
the proof represent the connection between program execution and its
linearization. To get around the incompleteness of forward simulation, several
authors have introduced auxiliary notions that support limited reasoning about
future behaviour in the execution, and thus allow the proof to decide the order
of operations in the linearization~\cite{rgsep,rgsim-lin,turon-popl13}.
However, these new constructs have subtle semantics, which results in proofs
that are difficult to understand intuitively.



Our approach is based on program logic, and therefore is a kind of forward
simulation. The difference between us and previous program logics is that we do
not explicitly construct a linear order on operations, but only a partial order.
This removes the need for special constructs for reasoning about future
behaviour, but creates the obligation to show that the partially ordered
abstract history can always be linearized.

One related approach to ours is that of Hemed et al.~\cite{cal}, who
generalise linearizability to data structures with concurrent specifications
(such as barriers) and propose a proof method for establishing it. To this end,
they also consider histories where some events are partially ordered---such
events are meant to happen concurrently. However, the goal of Hemed et al.'s
work is different from ours: their abstract histories are never linearized, to
allow concurrent specifications; in contrast, we guarantee the existence of a
linearization consistent with a sequential specification. It is likely that the
two approaches can be naturally combined.




Aspect proofs~\cite{henzinger-concur13} are a non-simulation approach that is
related to our work. An aspect proof imposes a set of forbidden shapes on the
real-time order on methods; if an algorithm avoids these shapes, then it is
necessarily linearizable. These shapes are specific to a particular
data structure, and indeed the method as proposed in~\cite{henzinger-concur13}
is limited to queues (extended to stacks in~\cite{dodds-popl15}). In contrast,
our proof method is generic, not tied to a particular kind of data
structure. Furthermore, checking the absence of forbidden shapes in the aspect
method requires global reasoning about the whole program execution, whereas our
approach supports inductive proofs. The original proof of the TS stack used an
extended version of the aspect approach~\cite{dodds-popl15}. However, without a
way of reasoning inductively about programs, the proof of correctness reduced to
a large case-split on possible executions. This made the proof involved and
difficult. Our proof is based on an inductive argument, which makes it easier.

Another class of algorithms that are challenging to verify are those that use
\emph{helping}, where operations complete each others' work. In such algorithms,
an operation's position in the linearization order may be fixed by a helper
method. Our approach can also naturally reason about this pattern: the helper
operation may modify the abstract history to mark the event of the operation
being helped as completed.

The Optimistic set was also proven linearizable by O'Hearn et al. in
\cite{hindsight}. The essence of the work is a collection of lemmas (including
the Hindsight Lemma) proven outside of the logic to justify conclusions about
properties of the past of executions based on the current state. Based on our
case study of the Optimistic set algorithm, we conjecture that at commitment
points we make a constructive decision about extending abstract history where
the hindsight proof would use the Hindsight Lemma to non-constructively extend a
linearization with the {\tt contains} operation.

\section{Conclusion and Future Work}

\vspace{-1pt}
The popular approach to proving linearizability is to construct a total
linearization order by appending new operations as the program executes. This
approach is straightforward, but is limited in the range of algorithms it can
handle. In this paper, we present a new approach which lifts these
limitations, while preserving the appealing incremental proof structure of
traditional linearization points. As with linearization points, our
fundamental idea can be explained simply: at commitment points, operations
impose order between themselves and other operations, and all linearizations of
the order must satisfy the sequential specification. Nonetheless, our
technique generalises to far more subtle algorithms than traditional
linearization points. 

We have applied our approach to two algorithms known to present particular
problems for linearization points. Although, we have not presented it here, our
approach scales naturally to \emph{helping}, where an operation is completed by
another thread. We can support this, by letting any thread complete the
operation in an abstract history. In future work, we plan to apply our approach
to the Time-Stamped stack \cite{dodds-popl15}, which poses verification
challenges similar to the TS queue; a flat-combining style algorithm, which
depends fundamentally on helping, as well as a range of other challenging
algorithms. In this paper we have concentrated on simplifying manual proofs.
However, our approach also seems like a promising candidate for automation, as
it requires no special meta-theory, just reasoning about partial orders. We are
hopeful that we can automate such arguments using off-the-shelf solvers such as
Z3, and we plan to experiment with this in future.




%% file: appendix.tex

\newcommand{\added}{{\rm added}}
\newcommand{\assert}[2]{{#1} \models {#2}}
\newcommand{\rf}{{\sf rf}}

\newcommand{\dependson}[1]{
}

\clearpage
\appendix
\def\addcontentsline#1#2#3{\oldaddcontentsline{#1}{#2}{#3}}
\setcounter{tocdepth}{2}
\renewcommand{\contentsname}{Table of Annexes}
\tableofcontents
\input{app-proglang}

\input{app-logic}
\clearpage
\input{app-tsq}
\clearpage
\input{app-set}
\clearpage
\input{app-hwq}

%% file: app-proglang.tex

\newcommand{\hsemn}[2]{\mathcal{H}_{#1}(#2)}

\newcommand{\cont}{c}
\newcommand{\contType}{{\sf Cont}}

\newcommand{\acont}{\mathbbm{c}}
\newcommand{\acontType}{{\sf ACont}}

\newcommand{\astateType}{{\sf SState}}

\DeclarePairedDelimiter{\doublep}{\llparenthesis}{\rrparenthesis}

\newcommand{\trans}[3]{{#1} \mathrel{{\rightarrowtail}_{#2}} {#3}}

\newcommand{\id}{{\sf id}}

\section{Syntax and Semantics of Data Structure Operations}\label{app:proglang}


\ag{Check if we need formal constraints on thread-locality of arg and res locations.}

\newparagraph{Operation syntax}
Data structures implement every operation $\op \in \opType$ as {\em sequential
commands} with the following syntax:
\[
\com \in \comType ::= \pcom \mid \com \seqc \com \mid \com + \com
\mid \com^* \mid \cskip,\quad \mbox{where } \pcom \in \pcomType
\]
The grammar includes {\em primitive commands} $\pcom$ from a set $\pcomType$,
sequential composition $\com \seqc \com$, non-deterministic choice $\com
+ \com$, finite iteration $\com^*$ (we are interested only in
terminating executions) and a termination marker $\cskip$. We use $+$ and
$(\ )^*$ instead of conditionals and while loops for theoretical
simplicity: as we show further, given appropriate primitive commands
conditionals and loops can be encoded.


\newparagraph{Operations semantics} Assuming a set ${\sf Loc} \subseteq
\valueType$ of memory locations, we let $\stateType = {\sf Loc} \to \valueType$
denote the set of all possible states of an implementation and let $\state$
range over them. States are shared among threads from $\tidType$.


We assume that for every thread $\tid$ there are locations $\larg[\tid],
\lres[\tid] \in \dom{\stateType}$, which are used only by a thread $\tid$ for
storing an argument of a data structure operation and for returning its result
correspondingly.

We assume that the semantics of each primitive command $\pcom$ is given by a
non-deterministic {\em state transformer} $\db{\pcom}_\tid : \stateType \to
\powerset{\stateType}$, where $\tid \in \tidType$. For a state $\state$, the set
of states $\intp{\tid}{\pcom}{\state}$ is the set of possible states resulting
from executing $\pcom$ atomically in a state $\state$ by a thread $\tid$. We
also assume a primitive command $\id \in \pcomType$ with the interpretation
$\intp{\tid}{\id}{\state} \triangleq \{\state \}$.

State transformers may have different semantics depending on a thread
identifier, which we use to access thread-local memory locations such as
$\larg[\tid]$ and $\lres[\tid]$ for each thread $\tid$.


Figure~\ref{fig:opsem} gives selected rules of operational semantics;
$\sttrans{\com}{\state}{\tid}{\com'}{\state'}$ indicates a transition
from $\com$ to $\com'$ by performing a primitive command $\pcom$ in a thread
$\tid$ that updates the state from $\state$ to $\state'$. The rules of the
operational semantics are standard.




\begin{figure}[t]
${\rightarrowtail}
  \subseteq \comType
  \times \pcomType
  \times \comType:$

{\centering
$\begin{array}{l l l}
\myfrac{
  \trans{\com_1}{\pcom}{\com'_1}
}{
  \trans{\com'_1 \seqc \com_2}{\pcom}{\com'_1 \seqc \com_2}
} &
\myfrac{
  \,
}{
  \trans{\com^*}{\id}{\com; \com^*}
} &
\myfrac{
  \,
}{
  \trans{\pcom}{\pcom}{\cskip}
} \\
\myfrac{
  \,
}{
  \trans{\cskip \seqc \com}{\id}{\com}
} &
\myfrac{
  \,
}{
  \trans{\com^*}{\id}{\cskip}
} &
\myfrac{
  i \in \{1, 2\}
}{
  \trans{\com_1 + \com_2}{\id}{\com_i}
}
\end{array}$\par}

\medskip
${\longrightarrow}
  \subseteq (\comType \times \stateType)
  \times \tidType
  \times (\comType \times \stateType):$

{\centering
$\myfrac{
  \state' \in \intp{\tid}{\pcom}{\state} \quad \trans{\com}{\pcom}{\com'}
}{
  \sttrans{\com}{\state}{\tid}{\com'}{\state'}
}$\par}

\caption{The operational semantics of sequential commands
\label{fig:opsem}}
\end{figure}

Let us show how to define traditional control flow primitives, such as an
if-statement and a while-loop, in our programming language. Assuming a language
for arithmetic expressions, ranged over by $\mathcal{E}$, and a function
$\db{\mathcal{E}}_\state$ that evaluates expressions in a given state $\state$,
we define a primitive command ${\tt assume}(\mathcal{E})$ that acts as a filter
on states, choosing only those where $\mathcal{E}$ evaluates to non-zero
values.
\[
\intp{\tid}{{\tt assume}(\mathcal{E})}{\state} \triangleq
  (\mbox{if } \db{\mathcal{E}}_\state \neq 0
  \mbox{ then } \{\state\}
  \mbox{ else } \emptyset).
\]
Using ${\tt assume}(\mathcal{E})$ and the C-style negation $!\mathcal{E}$ in expressions, a conditional
and a while-loop can be implemented as the following commands:
\begin{gather*}
  {\tt if}\ \mathcal{E}\ {\tt then}\ C_1\ {\tt else}\ C_2 \triangleq ({\tt assume}(\mathcal{E}); C_1) + ({\tt assume}(!\mathcal{E}); C_2) \\
  {\tt while}\ \mathcal{E}\ {\tt do}\ C \triangleq ({\tt assume}(\mathcal{E}); C)^*; {\tt assume}(!\mathcal{E})
\end{gather*}

\newparagraph{Data structure histories}
We now define the set of histories produced by a data structure implementation
$\cdt$, which is required by the definition of linearizability
(Definition~\ref{dfn:linearizability}, \S\ref{sec:linearizability}). Informally,
these are the histories produced by threads repeatedly invoking data structure
operations in any order and with any possible arguments (this can be thought of
as running the data structure implementation under its {\em most general
  client}~\cite{linown}). We define this formally using a concurrent small-step
semantics of the data structure $D$ that also constructs corresponding
histories: ${\twoheadrightarrow}_{\cdt} \subseteq (\tpType \times \stateType
\times \historyType)^2$, where $\tpType = \tidType \to (\comType \uplus
\{\idle\})$. Here a function $\tp \in \tpType$ characterises the progress of an
operation execution in each thread $\tid$: $\tp(\tid)$ gives the continuation of
the code of the operation executing in thread $\tid$, or $\idle$ if no operation
is executing. The relation ${\twoheadrightarrow}_{\cdt}$ defines how a step of
an operation in some thread transforms the data structure state and the history:
\[
\begin{array}{@{}c@{}}
\myfrac{
	\eid \notin \ids{\events}
	\quad \args \in \valueType
  	\quad \events' = \events \sub{\eid}{(\tid, \op, \args, \Todo)}
  	\quad \order' = \order \cup \{ (\eidp, \eid) \mid \eidp \in \complete{\events} \}
}{
  \htrans{\tp \sub{\tid}{\idle}, \state, (\events, \order)}
         {\cdt}
         {\tp \sub{\tid}{\cdt(\op)}, \state \sub{\larg[\tid]}{\args}, (\events', \order')}
}
\\[10pt]
\myfrac{
	\sttrans{\com}{\state}{\tid}{\com'}{\state'}
}{
  \htrans{\tp \sub{\tid}{\com}, \state, (\events, \order)}
         {\cdt} 
         {\tp \sub{\tid}{\com'}, \state', (\events, \order)}
}
\\[10pt]
\myfrac{
	\eid = {\sf last}(\tid, (\events, \order))
	\quad \events(\eid) = (\tid, \op, \args, \Todo)
	\quad \events' = \events \sub{\eid}{(\tid, \op, \args, \state(\lres[\tid]))}
}{
  \htrans{\tp \sub{\tid}{\cskip}, \state, (\events, \order)}
         {\cdt}
         {\tp \sub{\tid}{\idle}, \state, (\events', \order)}
}
\end{array}
\]
First, an idle thread $\tid$ may call any operation $\op \in \opType$ with any
argument $\args$. This sets the continuation of thread $\tid$ to $\cdt(\op)$,
stores $\args$ into $\larg[\tid]$ and adds a new event $i$ to the history,
ordered after all completed events. Second, a thread $\tid$ executing an
operation may do a transition allowed by the sequential semantics of the
operation's implementation. Finally, when a thread $\tid$ finishes executing an
operation, as denoted by a continuation $\cskip$, the corresponding event is
completed with the return value in $\lres[\tid]$. The identifier ${\sf last}(t,
(E, R))$ of this event is determined as the last one in $E$ by thread $t$
according to $R$: as per Definition~\ref{def:history}, events by each thread are
totally ordered in a history, ensuring that ${\sf last}(\tid, \history)$ is
well-defined.
\[
{\sf last}(\tid, (\events, \order)) \triangleq
\begin{cases}
\eid, & \mbox{ such that } \tidof{\events(\eid)} = \tid \\
 & \mbox{ and } (\forall j \ldotp j \neq i \land \tidof{\events(j)} = \tid \implies \edge{j}{\order}{i})\\
\bot, & \mbox{ if } \forall \eid \ldotp \tidof{\events(\eid)} \neq \tid
\end{cases}
\]


\newparagraph{Specification histories} We assume a set of specification states $\astateType$, ranged over by $\astate$, and a specification $\adt$ of a data
structure that interprets every operation $\op \in \opType$ as a sequential
state transformer $\doublep{\op}_\adt : (\astateType \times \valueType) \to
\powerset{\astateType \times \valueType}$. When $(\astate', r) \in
\doublep{\op}_\adt(\astate, a)$, we say that sequential execution of $\op$ with
an argument $a$ leads to a state $\astate'$ with a return value $r$.

We generate all sets of histories of a specification $\adt$ starting from an
initial state $\astate_0$ as follows:
\[
\hsem{\adt, \astate_0} \triangleq \{
  \history \mid
  \isSeq(\history) \land
  \mhtrans{(\astate_0, (\emptyset, \emptyset)}
         {\adt}
         {\_, \_, \history} \},
\]
where ${\twoheadrightarrow}_{\adt}$ is a relation constraining a single step in
the generation of sequential histories (similarly to
${\twoheadrightarrow}_{\cdt}$ from Section~\ref{sec:lang}):
\[
\myfrac{
  \begin{array}{c}
  \eid \notin \ids{\events}
  \quad \args \in \valueType
  \quad \tid \in \tidType
  \quad (\astate', r) \in \doublep{\op}_\adt(\astate, a)
  \\
    \events' = \events \sub{\eid}{(\tid, \op, \args, r)}
    \quad \order' = \order \cup \{ (\eidp, \eid) \mid \eidp \in \complete{\events} \}
  \end{array}
}{
  \htrans{\astate, (\events, \order)}
         {\adt}
         {\astate', (\events', \order')}
}
\]

Having generated all histories $\histories = \hsem{\adt, \astate_0}$ of a data
structure specification $\adt$, we can use Theorem~\ref{thm:soundness} to
conclude that $\hsem{\cdt, \state_0} \sqsubseteq \hsem{\adt, \astate_0}$.

\artem{TODO specification histories}

%% file: app-logic.tex
\newcommand{\bigland}[2]{\bigwedge_{#1}\,({#2})}
\newcommand{\safe}[5]{{\sf safe}_{#1}(#2, #3, #4, #5)}
\newcommand{\ssw}{{\sf ssw}}

\newcommand{\continv}{{\sf cinv}}
\newcommand{\isidle}{{\rm isIdle}}

\newcommand{\nhtrans}[4]{
\langle {#2} \rangle
\mathrel{{\twoheadrightarrow}^{#1}_{#3}}
\langle {#4} \rangle
}

\section{Logic}\label{app:logic}

Assertions $\lP, \lQ \in \assnType$ are described with the following grammar:
\[
\begin{array}{lcl}
  E, F & ::= & a \mid X \mid E+F \mid \dots, \quad \mbox{where } X \in {\sf LVars}, a \in \valueType \\
  p \in {\sf Pred} & ::= & E = F \mid E \mapsto F \mid ... \\
  \lP, \lQ \in {\sf Assn} & ::= &
    p \mid \lP \lor \lQ \mid \lP \land \lQ \mid \lP \Rightarrow \lQ \mid \exists X \ldotp \lP
    \mid \forall X \ldotp \lP
\end{array}
\]
Thus, assertions from $\assnType$ contain the standard logical connectives, and
among them the existential and universal quantification over logical variables
$X$, ranging over a set $\lvarsType$. We assume a set ${\sf Pred}$ of
predicates, which includes a predicate $E \mapsto F$ denoting a concrete state
that describes a singleton heap. We also assume a function $\lint : \lvarsType
\rightharpoonup \valueType$ denoting an interpretation of logical variables.

Formulas ${\sf Assn}$ denote sets of concrete states, histories, ghost states and interpretation of logical variables as defined by a satisfaction
relation $\models$ in Figure~\ref{fig:assnSem}. Additionally, we define a function $\evalf{-}{-} : {\sf Assn} \times ({\sf Expr} \times {\sf LVars} \to \valueType) \to \powerset{\stateType \times \historyType \times \ghostType}$:
\[
\evalf{\lP}{\lint} = \{ (\state, \history, \ghost) \mid \assert{\state, \history, \ghost, \lint}{\lP} \}
\]

The proof rules of the logic are presented in Figure~\ref{fig:proofrules}.

\begin{figure}[t]
$
\begin{array}{ll}
\assert{(\state, \history, \ghost, \lint)}{E \mapsto F},
& \mbox{iff }
  \state(\db{E}_\lint) = \db{F}_\lint
\\
\assert{(\state, \history, \ghost, \lint)}{E = F},
& \mbox{iff }
  \db{E}_\lint = \db{F}_\lint
\\
\assert{(\state, \history, \ghost, \lint)}{\lP \lor \lQ},
& \mbox{iff either }
  \assert{(\state, \history, \ghost, \lint)}{\lP} 
  \mbox{ or }
  \assert{(\state, \history, \ghost, \lint)}{\lQ}
  \mbox{ holds}
\\
\assert{(\state, \history, \ghost, \lint)}{\lP \land \lQ},
& \mbox{iff both }
  \assert{(\state, \history, \ghost, \lint)}{\lP} 
  \mbox{ and }
  \assert{(\state, \history, \ghost, \lint)}{\lQ}
  \mbox{ hold}
\\
\assert{(\state, \history, \ghost, \lint)}{\lP \Rightarrow \lQ},
& \mbox{iff }
  \assert{(\state, \history, \ghost, \lint)}{\lQ} 
  \mbox{ holds, when so does }
  \assert{(\state, \history, \ghost, \lint)}{\lP}
\\
\assert{(\state, \history, \ghost, \lint)}{\exists X \ldotp \lP},
& \mbox{iff there is } a \in \valueType \mbox{ such that }
  \assert{(\state, \history, \ghost, \lint\sub{X}{a})}{\lP} 
  \mbox{ holds}
\\
\assert{(\state, \history, \ghost, \lint)}{\forall X \ldotp \lP},
& \mbox{iff for any } a \in \valueType,
  \assert{(\state, \history, \ghost, \lint\sub{X}{a})}{\lP} 
  \mbox{ holds}
\end{array}
$
\caption{Semantics of the assertion language $\sf Assn$. Here we assume a
function $\evalf{-}{\lint} : {\sf Expr} \times ({\sf LVars} \to \valueType) \to
\valueType$ substitutes logical variables in expressions with their
interpretation and evaluates the result.}
\label{fig:assnSem}
\end{figure}

\begin{figure}[t]
$
\begin{array}{r@{\hspace{2em}}l}
\mbox{\footnotesize\sc(RG-Weaken)}
&
\dfrac{
  \spec{\rely, \guar}{\tid}{\lP}{\com}{\lQ}
  \quad
  \lP' \Rightarrow \lP
  \quad
  \rely' \subseteq \rely
  \quad
  \guar \subseteq \guar'
  \quad
  \lQ \Rightarrow \lQ'
}{
  \spec{\rely', \guar'}{\tid}{\lP'}{\com}{\lQ'}
} 
\\[10pt]
\mbox{\footnotesize\sc(Skip)}
&
\dfrac{
  \,
}{
  \spec{\rely, \guar}{\tid}{\lP}{\cskip}{\lP}
}
\\[10pt]
\mbox{\footnotesize\sc(Seq)}
&
\dfrac{
  \spec{\rely, \guar}{\tid}{\lP}{\com}{\lP'}
  \quad
  \spec{\rely, \guar}{\tid}{\lP'}{\com'}{\lQ}
}{
  \spec{\rely, \guar}{\tid}{\lP}{\com \seqc \com'}{\lQ}
}
\\[10pt]
\mbox{\footnotesize\sc(Choice)}
&
\dfrac{
  \spec{\rely, \guar}{\tid}{\lP}{\com}{\lQ}
  \quad
  \spec{\rely, \guar}{\tid}{\lP}{\com'}{\lQ}
}{
  \spec{\rely, \guar}{\tid}{\lP}{\com + \com'}{\lQ}
}
\\[10pt]
\mbox{\footnotesize\sc(Iter)}
&
\dfrac{
  \spec{\rely, \guar}{\tid}{\lP}{\com}{\lP}
}{
  \spec{\rely, \guar}{\tid}{\lP}{\com^*}{\lP}
}
\end{array}
$
\caption{Proof rules of Rely/Guarantee
}
\label{fig:proofrules}
\end{figure}

\newparagraph{Semantics of Hoare triples} For every set of configurations $p \in
\powerset{\Config}$ and each rely $\rely$, let $\ssw(p, \rely)$ be the {\em
strongest stable weaker} set of configurations:
\[
\ssw(p, \rely) = p \cup \{ \conf' \mid \conf \in p
\land (\conf, \conf') \in \rely \}
\]

\begin{definition}[Safety Judgement]\label{app:def:safe}
 We define ${\sf safe}_\tid$ as the greatest relation such that the following holds whenever $\safe{\tid}{\rely, \guar}{p}{\com}{q}$ does:
  \begin{itemize}
    \item if $\com \neq \cskip$, then
    $\begin{multlined}[t]
      \forall \com', \pcom \ldotp
          \trans{\com}{\pcom}{\com'} \implies
          \exists p' \ldotp {\sf stable}(p', \rely) \\ {}
          \land \semj{\guar}{\tid}{\ssw(p, \rely)}{\pcom}{p'}
          \land \safe{\tid}{\rely, \guar}{p'}{\com'}{q},
    \end{multlined}$
    \item if $\com = \cskip$, then $\ssw(p, \rely) \subseteq q$.
  \end{itemize}
\artem{any need in stabilising postconditions?}
\end{definition}


\begin{lemma}\label{app:lem:safe}
For any $\tid, \lP, \com, \op, \args$ and $\lQ$, if
$\spec{\rely, \guar}{\tid}{\lP}{\com}{\lQ}$ holds then
$\forall \lint \ldotp
\safe{\tid}{\rely, \guar}{\evalf{\lP}{\lint}}{\com}{\evalf{lQ}{\lint}}$
\end{lemma}

\subsection{Proof of Theorem~\ref{thm:soundness}}

For convenience, we further refer to the assumptions of
Theorem~\ref{thm:soundness} as a relation ${\sf safelib}$ defined as follows.
\begin{definition}\label{app:def:safelib}
Given a data structure $\cdt$, its initial state $\state_0 \in \stateType$ and a
set of sequential histories $\histories$, we say that ${\sf safelib}(\cdt,
\adt, \histories)$ holds, if there exists an assertion $\mathcal{I}$ and
relations $\rely_{\tid}, \guar_{\tid} \subseteq \Config^2$ for each $\tid \in
\tidType$ such that:
\begin{enumerate}
\item
  $\exists \ghost_0.\, \forall \lint \ldotp (\state_0, (\emptyset, \emptyset),
  \ghost_0) \in \evalf{I}{\lint}$;
\item 
$\forall \tid, \lint. \, {\sf stable}(\evalf{I}{\lint}, \rely_\tid)$;
\item 
  $\forall \history, \lint \ldotp (\_, \history, \_) \in \evalf{I}{\lint} \implies
  \abs(\history, \histories)$;
\item
$\forall \tid, \op.\, (\spec{\rely_\tid, \guar_\tid}
      {\tid}
      {\begin{array}{@{}c@{}}
      I \land \started_{\mathcal{I}}(\tid, \op)
      \end{array}}
      {\cdt(\op)}
      {\begin{array}{@{}c@{}}
      I \land \finished(\tid, \op)
      \end{array}})$;
\item
$\forall \tid, \tid' \ldotp \tid \neq \tid' \implies
  \guar_{\tid} \cup {\dashrightarrow}_{\tid} \subseteq \rely_{\tid'}$.
\end{enumerate}
\end{definition}


\begin{proposition}\label{app:prop:abs}
If $\mhupd{\history}{\history'}$, then there exists $\history''$ such that $\completes{\history}{\history''} \land \refines{\history''}{\history'}$
\end{proposition}
The proof is straightforward: when $\mhupd{(\events,
\order)}{(\events', \order')}$ holds, $\completes{(\events,
\order)}{(\events', \order)}$ and $\refines{(\events', \order)}{(\events',
\order')}$.

In the following theorem, we prove that the ${\leadsto}$ relation (defined in
\S\ref{sec:logic}) is the correspondence established between a concrete history
of a data structure and a matching abstract history under conditions of
Theorem~\ref{thm:soundness}.

\begin{theorem}\label{app:thm:abs}
Given a data structure $\cdt$, its initial state $\state_0 \in \stateType$ and
  a set of sequential histories $\histories$,
if ${\sf safelib}(\cdt, \adt, \histories)$ holds, then the following is true:
\[
\forall h \ldotp h \in \hsem{\cdt, \state_0} \implies
\exists \history \ldotp \mhupd{h}{\history} \land \abs(\history, \histories)
\]
\end{theorem}

Intuitively, Theorem~\ref{app:thm:abs} describes the main idea of our method:
for every concrete history $h \in \hsem{\cdt, \state_0}$, we build a matching
abstract history $\history$ such that all of its linearizations are sequential
histories from $\histories$.

\newparagraph{Proof of Theorem~\ref{thm:soundness}} We show that
  Theorem~\ref{app:thm:abs} is a corollary of Theorem~\ref{thm:soundness}: given
  a data structure $\cdt$, its initial state $\state_0 \in \stateType$ and a set
  of sequential histories $\histories$, we have $(\cdt, \state_0)$ {\em
  linearizable} with respect to $\histories$ if the following holds:
  \begin{equation}\label{app:eq:abs}
\forall \history_1 \ldotp \history_1 \in \hsem{\cdt, \state_0} \implies
\exists \history \ldotp \mhupd{\history_1}{\history} \land \abs(\history,
\histories)
\end{equation}

Let us consider every history $\history_1 \in \hsem{\cdt, \state_0}$. To
conclude linearizability w.r.t. $\histories$, we need to show the following:
\begin{equation}\label{app:thm:soundness:goal}
\exists \history_2 \in \histories \ldotp \exists \history'_1 \ldotp
  \completes{\history_1}{\history'_1} \land
  \refines{\history'_1}{\history_2}
\end{equation}

According to (\ref{app:eq:abs}), there exists $\history$ such that
$\abs(\history, \histories)$ and $\mhupd{\history_1}{\history}$ both hold. By
Definition~\ref{dfn:abstracthistory}, the former gives the followings:
\begin{equation}\label{app:thm:soundness:subgoal1}
\{\history' \mid \refines{\complete{\history}}{\history'}
  \land \isSeq(\history')\} \subseteq \histories
\end{equation}
Note that the set above is not empty, since $\history$ is acyclic and thus has
at least one linearization $\history'$. Thus, $\history' \in \histories$ holds.

By Proposition~\ref{app:prop:abs}, there exists a history $\history''$ such that
$\completes{history_1}{\history''}$ and $\refines{\history''}{\history}$. It is
easy to see that the following two observations can be made for $\history_1$,
$\history$ and $\history''$: \artem{(lemmatise them?)}
\begin{itemize}
\item since $\refines{\history''}{\history}$ holds, so does
$\refines{\complete{\history''}}{\complete{\history}}$, and
\item since $\completes{\history_1}{\history''}$ holds, so does
$\completes{\history_1}{\complete{\history''}}$.
\end{itemize}
By combining these two observations with (\ref{app:thm:soundness:subgoal1}), we
conclude that there exist $\history_2 = \history'$ and $\history'_1 =
\complete{\history''}$ such that $\completes{\history_1}{\history'_1}$ and
$\refines{\history'_1}{\history_2}$. This concludes the proof of
(\ref{app:thm:soundness:goal}).
\qed

Thus, we reduced the proof of Theorem~\ref{thm:soundness} to proving
Theorem~\ref{app:thm:abs}.

\subsection{Proof of Theorem~\ref{app:thm:abs}}
In this subsection, we first introduce auxiliary definitions used in the proof
of Theorem~\ref{app:thm:abs}, and then present the proof itself.

\begin{definition}
We let $\isidle(\tid) \in \assnType$ be an assertion satisfying the
following:
\[
\evalf{\isidle(\tid)}{\lint}
= \{(\state, \history, \ghost) \mid
  {\sf last}(\tid, \history) = \bot
  \lor {\sf last}(\tid, \history) \in \complete{\history} \}.
\]
\end{definition}
The assertion $\isidle(\tid)$ represents the set of configurations, in
which abstract histories do not have uncompleted events in a thread $\tid$.

\begin{definition}
We let $\continv$ be a relation such that $\continv(\tid, \cont, p_\tid, \rely_\tid, \guar_\tid, \mathcal{I})$ holds whenever the
following is true:
\begin{itemize}
  \item ${\sf stable}(p_\tid, \rely_\tid)$ and $p_\tid \subseteq
  \evalf{\mathcal{I}}{\lint}$ holds;
  \item if $\cont(\tid) = \idle$ then $p_\tid \subseteq
  \evalf{\isidle(\tid)}{\lint}$ holds;
  \item
  if $\cont(\tid) \neq \idle$ then there is $\op$ such that
  $\safe{\tid}
    {\rely_\tid, \guar_\tid}
    {p_\tid}
    {\cont(\tid)}
    {\evalf{\mathcal{I} \land \finished(\tid, \op)}{\lint}}$.
\end{itemize}
\end{definition}



\newparagraph{Proof of Theorem~\ref{app:thm:abs}}
Let us consider a data structure $\cdt$, its initial state $\state_0 \in
\stateType$ and the set of specification histories $\histories$. Let us assume
that ${\sf safelib}(\cdt, \adt, \histories)$ holds. In particular, there exist
$\rely_\tid$, $\guar_\tid$ and $\mathcal{I}$ satisfying the constraints in ${\sf
safelib}(\cdt, \adt, \histories)$. We prove that $(\cdt, \state_0)$ is
linearizable with respect to $\histories$, i.e., $\hsem{\cdt, \state_0}
\sqsubseteq \histories$. To this end, we strengthen the statement of the theorem
as follows:
\begin{multline*}
\forall \cont, \state, h \ldotp 
    \mhtrans{(\lambda \tid \ldotp \idle), \state_0, (\emptyset, \emptyset)}{\cdt}{\cont, \state, h}
    \implies{}
    \\
    \exists \history, \ghost \ldotp \mhupd{h}{\history} \land \abs(\history, \histories) \land{} \\
      (\forall \lint, \tid \ldotp \exists p_\tid \ldotp
      (\state, \history, \ghost) \in \evalf{p_{\tid}}{\lint} \land
      \continv(\lint, \tid, \cont, p_\tid, \rely_\tid, \guar_\tid, \mathcal{I}))
\end{multline*}

The proof is done by induction on the length $n$ of executions in
$\hsem{\cdt, \state_0}$. We define the following formula:
\begin{multline*}
\phi(n) \triangleq
\forall \cont, \state, h \ldotp 
    \nhtrans{n}{(\lambda \tid \ldotp \idle), \state_0, (\emptyset, \emptyset)}
      {\cdt}{\cont, \state, h}
    \implies{}
    \\
    \exists \history, \ghost \ldotp \mhupd{h}{\history} \land{} \\
      (\forall \lint, \tid \ldotp \exists p_\tid \ldotp
      (\state, \history, \ghost) \in p_{\tid} \land
      \continv(\lint, \tid, \cont, p_\tid, \rely_\tid, \guar_\tid, \mathcal{I}))
\end{multline*}
and prove that $\forall n \geq 0 \ldotp \phi(n)$ holds. Note that in $\phi(n)$
we omit the requirement $\abs(\history, \histories)$, since it is implied by the
other requirements. Specifically, when $\continv$ holds of each thread $\tid$,
the configuration $(\state, \history, \ghost)$ satisfies the invariant
$\mathcal{I}$. Consequently, by ${\sf safelib}(\cdt, \state_0, \histories)$,
$\abs(\history, \histories)$ holds.

\noindent\textbf{Base of the induction.}
We need to show that $\phi(n)$ holds when $n=0$. In this case, $\cont = (\lambda
\tid \ldotp \idle)$, $\state = \state_0$ and $h = (\emptyset, \emptyset)$.
According to Definition~\ref{app:def:safelib}.1 of ${\sf safelib}(\cdt, \adt,
\histories)$, there exists a ghost state $\ghost_0$ such that:
\[
\forall \lint \ldotp (\state_0, (\emptyset, \emptyset), \ghost_0) \in \evalf{I}{\lint}
\]
It is easy to see that $\phi(0)$ holds for $\history = (\emptyset, \emptyset)$
and $\ghost = \ghost_0$ when $p_\tid = \mathcal{I}$ for each thread $\tid$.

\noindent\textbf{Induction step.}
We need to show that $\forall n \ldotp \phi(n) \implies \phi(n+1)$. Let us
choose any $n$ and assume that $\phi(n)$ holds. We need to prove that so does
$\phi(n+1)$, i.e., for every $\cont', \state'$ and $h'$ such that
$\nhtrans{n+1}{(\lambda \tid \ldotp \idle), \state_0, (\emptyset,
\emptyset)}{\cdt}{\cont', \state', h'}$ holds, the following is true:
\begin{multline}\label{app:thm:abs:goal}
    \exists \history', \ghost' \ldotp \mhupd{h'}{\history'} \land \abs(\history', \histories) \land{} \\
      (\forall \lint, \tid \ldotp \exists p'_\tid \ldotp
      (\state', \history', \ghost') \in p'_{\tid} \land
      \continv(\lint, \tid, \cont', p'_\tid, \rely_\tid, \guar_\tid, \mathcal{I}))
\end{multline}
When $\nhtrans{n+1}{(\lambda \tid \ldotp \idle), \state_0, (\emptyset,
\emptyset)}{\cdt}{\cont', \state', h'}$, there exist $\cont$, $\state$ and
$\history$ such that:
\[
\nhtrans{n}{(\lambda \tid \ldotp \idle), \state_0, (\emptyset,
\emptyset)}{\cdt}{\cont, \state, h} \land \htrans{\cont, \state,
h}{\cdt}{\cont', \state', h'}
\]
According to the induction hypothesis $\phi(n)$, for $\cont, \state$ and $h$ there exist $\history$ and $\ghost$ such that:
\begin{multline}\label{app:thm:abs:hypo}
\mhupd{h}{\history} \land \abs(\history, \histories) \land{} \\
      (\forall \lint, \tid \ldotp \exists p_\tid \ldotp
      (\state, \history, \ghost) \in \evalf{p_{\tid}}{\lint} \land
      \continv(\lint, \tid, \cont, p_\tid, \rely_\tid, \guar_\tid, \mathcal{I}))
\end{multline}

By definition of the transition relation ${\twoheadrightarrow}_{\cdt}$, a
transition $\htrans{\cont, \state, h}{\cdt}{\cont', \state', h'}$ corresponds to
one of the three cases for a continuation $\cont(T)$ of some thread $T$: an
invokation of an arbitrary new operation in $T$, a return from the current
operation of $T$, or a transition in $T$. We consider each case separately. 

Let $\events'$ and $\order'$ be such that $h' = (\events', \order')$, and let
$\events$ and $\order$ be such that $h = (\events, \order)$.

\noindent\textbf{Case \#1.} There exists a thread $T$ such that $\cont(T) = \idle$, an operation $\op \in \opType$, its arguments $\args$, its event identifier $\eid \notin \ids{\events}$ such that $\htrans{\cont, \state, (\events, \order)}{\cdt}{\cont', \state', (\events', \order')}$ holds and the following is true:
\begin{itemize}
  \item $\cont' = \cont \sub{T}{\cdt(\op)}$,
  \item $\state' = \state \sub{\larg[T]}{\args}$,  
  \item $\events' = \events \sub{\eid}{(T, \op, \args, \Todo)}$,
  \item $\order' = \order \cup \{ (\eidp, \eid) \mid \eidp \in \complete{\events} \}$.
  \end{itemize}
Let us consider any $\lint$. According to (\ref{app:thm:abs:hypo}), for a thread
$T$ there exists $p_T$ such that $\continv(\lint, T, \cont, p_T,
\rely_T, \guar_T, \mathcal{I})$ and $(\state, \history, \ghost) \in p_T$ both
hold.

From the former we learn that $\conf \in p_T \subseteq \evalf{\mathcal{I} \land
\isidle(T)}{\lint}$, since $\cont(T) = \idle$. It is easy to see that under
this condition, for every abstract history $\history'$ such that
$\addevent{T}{\state, \history, \ghost}{\state', \history', \ghost}$ holds, the
following is true:
\begin{itemize}
  \item $\mhupd{h'}{\history'}$ holds, since the new event $\sub{\eid}{(T, \op,
  \args, \Todo)}$ and the new edges added into $h = (\events, \order)$ can also
  be added into $\history$,
  \item $(\state', \history', \ghost) \in \evalf{\mathcal{I} \land \started(T,
  \op)}{\lint}$ holds.
\end{itemize}
According to ${\sf safelib}(\cdt, \state_0, \histories)$, the command $\cdt(\op)$ fulfils the following specification:
\[
\spec{\rely_T, \guar_T}
      {T}
      {\begin{array}{@{}c@{}}
      \mathcal{I} \land \started_{\mathcal{I}}(T, \op)
      \end{array}}
      {\cdt(\op)}
      {\begin{array}{@{}c@{}}
      \mathcal{I} \land \finished(T, \op)
      \end{array}})
\]
Hence, $\safe{T}{\rely_T, \guar_T} {\evalf{\mathcal{I} \land
\started_{\mathcal{I}}(T, \op)}{\lint}} {\cdt(\op)} {\evalf{\mathcal{I} \land
\finished(T, \op)}{\lint}}$ holds by Lemma~\ref{app:lem:safe}. Let $p'_T =
\ssw(\evalf{\mathcal{I} \land \started_{\mathcal{I}}(T, \op)}{\lint}, \rely_T)$.
It is easy to see that $\safe{T}{\rely_T, \guar_T} {p'_T} {\cdt(\op)}
{\evalf{\mathcal{I} \land \finished(T, \op)}{\lint}}$ holds as well. Thus,
$\continv(\lint, T, \cont', p'_T, \rely_T, \guar_T, \mathcal{I})$ holds. For a
thread $T$, we have found $p'_T$ such that:
\[
  (\state', \history', \ghost) \in p'_{T} \land
      \continv(\lint, \tid, \cont, p'_T, \rely_T, \guar_T, \mathcal{I})
\]
Let us consider every thread $\tid \neq T$. By the hypothesis
(\ref{app:thm:abs:hypo}), there exists $p_\tid$ such that $(\state,
\history, \ghost) \in p_{\tid}$ and $\continv(\lint, \tid,
\cont, p_\tid, \rely_\tid, \guar_\tid, \mathcal{I})$ hold. According to the
latter, $p_\tid$ is stable under $\rely_\tid$. By ${\sf safelib}(\cdt, \state_0,
\histories)$, $\rely_\tid$ includes $\dashrightarrow_{T}$, so $p_\tid$ is stable
under $\dashrightarrow_{T}$ as well. Consequently, $(\state', \history', \ghost)
\in p_{\tid}$ holds.

Thus, we have found $\history'$, $\ghost' = \ghost$ and a new set of
configurations $p_T$ for a thread $T$ such that (\ref{app:thm:abs:goal}) holds.

\noindent\textbf{Case \#2.} There exists a thread $T$ such that $\cont(T) = \cskip$, an operation $\op \in \opType$, its arguments $\args$, its event identifier $\eid = {\sf last}(\tid, (\events, \order))$ such that $\htrans{\cont, \state, (\events, \order)}{\cdt}{\cont', \state', (\events', \order')}$ holds and the following is true:
\begin{itemize}
  \item $\cont' = \cont \sub{T}{\idle}$,
  \item $\events(\eid) = (\tid, \op, \args, \Todo)$
  \item $\events' = \events \sub{\eid}{(\tid, \op, \args, \state(\lres[\tid]))}$,
  \item $\state' = \state$ and $\order = \order'$  
  \end{itemize}
Let us consider any $\lint$. According to (\ref{app:thm:abs:hypo}), for a thread
$T$ there exists $p_T$ such that $\continv(\lint, T, \cont, p_T,
\rely_T, \guar_T, \mathcal{I})$ and $(\state, \history, \ghost) \in p_T$ both
hold. From the former we learn that $\safe{T}{\rely_T, \guar_T}{p_T}{\cskip}{\evalf{\mathcal{I} \land \finished(T, \op)}{\lint}}$ holds. Then the following is true:
\[
(\state, \history, \ghost) \in p_T \subseteq \evalf{\mathcal{I} \land \finished(T, \op)}{\lint} \subseteq \evalf{\mathcal{I} \land \isidle(T)}{\lint}
\]
Hence, the abstract history $\history$ does not have uncompleted events in a thread $T$. By the induction hypothesis, $\mhupd{h}{\history}$ holds. Since $h'$ completes the event that is already completed in $\history$, we can conclude that $\mhupd{h'}{\history}$ holds too.

Also, when $p_T \subseteq \evalf{\mathcal{I} \land \isidle(T)}{\lint}$ holds, so
does $\continv(\lint, T, \cont', p_T, \rely_T, \guar_T, \mathcal{I})$. Thus, for
a thread $T$, we have found $p'_T = p_T$ such that:
\[
  (\state', \history, \ghost) \in p'_{T} \land
      \continv(\lint, \tid, \cont', p'_T, \rely_T, \guar_T, \mathcal{I})
\]
Overall, we have found $\history' = \history$ and $\ghost' = \ghost$ such that
(\ref{app:thm:soundness:goal}) holds.

\noindent\textbf{Case \#3.} There exists a thread $T$, a primitive command
$\pcom$, commands $\com$ and $\com'$ such that $\trans{\com}{\pcom}{\com'}$,
$\state' \in \intp{T}{\pcom}{\state}$ and the following is true:
\begin{itemize}
  \item $\cont(T) = \com$ and $\cont' = \cont \sub{T}{\com'}$,
  \item $\events' = \events$ and $\order = \order'$  
  \end{itemize}
Let us consider any $\lint$. According to (\ref{app:thm:abs:hypo}), for a thread
$T$ there exists $p_T$ such that $\continv(\lint, T, \cont, p_T,
\rely_T, \guar_T, \mathcal{I})$ and $(\state, \history, \ghost) \in p_T$ both
hold. According to the former, $\safe{T}{\rely_T, \guar_T}{p_T}{\com}{\evalf{\mathcal{I} \land \finished(T, \_)}{\lint}}$. By Definition~\ref{app:def:safe}, there exists a stable $p'_T$ such that $\semj{\guar_T}{T}{p_T}{\pcom}{p'_T}$ and $\safe{T}{\rely_T, \guar_T}{p'_T}{\com'}{\evalf{\mathcal{I} \land \finished(T, \_)}{\lint}}$. Also, the following holds:
\begin{itemize}
  \item By definition of $\semj{\guar_T}{T}{p_T}{\pcom}{p'_T}$, there exist
  $\history'$ and $\ghost'$ such that $(\state', \history', \ghost') \in p'_T$. Also, $\semj{\guar_T}{T}{p_T}{\pcom}{p'_T}$ implies that $\mhupd{h' = h}{\mhupd{\history}{\history'}}$.
  \item By ${\sf safelib}(\cdt, \state_0, \histories)$,
  $\evalf{\mathcal{I}}{\lint}$ is stable under $\guar_T$, meaning that $p'_T
  \subseteq \evalf{\mathcal{I}}{\lint}$ holds.
\end{itemize}
It is easy to see that $\continv(\lint, T, \cont', p'_T,
\rely_T, \guar_T, \mathcal{I})$ holds. Thus, for a thread
$T$, we have found $p'_T$ such that:
\[
  (\state', \history', \ghost') \in p'_{T} \land
      \continv(\lint, \tid, \cont', p'_T, \rely_T, \guar_T, \mathcal{I})
\]

Let us consider every thread $\tid \neq T$. By the hypothesis
(\ref{app:thm:abs:hypo}), there exists $p_\tid$ such that $(\state,
\history, \ghost) \in p_{\tid}$ and $\continv(\lint, \tid,
\cont, p_\tid, \rely_\tid, \guar_\tid, \mathcal{I})$ hold. According to the
latter, $p_\tid$ is stable under $\guar_T \subseteq \rely_\tid$. Consequently,
$(\state', \history', \ghost') \in p_{\tid}$ holds.

Thus, we have found $\history'$, $\ghost'$ and a new set of
configurations $p'_T$ for a thread $T$ such that (\ref{app:thm:abs:goal}) holds.
\qed


%% file: app-tsq.tex

\newcommand{\fispotcand}{{\rm isPotCand}}
\newcommand{\fnopotcand}{{\rm noPotCand}}

\newcommand{\visible}{{\sf visible}}
\newcommand{\fdisc}{{\rm disc}}

\newcommand{\sq}{{\sf sq}}

\newcommand{\seqrf}[1]{\rf\{{#1}\}}

\newcommand{\invseq}{\inv_{\sf SEQ}}

\section{Proof details for the Time-Stamped Queue}\label{app:tsqdetails}

\begin{figure}[t]
\hfill
\begin{minipage}[t]{0.96\textwidth}
\begin{lstlisting}[numbers=none]
enqueue(Val v) {
  #\lstassert{
    \inv \land \started(t, \enqOp)
  }#
  atomic {
    PoolID node := insert(#$\myTid$#, v);
\end{lstlisting}
\begin{lstlisting}[numbers=none,backgroundcolor=\color{gray}]      
    #$\ghostts[\myEid]$ := $\top$;#
\end{lstlisting}
\begin{lstlisting}[numbers=none]         
  }
  #\lstassert{
    \inv \land \started(\tid, \enqOp)
  }#  
  TS timestamp := newTimestamp();
  #\lstassert{
    \exists T \ldotp \inv \land
    \started(\tid, \enqOp)
    \land \newts({\tt timestamp})
  }#  
  atomic {
    setTimestamp(#$\myTid$#, node, timestamp);
\end{lstlisting}
\begin{lstlisting}[numbers=none,backgroundcolor=\color{gray}]      
    #$\ghostts[\myEid]$ := ${\tt timestamp}$;#
    #$\resof{\events(\myEid)}$ := $\Done$;#
\end{lstlisting}
\begin{lstlisting}[numbers=none]      
  }
  #\lstassert{
    \inv \land \finished(t, \enqOp)
  }#
  return $\Done$;
}
\end{lstlisting}
\end{minipage}

\caption{The proof outline for the enqueue method of the TS queue.
\label{tsq:outline-enq}
}
\end{figure}

\subsection{Proof outline}

We prove the following specifications for $\enqOp$ and $\deqOp$ for each thread
$\tid \in \tidType$ (proof outlines are provided in Figure~\ref{tsq:outline-enq}
and Figure~\ref{tsq:outline-deq}):
\[
\begin{array}{@{}c@{}}
\spec{\rely_\tid, \guar_\tid}
      {\tid}
      {\begin{array}{@{}c@{}}
      \inv \land \started(\tid, \enqOp)
      \end{array}}
      {\enqOp}
      {\begin{array}{@{}c@{}}
      \inv \land \finished(\tid, \enqOp)
      \end{array}}
\\[5pt]
\spec{\rely_\tid, \guar_\tid}
      {\tid}
      {\begin{array}{@{}c@{}}
      \inv \land \started(\tid, \deqOp)
      \end{array}}
      {\deqOp}
      {\begin{array}{@{}c@{}}
      \inv \land \finished(\tid, \deqOp)
      \end{array}}
\end{array}
\]
In the specifications, $\inv$ is the global invariant, and $\rely_\tid$ and
$\guar_\tid$ are rely and guarantee relations defined in \S\ref{sec:details}.
Assertions $\started(\tid, \op)$ and $\finished(\tid, \op)$ are defined in
\S\ref{sec:logic}.


    

We introduce assertions $\fnopotcand$ and $\fispotcand$ to denote the properties
analogous to $\fnocand$ and $\fiscand$ that hold of an enqueue event in the
front of the pool of thread $\tt k$.
To this end, we let $\seen_\ttvisited(\conf,
d)$ denote the set of enqueue events observed $d$ in the pools of threads from $\ttvisited$ (see Definition~\ref{def:seen}), and we also let $\fmints_\ttvisited({\tt ENQ})$ be a predicate
asserting that ${\tt ENQ}$'s timestamp is minimal among enqueues in $\pool{\tt
k}$.
\[
\begin{array}{lcl}
\evalf{\fnopotcand}{\lint}
& = &
\{ (\state, \history, \ghostts)
\mid
\seen_{\tt k}((\state, \history, \ghostts), \myEid) = \emptyset \land \state({\tt pid}) = \NULL
\};
\\
\evalf{\fispotcand}{\lint}
& = &
\{ (\state, \history, \ghostts)
\mid
  \exists {\tt ENQ} \ldotp
 {\tt ENQ} = \getEvent(\events, \ghost, \state({\tt k}), \state({\tt ts}))
 \\ & & \hspace{1em} {} \land
  \fmints_{\tt k}({\tt ENQ}) \land 
  ({\tt ENQ} \in \inqueue{\state(\ttpools), \events, \ghostts} \implies{}
\\ & & \hfill
    {\tt ENQ} \in \seen_{\tt k}((\state, \history, \ghostts), \myEid)) \land
      \state({\tt pid}) \neq \NULL
\};
\end{array}
\]
We further explain how they are used in \S\ref{app:tsq:linvbuild}.

\begin{figure}[t]
\hfill
\begin{minipage}[t]{0.99\textwidth}
\begin{lstlisting}[numbers=none]
Val dequeue() {
  #\lstassert{\inv \land \started(\tid, \deqOp)}#
  Val ret := NULL;
\end{lstlisting}
\begin{lstlisting}[numbers=none,backgroundcolor=\color{gray}]
  #$\eidType$ $\ttcand$;#
\end{lstlisting}
\begin{lstlisting}[numbers=none]
  do {
    TS #\ttstartts# := newTimestamp(); 
    PoolID pid, #\ttcandpid# := NULL;
    TS ts, #\ttcandts# := $\top$;
    ThreadID #\ttcandtid#;
    #\lstassert{
      \inv \land \started(\tid, \deqOp) 
      \land \linv \land \lres[\tid] = \NULL
    }#
    for each k in 1..NThreads do {
      #\lstassert{
        \inv \land \started(\tid, \deqOp) 
        \land \linv \land \lres[\tid] = \NULL
      }#    
      atomic {
        (pid, ts) := getOldest(k);
\end{lstlisting}
\begin{lstlisting}[numbers=none,backgroundcolor=\color{gray}]  
        #$\order$ := ($\order \cup \{ (e, \myEid) \mid
        e \in \ids{\complete{\events}} \cap
          \inqueue{\ttpools, \events, \ghostts}$# 
                                      #${} \land
        \lnot(\ttstartts \tsless \ghostts(e))\})^+;$#
\end{lstlisting}
\begin{lstlisting}[numbers=none]           
      }
      #\lstassert{
        \inv \land \started(\tid, \deqOp) 
        \land \linv \land \lres[\tid] = \NULL \\ {} \land
        (\fnopotcand \lor \fispotcand)
      }#
      if (pid ${\neq}$ NULL && ${\tt ts} \tsless \ttcandts$ && #$\lnot$#(#$\ttstartts \tsless {\tt ts}$#)) {
        #(\ttcandpid, \ttcandts, \ttcandtid)# := (pid, ts, k);
\end{lstlisting}
\begin{lstlisting}[numbers=none,backgroundcolor=\color{gray}]        
        #$\ttcand$ := \getEvent($\events$, $\ghostts$, \ttcandtid, \ttcandts);#
\end{lstlisting}
\begin{lstlisting}[numbers=none]        
      }
      #\lstassert{
        \inv \land \started(\tid, \deqOp) 
        \land \linv \land \lres[\tid] = \NULL
      }#
    }
    #\lstassert{
      \inv \land \started(\tid, \deqOp) 
      \land \linv \land \lres[\tid] = \NULL
    }#
    if (#\ttcandpid# #$\neq$# NULL)
      #\lstassert{
        \inv \land \started(\tid, \deqOp) 
        \land \fiscand \land \ttcandpid \neq \NULL
        \land \lres[\tid] = \NULL
      }#
      atomic {
        ret := remove(#\ttcandtid#, #\ttcandpid#);
\end{lstlisting}
\begin{lstlisting}[numbers=none,backgroundcolor=\color{gray}]          
        if (#${\tt ret} \neq \NULL$#) {
          #$\resof{\events(\myEid)}$ := ret;#
          #$\order$ := ($\order \cup
%            \{ (\ttcand, \myEid) \}$#
%                  #${}\cup
                  \{ (\ttcand, e) \mid
                    e \in \inqueue{\ttpools, \events, \ghostts} \}$#
                  #${}\cup
                  \{ (\myEid, d) \mid
                            \opof{\events(d)} = \deqOp \land
                              d \in \ids{\incomplete{\events}} \}
                  )^+;$#
        }
\end{lstlisting}
\begin{lstlisting}[numbers=none]        
      }
      #\lstassert{
        \inv \land ((\started(\tid, \deqOp) \land \lres[\tid] = \NULL)
        \lor \finished(\tid, \deqOp))
      }#
  } while (ret = NULL);
  #\lstassert{
    \inv \land \finished(t, \deqOp)}# 
  return ret;
}
\end{lstlisting}
\end{minipage}
\caption{The proof outline for the dequeue method of the TS queue.
\label{tsq:outline-deq}
}
\end{figure}

\FloatBarrier
\subsection{Preservation of the loop invariant.}
\label{app:tsq:linvbuild}

We consider the current dequeue operation $\myEid$ in a thread $\tid$, which
generates a timestamp $\ttstartts$ and proceeds to execute {\tt for each} loop.
At ${\tt k}$'s iteration of the loop, the following two steps performed:
\begin{itemize}
  \item \textbf{Step 1:} Using $\tt getOldest(k)$, we learn a pool identifier ${\tt pid}$ and a
  timestamp ${\tt ts}$ of the value in the front of the queue (if there is any).
  Additionally, every enqueue in $\pool{\tt k}$ with a timestamp not greater
  than $\ttstartts$ is ordered in front of $\myEid$. As a result, one of the two
  cases takes place:
  \begin{itemize}
    \item $\fnopotcand$ holds, in which case we say that there is no potential
    candidate in $\pool{\tt k}$. This describes configurations, in which
    $\pool{\tt k}$ is either empty or $\tt ts$ is greater than $\ttstartts$.
    \item $\fispotcand$ holds, in which case we say that the enqueue event ${\tt
    ENQ} = \getEvent(\events, \ghostts, {\tt k}, {\tt ts})$ is the potential
    candidate in $\pool{\tt k}$. This describes configurations, in which $\tt
    ts$ is not greater than $\ttstartts$. Additionally, this requires that $\tt
    ts$ be smaller than other timestamps currently present in $\pool{\tt k}$.
    Note that the latter follows from $\invalg$(ii) and $\invalg$(i)
    after Step 1.
  \end{itemize}
  \item \textbf{Step 2:} if there is a potential candidate, its timestamp is compared to the
  timestamp $\ttcandts$ of the current candidate for removal and the earliest of
  the two is kept as the candidate;
\end{itemize}
To show that the loop invariant $\linv$ by the $\tt k$'s iteration, we consider
separately the cases when there is the potential candidate ${\tt ENQ}$ and when
there is no such enqueue event.

Let us first assume that $\fispotcand$ holds, and ${\tt ENQ} =
\getEvent(\events, \ghostts, {\tt k}, {\tt ts})$ is the potential candidate in
$\pool{\tt k}$. At step 2 (line~\ref{line:tsq_compare}), the current dequeue
compares ${\tt ts}$ to $\ttcandts$ and decides whether to chose ${\tt ENQ}$ as
the candidate for removal. According to the loop invariant, there are two
possibilities: either no candidate has been chosen so far ($\fnocand$ holds), or
there is a candidate $\ttcand$ ($\fiscand$ holds). When the former is the case,
$\seen_\ttvisited(\conf, \myEid) = \emptyset$. It is easy to see that
$\fispotcand$ immediately implies $\fiscand$ after ${\tt k}$'s iteration. Let us
now consider the case when $\fiscand$ holds, and $\ttcand$ has been selected as
the candidate for removal out of enqueues in threads from $\ttvisited$. Let us
assume that ${\tt ts} \tsless \ttcandts$ takes place (the other situation is
justified analogously). To conclude $\fiscand$ for the next iteration, we need
to show that:
\begin{itemize}
\item ${\tt ENQ} \in \inqueue{\ttpools, \events, \ghostts} \implies
   {\tt ENQ} \in \seen_{\ttvisited \uplus \{ {\tt k} \}}((\state, \history, \ghostts), \myEid)$, and
\item $\fmints_{\ttvisited \uplus \{ {\tt k} \}}({\tt ENQ})$.
\end{itemize}
The first requirement follows trivially from $\fispotcand$: if ${\tt ENQ} \in \inqueue{\ttpools, \events, \ghostts}$ holds, then so does:
\[
{\tt ENQ} \in \seen_{\{ {\tt k} \}}((\state, \history, \ghostts), \myEid) \subseteq \seen_{\ttvisited \uplus \{ {\tt k} \}}((\state, \history, \ghostts), \myEid).
\]
It remains to show that $\fmints_{\ttvisited \uplus \{ {\tt k} \}}({\tt ENQ})$
holds, i.e. that every other enqueue in $\seen_{\ttvisited \uplus \{ {\tt k}
\}}(\conf, \myEid)$ does not have a timestamp smaller than ${\tt ENQ}$.
According to $\fispotcand$, $\tt ENQ$ is minimal among enqueues in thread $\tt
k$. Let us assume that $\tt ENQ$ is not minimal among enqueues in $\ttvisited$,
i.e. that there is $\tt ENQ' \in \seen_\ttvisited(\conf, \myEid)$ such that
$\ghostts({\tt ENQ'}) \tsless \ghostts({\tt ENQ})$. Knowing that $\ghostts({\tt
ENQ}) \tsless \ttcandts$, we conclude that $\ghostts({\tt ENQ'}) \tsless
\ttcandts$, which contradicts the loop invariant. Therefore, $\tt ENQ$ is minimal among enqueues in both $\ttvisited$ and $\tt k$.


\artem{The treatment of configurations here is a mess.}


Now let us assume that $\fnopotcand$ holds, i.e. that there is no potential
candidate in $\pool{\tt k}$. In this case, the candidate for removal remains
unchanged. Intuitively, when there is no potential candidate in $\pool{\tt k}$,
all values occurring in the pool have timestamps greater than $\ttstartts$.
According to the invariant $\invalg$(i), all successors of corresponding events
will have even greater timestamps.

Prior to ${\tt k}$'s iteration, either $\fnocand$ or $\fiscand$ holds. Let us
first assume the former. Then no candidate has been chosen after iterating over
$\ttvisited$. Together, $\fnocand$ and $\fnopotcand$ immediately imply
$\fnocand$ for the next iteration. Let us now consider the case when $\fiscand$
holds. Then there is a candidate for removal $\ttcand$. It is easy to see that
$\ttcand \in \seen_{\ttvisited \uplus \{ {\tt k} \}}(\conf, \myEid)$ holds, so
it remains to ensure that $\fmints_{\ttvisited \uplus \{ {\tt k} \}}(\ttcand)$
holds. To this end, we need to demonstrate that for every enqueue $e \in
\seen_{\tt k}(\conf, \myEid)$, $\lnot(\ghostts(e) \tsless \ghostts(\ttcand))$
holds. However, according to $\fnopotcand$, there are no such enqueues $e$, so
$\fiscand$ can be concluded for the next iteration.

\subsection{Auxiliary proofs for the loop invariant (Lemma~\ref{lem:tsqmin})}

\begin{lemma}\label{lem:tsqinvis}
Given any configuration $\conf = (\state, (\events, \order), \ghostts)$
satisfying $\inv$ and an identifier $d$ of a dequeue event that has generated
its timestamp $\ttstartts$, an enqueue by a visited thread not seen by $d$ does not precede any enqueue seen by $d$:
\begin{multline*}
\forall \eid, \eid' \ldotp
\eid \in \inqueue{\state, \events, \ghostts}
    \setminus \seen(\conf, d)
    \\ {}
    \land \eid' \in \seen(\conf, d)
    \land
\{\tidof{\events(\eid)}, \tidof{\events(\eid')}\} \subseteq \ttvisited \implies \nedge{\eid}{\order}{\eid'}.
\end{multline*}
\end{lemma}

\begin{proof}
We do a proof by contradiction. Let us assume that there exist $\eid$ and
$\eid'$ satisfying the premise of the implication above and such that
$\edge{\eid}{\order}{\eid'}$ holds. Since $\tidof{\events(\eid)} \in
\ttvisited$ and $\eid \in \inqueue{\state, \events, \ghostts}
\setminus \seen(\conf, d)$, the following
holds by definition of $\seen$:
\begin{itemize}
\item[(a)] $\eid \notin \complete{\ids{\events}}$,
\item[(b)] $\state(\ttstartts) \tsless \ghostts(\eid)$,
\item[(c)] $\nedge{\eid}{\order}{d}$.
\end{itemize}
Note that (b) takes place whenever (a) does. Let us assume (a). Since $\eid$ is
not completed, $\ghostts(\eid) = \top$ holds (by $\invwf$). On the other hand,
by the assumption of the lemma, $\ttstartts$ contains a non-maximal timestamp.
Under such conditions, (b) holds.

Let us obtain a contradiction for (b). Since $\conf$ satisfies the invariant
$\inv$, from $\invalg$(i) and $\edge{\eid}{\order}{\eid'}$ we learn that
$\ghostts(\eid)
\tsless \ghostts(\eid')$. Consequently:
\[
\state(\ttstartts) \tsless \ghostts(\eid) \tsless \ghostts(\eid')
\]
On the other hand, since $\eid' \in \seen(\conf, d)$ holds, so does $\lnot
(\state(\ttstartts) \tsless \ghostts(\eid'))$ by definition of $\seen$. Thus, we
arrived to a contradition.

Let us obtain a contradiction for (c). Since $\eid' \in \seen(\conf, d)$,
$\edge{\eid'}{\order}{d}$ holds. By Defition~\ref{def:history} of a history,
$\order$ is a transitive relation. Thus, $\edge{\eid}{\order}{\eid'}$ and
$\edge{\eid'}{\order}{d}$ together imply $\edge{\eid}{\order}{d}$, which
contradicts (c).
\qed
\end{proof}

\begin{proof}[Lemma~\ref{lem:tsqmin}]
Let us take any interpretation of logical variables $\lint$ and a configuration $\conf = (\state,
(\events, \order), \ghostts) \in \evalf{\fiscand}{\lint}$ such that $\ttcand =
\getEvent(\events, \ghostts, \ttcandtid, \ttcandts)$ and $\ttcand \in
\inqueue{\state({\tt pools}), \events, \ghostts}$ both hold. We need to prove
that:
\begin{equation*}
\forall e \in \inqueue{\state(\tt pools), \events, \ghostts} \ldotp
\tidof{\events(e)} \in \ttvisited
\implies
\nedge{e}{\order}{\ttcand}
\end{equation*}

Let ${\tt DEQ} = \myEid$ be the current dequeue event. By definition, the set
$\seen(\conf, {\tt DEQ})$ is a subset of $\inqueue{\ttpools, \events,
\ghostts}$. In other words, every enqueue with a value in the data structure is
either seen by ${\tt DEQ}$ or not.

According to $\fiscand$, $\ttcand \in \seen(\conf, {\tt DEQ})$. By Lemma~\ref{lem:tsqinvis}, no unseen enqueue
can precede $\ttcand$ in the abstract history. Additionally, since $\ttcand$'s
timestamp is minimal among enqueues seen by ${\tt DEQ}$, $\ttcand$ is necessary
$\order$-minimal among them according to $\invalg$(i).
\qed
\end{proof}

\subsection{Preservation of $\invalg$(i)} \label{sec:invts}

Showing that the invariant $\inv$ is preserved by all primitive commands is
mostly straightforward, except for the command assigning a timestamp to an
enqueued value at line~\ref{line:tsq_setts}. When the latter happens, it is
necessary to prove that the property of the timestamps $\invalg$(i) is not
invalidated. To show that this is indeed the case, one has to observe a certain
property of timestamps generated by the function {\tt newTimestamp}: a timestamp
generated for the current enqueue event $\myEid$ and stored in a
memory cell $\tt timestamp$ is greater than timestamps of all enqueues that
precede $\myEid$ in the abstract history and still have their values in the
data structure. Specifically, we define an assertion $\newts({\tt ts})$ denoting
configurations $(\state, (\events, \order), \ghost)$ that satisfy the following:
\[
\forall \eid \in \inqueue{\state({\tt pools}), \events, \ghost} \ldotp \edge{\eid}{\order}{\myEid} \implies
\ghostts(\eid) \tsless {\tt ts}
\]
It is easy to see that $\newts$ asserts the same property as $\invalg$(i), but only for the current event and a timestamp generated for it. When at line~\ref{line:tsq_setts} the timestamp gets assigned, $\newts$ enables concluding that $\invalg$(i) is preserved.

\begin{figure}[t]
\begin{lstlisting}[basicstyle=\small\ttfamily]
int counter = 1;

TS newTimestamp() {
#\lstassert{
    \inv \land \started(\tid, \enqOp)
}#   
  int oldCounter = counter;
#\lstassert{
    \inv \land \started(\tid, \enqOp) \land
    {\sf cntProp}({\tt oldCounter})\\
    {}\land {\tt oldCounter} \leq {\tt counter} 
}# 
  TS result;
  if (CAS(counter, oldCounter, oldCounter+1))
    #\lstassert{
        \inv \land
        \started(\tid, \enqOp)
        \\{}\land {\sf cntProp}({\tt oldCounter})
        \land {\tt oldCounter} < {\tt counter} 
    }# 
    result = (oldCounter, oldCounter);
    #\lstassert{
        \inv \land
        \started(\tid, \enqOp) \land \newts({\tt result})
    }# 
  else
    #\lstassert{
        \inv \land
        \started(\tid, \enqOp)
        \\{}\land {\sf cntProp}({\tt oldCounter})
        \land {\tt oldCounter} < {\tt counter}
    }#
    result = (oldCounter, counter-1);
    #\lstassert{
        \inv \land
        \started(\tid, \enqOp) \land \newts({\tt result})
    }#
  #\lstassert{
    \inv \land
        \started(\tid, \enqOp) \land \newts({\tt result})
  }# 
  return result;
}
\end{lstlisting}
\caption{Proof outline for the timestamp generating algorithm}
\label{fig:outline-ts}
\end{figure}

We prove the following Hoare
specification for the timestamp generation algorithm and outline the proof if Figure~\ref{fig:outline-ts}:

\begin{minipage}{\textwidth}
\begin{lstlisting}[numbers=none,basicstyle=\small\ttfamily]
#\lstassert{
    \inv \land \started(\tid, \enqOp)
}# 
TS timestamp := newTimestamp();
#\lstassert{
  \inv \land \started(\tid, \enqOp) \land \newts
}# 
\end{lstlisting}
\end{minipage}

The assertion $\newts$ is obtained with the help of the following auxiliary assertion, which connects the generated timestamp to the real-time order using $\invalg$(iii):
\begin{multline*}
\evalf{{\sf cntProp}(C)}{\lint} \triangleq
\{
(\state, \history, \ghost) \mid
\forall a, b \ldotp \forall \eid \in \inqueue{\state({\tt pools}), \history,
\ghost} \ldotp \edge{\eid}{\order}{\myEid}
\\ {}
\land \ghostts(\eid) = (a, b) \implies b < C
\}
\end{multline*}
It is easy to see that ${\sf cntProp}({\tt counter})$ is implied by the
invariant property $\invalg$(iii). Thus, after the first line of {\tt
newTimestamp}, ${\sf cntProp}({\tt oldCounter})$ holds. Later on, when the
timestamp ${\tt result}$ is formed, ${\sf cntProp}({\tt counter})$ yields us the
fact that ${\tt result}$ is a timestamp greater than timestamps of all enqueues
that have a value in the pools and precede $\myEid$, which concludes the proof
of $\newts({\tt ts})$.

\subsection{Stability of the loop invariant}\label{app:lem:tsq:linv:stable}

\begin{proof}[Lemma~\ref{lem:tsq:linv:stablevis}]
Since $(\conf, \conf') \in \rely_{\tid}$, there exists a thread $\tid'$ such
that one of the following situations takes place:
\begin{itemize}
  \item $(\conf, \conf') \in {\dashrightarrow}_{\tid'}$,
  \item $(\conf, \conf') \in \guar_{\tid', {\tt local}}$, or
  \item there exists $\hat{\alpha}$ and $P$ such that
  $\guar_{\tid',\hat{\alpha},P} \subseteq \rely_{\tid}$ and $(\conf, \conf') \in
  \guar_{\tid',\hat{\alpha},P}$.
\end{itemize}
In further, we prove the lemma separately for each $\hat{\alpha}$ and $P$. In
each case, we assume that $\conf = (\state, (\events, \order), \ghostts)$ and
$\conf' = (\state', (\events', \order'), \ghostts')$.

\noindent\textbf{Case \#1:} $(\conf, \conf') \in {\dashrightarrow}_{\tid'}$.
This environment transition only adds a new event $e$ in a thread $\tid'$ and
orders it after completed events. As a result of this environment transition,
$e$ is uncompleted in $\conf'$. By Definition~\ref{def:history},
$\edge{e}{\order'}{{\tt DEQ}}$ does not hold. Consequently, $e \notin
\visible(\conf', {\tt DEQ})$. It is easy to see that all other enqueues outside
of $\visible(\conf, {\tt DEQ})$ are not affected by this environment transition,
so we can conclude that $\visible(\conf', {\tt DEQ}) \subseteq \visible(\conf,
{\tt DEQ})$.

\medskip

\noindent\textbf{Cases \#2 and \#3:} $\hat{\alpha}$ is either ${\tt insert}$ or
${\tt setTS}$, and $P = \inv\land\started(\tid',\enqOp)$. These environment
transitions only update the abstract history, concrete and ghost state
associated with an event $e$, which is uncompleted in $\conf$ ($\events(e) =
(\tid', \enqOp, v, \Todo)$). Since $e$ is uncompleted, by
Definition~\ref{def:history}, $\edge{e}{\order}{{\tt DEQ}}$ does not hold.
Neither of these environment transitions add any edges into the abstract
history, meaning that $\edge{e}{\order'}{{\tt DEQ}}$ does not hold either.
Consequently, $e \notin \visible(\conf', {\tt DEQ})$. It is easy to see that all
other enqueues outside of $\visible(\conf, {\tt DEQ})$ are not affected by this
environment transition, so we can conclude that $\visible(\conf', {\tt DEQ})
\subseteq \visible(\conf, {\tt DEQ})$.

\medskip

\noindent\textbf{Case \#4:} $\hat{\alpha} = {\tt scan}$ and $P =
\inv\land\started(\tid', \deqOp)$. This environment transition orders some of
the enqueue events in front of an uncompleted dequeue $d$ ($\events(d) = (\tid',
\deqOp, \_, \Todo)$). Let $e$ be an enqueue event in $\conf$ such that $e \notin
\visible(\conf, {\tt DEQ})$. Out of the reasons why $e$ is not visible by ${\tt
DEQ}$ in $\conf$, only $\edge{e}{\order}{{\tt DEQ}}$ may be affected by this
environment transition, as it simply adds edges in the abstract history.
However, we argue that an edge $\edge{e}{\order'}{{\tt DEQ}}$ is not added by
${\tt scan}$. Indeed, $d$ is uncompleted, so by Definition~\ref{def:history} it
cannot precede any other event in the abstract history. Consequently,
$\edge{e}{\order'}{{\tt DEQ}}$ is not added as implied by transitivity.

\noindent\textbf{Case \#5:} $\hat{\alpha} = {\tt remove}$ and $P = \inv \land
\started(\tid', \deqOp)$. Let $d$ be the uncompleted event by a thread $\tid'$,
i.e. such that $\events(d) = (\tid', \deqOp, \_, \Todo)$. Let $e \in
\inqueue{\state({\tt pools}), \events, \ghostts}$ be the enqueue event removed
by this environment transition. As a result,  $e \in
\inqueue{\state'({\tt pools}), \events', \ghostts'}$ does not hold in $\conf'$,
so $e \notin \visible(\conf', {\tt DEQ})$. It is easy to see that this
environment transition affects other enqueue events only by ordering them w.r.t.
other events. Consequently, if some ${\tt ENQ} \notin \visible(\conf, {\tt
DEQ})$, the only reason it may become visible in $\conf'$ is an addition of the
edge $\edge{{\tt ENQ}}{\order'}{{\tt DEQ}}$. However ${\tt remove}$ does not
introduce such edge, and it is not implied by transitivity. \artem{Say more, and
more formally.}
\medskip

\noindent\textbf{Case \#6:} $\hat{\alpha} = {\tt genTS}$ and $P = \inv$. This
transition does not affect any concrete state, ghost state or the abstract
history associated with any enqueue event.
\medskip

\noindent\textbf{Case \#7:} $(\conf, \conf') \in \guar_{\tid', {\tt local}}$.
This transition does not affect any concrete state, ghost state or the abstract
history associated with any enqueue event.
\qed
\end{proof}

%% file: app-set.tex

\newcommand{\streach}[3]{{#2} \mathrel{{\leadsto}_{#1}} {#3}}
\newcommand{\nstreach}[3]{{#2} \mathrel{{\not\leadsto}_{#1}} {#3}}

\newcommand{\ttcurr}{{\tt curr}}
\newcommand{\ttmarked}{{\tt marked}}
\newcommand{\ttval}{{\tt val}}

\section{Proof details for the Optimistic Set}
\label{app:set}

\subsection{Overview of proof details}

\begin{figure}[t]
\hfill
\begin{minipage}[t]{0.92\textwidth}
\begin{itemize}
\item[$(\invlin)$] all linearizations of completed events of the abstract
history satisfy the queue specification:
\[
\abs(\history, \histories_{\sf set})
\]
\item[($\invord$)]
completed insert and remove events are linearly ordered:
\[
\forall \event, \event' \ldotp \opof{\events(\event)}, \opof{\events(\event')}
\in \{ \insertOp, \removeOp \} \implies \edge{\event}{\order}{\event'} \lor
\edge{\event'}{\order}{\event}
\]
\item[($\invalg$)] for every node $n \in \dom{\state}$, the following holds:
\begin{itemize}[leftmargin=0pt]
\item[$(i)$]
$\begin{multlined}[t]
\node.\ttmarked = \FALSE \implies
(\events, \node) \in \dom{\insertOf} \land 
(\events, \node) \notin \dom{\removeOf}
\\ {} \land
\forall \eid \ldotp
\edge{\insertOf(\events, \node)}{\order}{\eid}
\land
\resof{\events(\eid)} = \TRUE
\implies \opof{\events(\eid)}
= \containsOp
\end{multlined}$
\item[$(ii)$] $\begin{multlined}[t]
n.\ttmarked = \TRUE \implies
(\events, \node) \in \dom{\insertOf} \land 
(\events, \node) \in \dom{\removeOf} \\ {} \land
\forall \eid \ldotp
\edge{\insertOf(\events, \node)}{\order}{\edge{\eid}{\order}{\removeOf(\events, \node)}}
\land
\resof{\events(\eid)} = \TRUE \implies{} \\
\opof{\events(\eid)}
= \containsOp
\end{multlined}$
\item[$(iii)$]
$\forall \node' \in \dom{\state} \ldotp
\streach{\state}{\node}{\node'} \implies \node.\ttval \leq \node'.\ttval$
\item[$(iv)$]
$\forall \node \ldotp (\events, \node) \in \dom{\removeOf} \implies
\edge{\insertOf(\events, \node)}{\order}{\removeOf(\events, \node)}$
\item[$(v)$]
$\streach{\state}{{\tt head}}{\node} \iff \node.\ttmarked = \FALSE$
\item[$(vi)$]
$\streach{\state}{\node}{{\tt tail}}$
\end{itemize}
\item[($\invwf$)]
\begin{itemize}
\item[$(i)$]
$\forall \event \ldotp
\resof{\events(\event)} = \TRUE
  \iff \dom{\ghostnode}$
\item[$(ii)$]
$\forall \event \ldotp \argsof{\events(\event)} = \ghostnode(\event).\ttval$
\item[$(iii)$]
$\forall \node \in \dom{\state} \ldotp \node.\ttmarked = \TRUE
  \lor \node.\ttmarked = \FALSE$
\end{itemize}
\end{itemize}
\end{minipage}
\caption{The Optimistic Set: The invariant $\inv = \invlin \land \invord \land
\invalg \land \invwf$
\label{app:fig:set:inv}
}
\end{figure}

We prove the following specifications for the set operations:
\[
\spec{\rely_\tid, \guar_\tid}
      {\tid}
      {\begin{array}{@{}c@{}}
      \inv \land \started(\tid, \insertOp)
      \end{array}}
      {\insertOp}
      {\begin{array}{@{}c@{}}
      \inv \land \finished(\tid, \insertOp)
      \end{array}}
\]
\[
\spec{\rely_\tid, \guar_\tid}
      {\tid}
      {\begin{array}{@{}c@{}}
      \inv \land \started(\tid, \removeOp)
      \end{array}}
      {\removeOp}
      {\begin{array}{@{}c@{}}
      \inv \land \finished(\tid, \removeOp)
      \end{array}}
\]
\[
\spec{\rely_\tid, \guar_\tid}
      {\tid}
      {\begin{array}{@{}c@{}}
      \inv \land \started(\tid, \containsOp)
      \end{array}}
      {\containsOp}
      {\begin{array}{@{}c@{}}
      \inv \land \finished(\tid, \containsOp)
      \end{array}}
\]

For each thread $\tid$, we generate rely and guarantee relations analogously to
\S\ref{sec:details}. To this end, we let ${\tt insert}$, ${\tt remove}$, ${\tt
contains}$ denote atomic steps corresponding to atomic blocks extended with
ghost code in Figure~\ref{fig:set} (at lines
\ref{line:insert_atomic_start}-\ref{line:insert_atomic_end},
\ref{line:remove_atomic_start}-\ref{line:remove_atomic_end} and
\ref{line:locate_atomic_start}-\ref{line:locate_atomic_end} accordingly). For
each thread $\tid$, relations $\guar_\tid$ and $\rely_\tid$ are then defined as
follows:

\[
\begin{array}{rcl}
\guar_\tid & \triangleq &
\guar_{\tid,{\tt insert},\inv} \cup
\guar_{\tid,{\tt remove},\inv} \cup
\guar_{\tid,{\tt contains},\inv} \cup
\guar_{\tid, {\tt local}},
\\
\rely_\tid & \triangleq &
\cup_{\tid' \in \tidType \setminus \{\tid\}}
(\guar_{\tid'} \cup {\dashrightarrow}_{\tid'})
\end{array}
\]
In the above, we assume a relation $\guar_{\tid, {\tt local}}$, which describes
arbitrary changes to certain program variables and no changes to the abstract
history and the ghost state. That is, we say that the nodes of the linked list
(such as ${\tt head}$ are {\em shared} program variables in the algorithm, and
all others are {\tt thread-local}, in the sense that every thread has its own
copy of them. We let $\guar_{\tid, {\tt local}}$ denote every possible change to
{\em thread-local} variables of a thread $\tid$ only.

In Figure~\ref{app:fig:set:inv} we present the invariant $\inv$. To formulate
the invariant, we characterise all of the nodes in the data structure as either
{\em reachable} or {\em unreachable}.

\begin{definition}[Reachable nodes]
For a set of nodes of the data structure in a state $\state$, we let
${\leadsto}_\state \subseteq \nodeType \times \nodeType$ to be a {\em
reachability} relation on the nodes and let $\streach{\state}{\node}{\node'}$
hold whenever there exists a sequence of node identifiers $\node_0, \node_1,
\dots,
\node_{k}$ ($k \geq 0$) such that $\node_{i+1} = ({*}\node_i).{\tt next}$,
$\node_0 = \node$, $\node_k = \node'$.
\end{definition}

Additionally, we define a function $\removeOf :
\powerset{\eventType} \times \nodeType \rightharpoonup \eidType$ which maps a
node identifier $n$ to a matching remove event identifier $\eid$ (if it exists).
\[
\removeOf(\events, n) =
\begin{cases}
\eid, \mbox{if } \ghostnode(\eid) = n \land \opof{\events(\eid)} = \removeOp\\
\mbox{undefined otherwise}
\end{cases}
\]

We also assume that $\eventType$ consists of well-typed queue events
$\sub{\eid}{(\tid, \op, \args, \res)}$ meeting the following constraints:
\begin{itemize}
\item $\op \in \opType = \{ \insertOp, \removeOp, \containsOp \}$,
\item $\args \in \valueType$, and
\item $\res \in \{ \FALSE, \TRUE \}$.
\end{itemize}

\subsection{Loop invariant}

The most important part of the proof are the obligations to satisfy the
sequential specification of the set at the commitment point of {\tt contains}. As we argue in \S\ref{sec:set}, it is necessary to demonstrate that the following two properties hold of the current {\tt contains} event $\myEid$:
\begin{itemize}
  \item if $\ttcurr.\ttval = \argsof{\events(\myEid)}$, then all successful
  removes after $\insertOf(\events, {\tt curr})$ are concurrent with $\myEid$;
  \item if $\ttcurr.\ttval > \argsof{\events(\myEid)}$, then all successful
  inserts after $\lastremove(\events, \argsof{\events(\myEid)})$ are concurrent
  with $\myEid$;
\end{itemize}
To discharge both obligations, we build a loop invariant $\linv$ for the loop in
the {\tt locate} operation invoked by the {\tt contains} operation in a thread
$\tid$. For a given interpretation of logical variables $\lint$, the loop
invariant $\linv$ denotes triples $(\state, (\events, \order), \ghostnode) \in
\evalf{\linv}{\lint}$ such that the following conditions hold of the current
node \ttcurr in a thread $\tid$ and every node $\node' \in \nodeType$:
\begin{itemize}
\item when $\node'$ is reachable from $\node$ and stores the value sought by the
{\tt contains} operation, it is either in the data structure or a matching
remove operation is concurrent with the current one:
\begin{multline*}
\streach{\state}{\ttcurr}{\node'} \land \node'.{\sf val} = \argsof{{\sf
last}(\tid, \events, \order)} \implies{}
\\
\node'.{\sf marked} = \FALSE \lor
\nedge{\removeOf(\node')}{\order}{{\sf last}(\tid, \events, \order)}
\end{multline*}
\item when $\node'$ is not reachable from $\node$ and stores the value sought by
the {\tt contains} operation, it is either removed from the data structure or it
has been inserted concurrently:
\begin{multline*}
\nstreach{\state}{\ttcurr}{\node'}
 \land \node'.{\sf val} = \argsof{{\sf
last}(\tid, \events, \order)} \implies{}
\\
\node'.{\sf marked} = \TRUE \lor
\nedge{\insertOf(\node')}{\order}{{\sf last}(\tid, \events, \order)}
\end{multline*}
\end{itemize}

\artem{informal explanations of lemmas are commented out. merge them with the proofs for readability}



\begin{lemma}
For every $\lint : \lvarsType \to \valueType$ and configuration $(\state,
(\events, \order), \ghostts) \in \evalf{\linv}{\lint}$, if ${\tt curr.val} =
\argsof{\events(\myEid)}$ then:
\[
\neg\exists \eid \ldotp 
\events(\eid) = (\_, \removeOp, \argsof{\events(\myEid)}, \TRUE)
\land
\edge{\insertOf(\events, \ttcurr)}{\order}{\edge{\eid}{\order}{\myEid}}
\]
\end{lemma}
\begin{proof}
According to the loop invariant, the following holds:
\begin{equation*}\label{app:eq:setlemmat}
\ttcurr.\ttmarked = \FALSE \lor
\nedge{\removeOf(\ttcurr)}{\order}{\myEid}
\end{equation*}
Let us first consider the case when $\ttcurr.\ttmarked = \FALSE$. According to
$\invalg$(i), the following is true:
\begin{multline*}
(\events, \ttcurr) \in \dom{\insertOf} \land 
(\events, \ttcurr) \notin \dom{\removeOf}
\\ {} \land
\forall \eid \ldotp
\edge{\insertOf(\events, \node)}{\order}{\eid}
\land
\resof{\events(\eid)} = \TRUE
\implies \opof{\events(\eid)}
= \containsOp
\end{multline*}
which immediately implies that no successful remove operation follows
$\insertOf(\events, \ttcurr)$ and allows us to conclude the statement of the
lemma.

Let us now consider the case when $\ttcurr.\ttmarked = \TRUE$. By
(\ref{app:eq:setlemmat}), $\nedge{\removeOf(\ttcurr)}{\order}{\myEid}$ also holds then. Let us assume that there exists a remove event $r$ contradicting the lemma:
\[ 
\opof{\events(r)} = \removeOp \land
\argsof{\events(r)} = \argsof{\events(\myEid)}
\land
\edge{\insertOf(\events, \ttcurr)}{\order}{\edge{r}{\order}{\myEid}}
\]
By $\invalg$(ii) and $\invord$, it can only be the case that
$\edge{\insertOf(\events, \ttcurr)}{\order}{\edge{\removeOf(\events,
\ttcurr)}{\order}{r}}$. However, together with the formula above, that implies
$\edge{\removeOf(\ttcurr)}{\order}{\myEid}$, so we arrived to a contradiction.
Consequently, the statement of the lemma holds.
\qed
\end{proof}

\begin{lemma}
For every $\lint : \lvarsType \to \valueType$ and configuration $(\state,
(\events, \order), \ghostts) \in \evalf{\linv}{\lint}$, if $\ttcurr.\ttval >
\argsof{\events(\myEid)}$ then:
\begin{multline*}
\neg\exists \eid \ldotp
\events(\eid) = (\_, \insertOp, \argsof{\events(\myEid)}, \TRUE)
\\ {} \land
\edge{\lastremove(\argsof{\events(\myEid)})}{\order}{\edge{i}{\order}{\myEid}}
\end{multline*}
\end{lemma}
\begin{proof}
It is easy to see that $\invalg$(i, ii, iii) and $\invwf$(i) together imply that:
\[
\forall \node, \eid \ldotp \node \in \dom{\state} \Leftrightarrow
\eid = \insertOf(\events, \node)
\]
Let us assume that there exists $\eid$ such that:
\begin{multline}\label{app:eq:setlemmafgoal}
\events(\eid) = (\_, \insertOp, \argsof{\events(\myEid)}, \TRUE)
\\ {} \land
\edge{\lastremove(\argsof{\events(\myEid)})}{\order}{\edge{i}{\order}{\myEid}}
\end{multline}
Then $\ghostnode(\eid).\ttval = \argsof{\events(\myEid)}$.
By $\invalg$(iii), $\ttcurr.\ttval >
\argsof{\events(\myEid)}$ implies that $\nstreach{\state}{\ttcurr}{\ghostnode(\eid)}$.
Hence, from the loop invariant we learn that:
\begin{equation}\label{app:eq:setlemmaf}
\node.\ttmarked = \TRUE \lor
\nedge{\eid}{\order}{\myEid}
\end{equation}
Let us first consider the case when $\ttcurr.\ttmarked = \TRUE$. According to
$\invalg$(ii), there exists $\removeOf(\events, \node)$. By $\invalg$(iv),
$\edge{\eid}{\order}{\removeOf(\events, \node)}$. Note that
$\lastremove(\argsof{\events(\myEid)})$ is the last remove event of a value
$\argsof{\events(\myEid)}$, so $\edge{\eid}{\order}{\lastremove(\argsof{\events(\myEid)})}$. However, that contradicts (\ref{app:eq:setlemmafgoal}).

Let us now consider the case when $\ttcurr.\ttmarked = \FALSE$. By
(\ref{app:eq:setlemmaf}), $\nedge{\eid}{\order}{\myEid}$ also holds
then. However, that contradicts (\ref{app:eq:setlemmafgoal}).

For all values $\ttcurr.\ttmarked$, we got a contradiction assuming that there
exists $\eid$ satisfying (\ref{app:eq:setlemmafgoal}). Consequently, such $\eid$
does not exist, which concludes the proof of the lemma.
\qed
\end{proof}

\artem{TODO preservation of $\linv$}



\artem{TODO preservation of $\invlin$}

\artem{TODO stability}

\artem{Alternative proof with helping is commented out}

%% file: app-hwq.tex

\newcommand{\slotType}{{\sf Slot}}
\newcommand{\ghostslot}{\ghost_{\sf slot}}

\newcommand{\ttarray}{{\tt Array}}
\newcommand{\ttk}{{\tt k}}
\newcommand{\ttn}{{\tt n}}
\newcommand{\ttback}{{\tt back}}

\section{The Herlihy-Wing Queue}

\subsection{The algorithm}

\begin{figure}[t]
\begin{minipage}[t]{0.36\textwidth}
\begin{lstlisting}[basicstyle=\small\ttfamily]
int back = 0;
int[] #\ttarray# =
          new int[+#$\infty$#];

enqueue(#$\valueType$# v) {
  #\lstassert{ \inv \land \started(\tid, \enqOp)}#
  atomic { // getSlot $\label{line:hwq:getslot}$
    #\ttk# := inc(#\ttback#);
    #\colorbox{gray}{
      $\ghostslot$[\myEid] := \ttk;
      }#
  }
  #\lstassert{ \inv \land {\sf hasSlot}(\tid, \enqOp)}#
  atomic { // insert $\label{line:hwq:insert}$
    #\ttarray[\ttk]# := v;
    #\colorbox{gray}{
      $\resof{\myEid} := \Done;$
      }#
  }
  #\lstassert{ \inv \land \finished(\tid, \enqOp)}#
}
\end{lstlisting}
\end{minipage}
\hfill
\begin{minipage}[t]{0.6\textwidth}
\begin{lstlisting}[firstnumber=14,basicstyle=\small\ttfamily]
#$\valueType$# dequeue() {
  Val res = $\NULL$;
  #\lstassert{
    \inv \land \started(\tid, \deqOp)
  }#
  do {
    #\ttn := \ttback;# $\label{hwq:back}$
    #\lstassert{
      \inv \land \started(\tid, \deqOp) \land \ttn \leq \ttback
    }#
    for k = 1 to n {
    #\lstassert{
      \inv \land \started(\tid, \deqOp) \land \linv
    }#   
    atomic { // remove  $\label{line:hwq:remove}$
      res := Swap(#\ttarray[\ttk]#, #$\NULL$#);
\end{lstlisting}
\begin{lstlisting}[firstnumber=23,backgroundcolor=\color{gray}]  
      if (res #${}\neq \NULL$#) {
        #$\eidType$# #${\tt ENQ}$# := #${\sf getEvent}$(\ttk);#
        #$\resof{\events(\myEid)}$# := res;
        #$\order := (\order \cup \{ ({\tt ENQ}, \myEid) \}$#
               #${} \cup \{ ({\tt ENQ}, e') \mid e' \in \untaken{\events, \ghost} \}$#
               #${}\cup \{ (\myEid, d) \mid \opof{\events(d)} = \deqOp$#
                         #${} \land d \in \ids{\incomplete{\events}} \})^+$#
      }
\end{lstlisting}
\begin{lstlisting}[firstnumber=31] 
    }
    #\lstassert{
      \inv \land \linv \land ((\started(\tid, \deqOp) \land {\tt res} = \NULL)
      \\ \hfill {} \lor (\finished(\tid, \deqOp) \land {\tt res} \neq \NULL)) 
    }#
    if res $\neq$ $\NULL$ then
      break;
    #\lstassert{
      \inv \land \started(\tid, \deqOp) \land \linv
    }#
  } } while (res != NULL);
  #\lstassert{
    \inv \land \finished(\tid, \deqOp)
  }#
}
\end{lstlisting}
\end{minipage}
\caption{The Herlihy-Wing queue
\label{app:fig:hwq}
}
\end{figure}

We now present the Herlihy-Wing queue~\cite{linearizability} as our next running
example. Values in the queue are stored in an infinite array, $\ttarray$, with
unbounded index $\ttback$ pointing to the first unoccupied cell of the
array. Initially, each cell of the array is considered empty and contains
$\NULL$. Accordingly, initially $\ttback = 0$.

An {\tt enqueue} operation performs two steps. First, it acquires an index
$\ttk$ with the help of atomic command ${\tt inc}$ returning the value of
$\ttback$ and then incrementing it. At the second step, the {\tt enqueue}
operations stores its argument in $\ttarray[\ttk]$.

A {\tt dequeue} operation obtains the length of the currently used part of the
array and stores it in $\ttn$. Then the operation iterates over array cells from
the beginning till $\ttn$ and looks at the values in them. If a non-$\NULL$
value is encountered, the cells gets overwritten with $\NULL$ to remove the
value from the queue, and the value itself is returned as a result of the {\tt
  dequeue} operation. Alternatively, if all cells of $\ttarray$ appeared to
store $\NULL$ during the loop, the algorithm restarts.

\subsection{Concrete and auxiliary state}

We assume that $\eventType$ consists of well-typed queue events $[\eid : (\tid,
\op, \args, \res)]$ meeting the following constraints:
\begin{itemize}
\item $\op \in \opType = \{ \enqOp, \deqOp \}$,
\item $\op = \deqOp \iff \args = \None$, and
\item $\res = \Done \implies \op = \enqOp$.
\end{itemize}

We consider a set of states $\stateType = \locType \to \valueType$, ranged over
by $\state$, where $\locType = \{ {\tt back} \} \cup \{ \ttarray[i] \mid i
\in \mathbb{N} \} \cup \{ {\tt k}[\tid] \mid \tid \in \tidType \} \cup ...$ is
the set of all memory locations including the global ${\tt back}$ and infinite
array $\ttarray[\-]$, as well as thread-local variables ($\ttk$, $\ttn$ etc).

We use a function $\ghostslot : \eidType \rightharpoonup \slotType$ as ghost
state in the proof in order to map event identifiers to slots in the infinite
array. The map is established with the help of auxiliary code in the atomic
block at line~\ref{line:hwq:getslot} in Figure~\ref{app:fig:hwq}.

For given $\ghostslot$ and $\events$, every enqueue event $e \in \ids{\events}$
can be one of the following:
\begin{itemize}
\item $e \notin \dom{\ghostslot}$ --- the slot is not assigned to the enqueue
yet,
\item $e \in {\sf withSlot}(\state, \events, \ghostslot)$ --- the enqueue has a
slot, but has not written a value into it yet:
\begin{multline*}
{\sf withSlot}(\state, \events, \ghostslot)
\triangleq
\{ \event \mid e \in \ids{\incomplete{\events}}
\land
\opof{\events(e)} = \enqOp
\\ {} \land
\state(\ttarray[\ghostslot(e)]) = \NULL
\}
\end{multline*}
\item $e \in \untaken{\state, \events, \ghostslot}$ --- the slot has writen a
value into its slot:
\begin{multline*}
\untaken{\state, \events, \ghostslot}
\triangleq
\{ \event \mid e \in \complete{\events}
\land
\opof{\events(e)} = \enqOp
\\ {} \land
\state(\ttarray[\ghostslot(e)]) = \argsof{\events(e)}
\}
\end{multline*}
\item $e \in \taken{\state, \events, \ghostslot}$ -- the value written into the
slot by the enqueue has been successfully taken by some dequeue event:
\begin{multline*}
\taken{\state, \events, \ghostslot}
\triangleq
\{ \event \mid e \in \complete{\events}
\land
\opof{\events(e)} = \enqOp
\\ {} \land
\state(\ttarray[\ghostslot(e)]) = \NULL
\}
\end{multline*}
\end{itemize}

\subsection{Commitment points}

To explain the construction of abstract histories for the Herlihy-Wing queue, we
instrument the code in Figure~\ref{app:fig:hwq} with auxiliary code. When an
operation starts, we automatically add a new uncompleted event into the set of
events $\events$ to represent this operation and order it in $\order$ after all
completed events. Aside from that, the {\tt enqueue} operation has two more
commitment points. For the first, the auxiliary code in the atomic block at
line~\ref{line:hwq:getslot} maintains the ghost state $\ghostslot$. For the
second, the auxiliary code at line~\ref{line:hwq:insert} completes the enqueue
event.

Upon a {\tt dequeue}'s start, we similarly add an event representing it, and
then the operation does one of the two commitment points. At
line~\ref{line:hwq:remove}, the current {\tt dequeue} operation encounters a
non-$\NULL$ value in a slot $\ttarray[\ttk]$, in which case it returns this
value and removes it from the array. The auxiliary code accompanying this change
to the state completes the {\tt dequeue} event and also adds three following
kinds of edges to $\order$ and then transitively closes it:
\begin{enumerate}
\item $({\tt ENQ}, \myEid)$, ensuring that in all linearizations of the abstract
  history, the current dequeue returns a value that has been already inserted by
  ${\tt ENQ} = \getEvent(\events, \ghost, \ttk)$.
\item $({\tt ENQ}, e)$ for each identifier $e$ of an enqueue event whose value
  is still in the pools. This ensures that the {\tt dequeue} removes the oldest
  value in the queue.
\item $(\myEid, d)$ for each identifier $d$ of an uncompleted dequeue event.
  This ensures that dequeues occur in the same order as they remove values from
  the queue.
\end{enumerate}

\subsection{The overview of proof details}

In Figure~\ref{app:fig:hwq}, we provide the
proof outlines for the enqueue and dequeue operations, in which we prove the
following specifications:
\[
\spec{\rely_\tid, \guar_\tid}
      {\tid}
      {\begin{array}{@{}c@{}}
      \inv \land \started(\tid, \deqOp)
      \end{array}}
      {\deqOp}
      {\begin{array}{@{}c@{}}
      \inv \land \finished(\tid, \deqOp)
      \end{array}}
\]
\[
\spec{\rely_\tid, \guar_\tid}
      {\tid}
      {\begin{array}{@{}c@{}}
      \inv \land \started(\tid, \enqOp)
      \end{array}}
      {\enqOp}
      {\begin{array}{@{}c@{}}
      \inv \land \finished(\tid, \enqOp)
      \end{array}}
\]

In the proof outlines, we use an auxiliary assertion describing an enqueue event
that has obtained a slot in the array, but has not written into it yet.
\begin{multline*}
\evalf{{\sf hasSlot}(\tid, \op)}{\lint} = \{
  (\state, \history, \ghost) \mid 
    \events({\sf last}(\tid, \history)) = 
            (\tid, \op, \state(\larg[\tid]), \Todo)
    \\ \hfill {}
    \land {\sf last}(\tid, \history) \in \dom{\ghost}
\};
\end{multline*}

For each thread $\tid$, we generate rely and guarantee relations analogously to
\S\ref{sec:details}. To this end, we let ${\tt getSlot}$, ${\tt insert}$ and
${\tt remove}$ 
denote atomic steps corresponding to atomic blocks extended with ghost code in
Figure~\ref{app:fig:hwq} (at lines \ref{line:hwq:getslot},
\ref{line:hwq:insert} and \ref{line:hwq:remove})
accordingly). For each thread $\tid$, relations $\guar_\tid$ and $\rely_\tid$
are then defined as follows:

\[
\begin{array}{rcl}
\guar_\tid & \triangleq &
\guar_{\tid,{\tt getSlot},\inv\land\started(\tid,\enqOp)} \cup
\guar_{\tid,{\tt insert},\inv\land{\sf hasSlot}(\tid,\enqOp)} \cup
\guar_{\tid,{\tt remove},\inv} \\ & & \hfill {} \cup
\guar_{\tid, {\tt local}},
\\
\rely_\tid & \triangleq &
\cup_{\tid' \in \tidType \setminus \{\tid\}}
(\guar_{\tid'} \cup {\dashrightarrow}_{\tid'})
\end{array}
\]
In the above, we assume a relation $\guar_{\tid, {\tt local}}$, which describes
arbitrary changes to certain program variables and no changes to the abstract
history and the ghost state. That is, we say that $\ttback$ and $\ttarray$
are {\em shared} program variables in the algorithm, and all others are {\tt
  thread-local}, in the sense that every thread has its own copy of them. We let
$\guar_{\tid, {\tt local}}$ denote every possible change to {\em thread-local}
variables of a thread $\tid$ only.

In Figure~\ref{fig:hwq_inv} we present the invariant $\inv$. It consists of several properties:
\begin{itemize}
\item $\invlin$ -- the main correctness property;
\item $\invord$ -- properties of uncompleted events that hold by construction of the partial order;
\item $\invalg$ -- a property of the array slots;
\item $\invwf$ -- well-formedness of ghost state.
\end{itemize}

\begin{figure}[t]
\hfill
\begin{minipage}[t]{0.91\textwidth}
\begin{enumerate}
\item[$(\invlin)$] all linearizations of completed events of the abstract
history satisfy the queue specification:
\vspace{-5pt}
\[
\abs(\history, \histories_{\sf queue})
\]
\item[$(\invord)$] completed dequeues precede uncompleted ones:
\vspace{-5pt}
\[
\forall d, d' \ldotp d \in \ids{\complete{\events}} \land
   d' \notin \ids{\complete{\events}} \land \opof{d} = \opof{d'} = \deqOp \implies \edge{d}{\order}{d'}
\]
\item[$(\inv_{\sf ALG})$] the order on untaken enqueue events does not
contradict the order in which they appear in the array:
\[
\forall e_1, e_2 \in \untaken{\state, \events, \ghost} \ldotp
  \edge{e_1}{\order}{e_2} \implies
  \ghostslot(e_1) < \ghostslot(e_2)
\]
\item[$(\invwf)$] well-formedness properties of ghost state that
 enumerate all possible combinations of states, ghost states and events in a
 history:
\begin{enumerate}[leftmargin=-8pt]
\item $\ghostslot$ maps some of the events from $\events$ to slots preceding $\tt
back$:
\[
\forall e \in \dom{\ghostslot} \ldotp e \in \ids{\events} \land
\ghostslot(e) < \state({\tt back})
\]
\item $\ghostslot$ is injective;
\item if $\state(\ttarray[k]) \neq \NULL$ then $k$ is a slot corresponding to an
untaken completed enqueue event:
\[
\forall k, v \ldotp \state(\ttarray)[k] \neq \NULL \iff 
  \exists e \in \untaken{\state, \events, \ghostslot} \ldotp \ghostslot(e) = k
\]
\item if $\state(\ttarray[k]) = \NULL$ then $k$ is a slot behind the $\tt
back$ of the array, or it has not been used yet, or it is assigned to an
uncompleted enqueue, or it has been inserted into and taken already:
\[
  \forall k \ldotp \state(\ttarray)[k] = \NULL \iff
    \begin{array}[t]{@{}l@{}}
    \state({\tt back}) \leq k \lor
    k \notin \dom{\ghostslot^{-1}} \lor{}\\
    (\exists e \in \taken{\state, \events, \ghostslot} \ldotp \ghostslot(e) = k)
  \lor{} \\
    (\exists e \in {\sf withSlot}(\events, \ghostslot) \ldotp \ghostslot(e) = k)
    \end{array}
\]
\end{enumerate}
\end{enumerate}
\end{minipage}
\caption{The invariant $\inv = \invlin \land \invord \land \invwf \land \inv_{\sf ALG}$}
\label{fig:hwq_inv}
\end{figure}

\subsection{Loop invariant}

We define a loop invariant $\linv$, which we use to ensure that the uncompleted
dequeue of a thread $\tid$ returns a correct return value (the value
inserted by the $\order$-minimal enqueue).

\begin{definition}
  Given interpretation of logical variables $\lint$, we let $\linv$ be an
  assertion denoting the set of configurations $\evalf{\linv}{\lint}$ such that
  every configuration $(\state, (\events, \order), \ghost)$ in it satisfies the
  following:
\begin{enumerate}
\item $\forall e, e' \in \untaken{\events, \ghost} \ldotp
\ghostslot(e) < \ttk \leq \ghostslot(e') \leq \ttn
\implies \nedge{e}{\order}{e'}$;
\item $\forall e \in \untaken{\events, \ghost} \ldotp
\ghostslot(e) < \state(\ttk)
\implies \nedge{e}{\order}{\myEid};$
\item $\state(\ttn) \leq \state(\ttback)$.
\end{enumerate}
\end{definition}
The loop invariant $\linv$ consists of three properties, which are formulated
w.r.t. the thread-local memory cells $\ttk$ (contains the current loop index),
$\ttn$ (contains the loop boundary), $\ttback$ and events of the abstract
history. The first property states that an enqueue event $e$ of a value in each
slot preceding the current ($\ghostslot(e) < k$) does not precede in $\order$ an
enqueue event $e'$ of a value from the subsequent part of the array. The second
property similarly requires that an enqueue event $e$ of a value in each slot
that has already been visited ($\ghostslot(e) < \state(k)$) 
does not precede the current dequeue event $\myEid$. Finally, the third property
simply asserts that the value in $\ttn$ is smaller than $\ttback$.

The following lemma justifies the history update by the atomic step ${\tt
remove}$.
\begin{lemma}\label{app:lem:hwqmin}
For every $\lint : \lvarsType \to \valueType$ and configuration $(\state,
(\events, \order), \ghostts) \in \evalf{\linv}{\lint}$, if
$\state(\ttarray[\ttk]) \neq \NULL$, then ${\tt ENQ} = \getEvent(\events,
\ghost, \ttk)$ is minimal among untaken enqueue events:
\[
\forall e \in \untaken{\state, \history, \ghost} \ldotp
\nedge{e}{\order}{{\tt ENQ}}
\]
\end{lemma}
\begin{proof}
  The statement of the lemma follows from the first property of the loop
  invariant, $\invalg$ and $\invwf$. According to the latter, every untaken
  enqueue $e$ has a value in a slot $\ghostslot(e) < \state(\ttback)$.

  When $\ttarray[\ttk] \neq \NULL$, it is easy to see that all untaken enqueues
  with slots later than $k$ in the array cannot precede ${\tt ENQ}$ according to
  $\invalg$, and the loop invariant asserts that all untaken enqueues before $k$
  in the array do not precede $\tt ENQ$ either.  Thus, $\tt ENQ$ is a minimal
  untaken enqueue.  \qed
\end{proof}

With the help of Lemma~\ref{app:lem:hwqmin}, we can conclude that the history
update of the atomic step ${\tt remove}$ at line~\ref{line:hwq:remove} in the
dequeue operation does not invalidate acyclicity of the partial order. Let $\ttarray[\ttk] \neq \NULL$ hold and let ${\tt ENQ} = \getEvent(\events, \ghost, \ttk)$ be an identifier of an enqueue event whose value is being removed. We consider separately each kind of
edges added into the abstract history:
\begin{enumerate}
\item \textbf{The case of $({\tt ENQ}, \myEid)$}. Note that prior to the
commitment point, $\myEid$ is an uncompleted event. By
Definition~\ref{def:history} of the abstract history, the partial order on its
events is transitive, and uncompleted events do not precede other events. Thus,
ordering ${\tt ENQ}$ before $\myEid$ does not create a cycle.
\item \textbf{The case of $(\myEid, d)$ for each identifier $d$ of an
uncompleted dequeue event}. Analogously to the previous case, if $d$ is
uncompleted event, it does not precede other events in the abstract history.
Hence, ordering $\myEid$ in front of all such dequeue events does not create
cycles.
\item \textbf{The case of $({\tt ENQ}, e)$ for each $e \in \untaken{\state, \history, \ghost}$}. By Lemma~\ref{app:lem:hwqmin}, from
$\linv$ it follows that no $e \in \untaken{\state, \history, \ghost}$
precedes ${\tt ENQ}$ in the abstract history. Consequently, ordering ${\tt ENQ}$
before all such enqueue events does not create cycles.
\end{enumerate}

\artem{Correctness of the emptiness check is commented out from the loop
invariant. The argument was flawed, and the algorithm in fact is not
linearizable with the emptiness check.}

\artem{Explanations of the loop invariant in hindsight are commented out.}


\artem{Preservation of the loop invariant is commented out}